%% file: operator-cumulant-hierarchy-paper-with-abstract.tex
\documentclass[11pt,a4paper]{article}
\usepackage{amssymb,amsmath, amsfonts}
\usepackage{graphicx,graphics}
\usepackage{mathtools}
\usepackage{bbold}
\usepackage[english]{babel}
\usepackage[utf8]{inputenc}
\usepackage{epsfig,url}
\usepackage{bbm,theorem}
\usepackage{a4wide}
\usepackage{xcolor}
\usepackage{enumerate}
\usepackage{calrsfs}
\usepackage[bookmarks=true]{hyperref}
\usepackage{bookmark}
\usepackage{verbatim}

\usepackage{stmaryrd}

\DeclareMathAlphabet{\pazocal}{OMS}{zplm}{m}{n}
\newtheorem{theorem}{Theorem}[section]

\newtheorem{lemma}[theorem]{Lemma}
\newtheorem{proposition}[theorem]{Proposition}
\newtheorem{assumption}[theorem]{Assumption}

\newtheorem{remark}[theorem]{Remark}

\numberwithin{equation}{section}
\numberwithin{theorem}{section}
\newcommand{\qed}{\hfill$\Box$}
\newenvironment{proof}{\begin{trivlist}\item[]{\em Proof:}\/}{\qed\end{trivlist}}

\usepackage{enumitem}

\newcommand{\lf}{\left}
\newcommand{\ri}{\right}

\newcommand{\mean}[1]{\lf\langle #1 \ri\rangle}

\newcommand{\innerprod}[2]{\langle #1,#2 \rangle}

\newcommand{\FT}[1]{\widehat{#1}}

\DeclareFontFamily{U}{mathx}{\hyphenchar\font45}
\DeclareFontShape{U}{mathx}{m}{n}{
      <5> <6> <7> <8> <9> <10>
      <10.95> <12> <14.4> <17.28> <20.74> <24.88>
      mathx10
      }{}
\DeclareSymbolFont{mathx}{U}{mathx}{m}{n}
\DeclareFontSubstitution{U}{mathx}{m}{n}
\DeclareMathAccent{\widecheck}{0}{mathx}{"71}
\DeclareMathAccent{\wideparen}{0}{mathx}{"75}

\newcommand{\defem}[1]{{\em #1\/}}
\newcommand{\rmd}{{\rm d}}
\newcommand{\rme}{{\rm e}}
\newcommand{\ci}{{\rm i}}

\newcommand{\T} {\mathbb{T}}

\newcommand{\vspan} {\mathrm{span}}

\newcommand{\ves} {{V_\mathrm{ES}}}
\newcommand{\es}[2]{ {#1 \llbracket#2\rrbracket} }
\newcommand{\ES}{ {\mathrm{ES}} }

\newcommand{\R}{{\mathbb R}}                                                    
\newcommand{\Z}{{\mathbb Z}}                                                    
\newcommand{\E}{{\mathbb E}}                                                    
\newcommand{\C}{{\mathbb C}}                                                    
\newcommand{\h}{{\mathfrak{h}}}

\newcommand{\N}{{\mathbb N}}                                                    
                                                   
\newcommand{\F}{{\mathcal F}}                                                   
\newcommand{\A}{{\mathcal A}}                                                   
\newcommand{\D}{{\mathcal D}}                                                   
\newcommand{\I}{{\mathcal I}}                                                   
                                                   
\newcommand{\B}{{\mathcal B}}                                                   
\newcommand{\LL}{{\mathcal L}}                                                   
\newcommand{\p}{{\mathcal P}}                                                   
\newcommand{\tp}{{\widetilde{P}}}                                                   
\newcommand{\hh}{{\mathcal H}}

\newcommand{\pr}{\mathrm{pr}}

\newcommand{\norm}[1]{\|#1\|}
\newcommand{\bes}{A_{\text{b}}} %% base even set

\newcommand{\xwick}[2]{{:\!#1\!:_{#2}}}
\newcommand{\bwick}[1]{{\xwick{#1}{\mathrm{s}}}}
\newcommand{\fwick}[1]{{\xwick{#1}{\mathrm{as}}}}

\DeclareMathOperator{\tr}{Tr} 
\DeclareMathOperator{\sgn}{sgn} 

\newcommand{\ind}{\mathbb{1}}

\newcommand{\lattice}{ \Lambda }

\newcommand{\duallattice}{ {\lattice^{\!*}} }

\newcommand{\bcum}{\kappa^{\mathrm{s}}}
\newcommand{\fcum}{\kappa^{\mathrm{as}}}

\newcommand{\ccol}{{\mathcal{C}^{\lambda}}}                                                    

\newcommand{\fta}{\FT{a}}

\newcommand{\id}{{\mathbf{1}}}
\newcommand{\ddt}{{\frac{\rmd}{\rmd t}}}
\newcommand{\dds}{{\frac{\rmd}{\rmd s}}}

\allowdisplaybreaks

\newcommand{\email}[1]{E-mail: {\tt #1}}
\newcommand{\emailjani}{\email{jani.lukkarinen@helsinki.fi}}
\newcommand{\emailsakari}{\email{sakari.pirnes@helsinki.fi} (corresponding author)}
\newcommand{\UHaddress}{\em University of Helsinki, Department of Mathematics and Statistics\\
\em P.O. Box 68, FI-00014 Helsingin yliopisto, Finland}

\begin{document}

\title{Symmetric and antisymmetric cumulant hierarchies and\\ {Wick} polynomials of operators}

\author{Jani Lukkarinen\thanks{\emailjani}\\[0.5em] Sakari Pirnes\thanks{\emailsakari}\\[1em]%
$\,^*,\,^\dag$\UHaddress}

\date{\today}

\maketitle
\begin{abstract}
There has been much progress in 
understanding mathematical foundations of kinetic theory in recent years, in
particular, for the original model of a rarefied gas of  classical particles
with repulsive interactions, as well as for weakly nonlinear wave evolutions
related to wave turbulence.  These results are mainly based on refined analysis
of evolution hierarchies of moments of the relevant random variables, or of
related quantities such as of cumulants or of correlation errors.  In this
work, we extend two of such basic tools, cumulants and Wick polynomials, to a
setting where the random variables no longer need to commute, with the goal of
handling expectations of creation and annihilation operators on a Fock space,
both in the  bosonic and fermionic case.  We call these constructions symmetric
and antisymmetric cumulants and Wick polynomials, respectively.  As an
application, we derive a hierarchy of evolution equations satisfied by such
cumulants, usually called truncated correlation functions in this context, for
a gas of bosons or fermions hopping on a finite lattice and interacting via a
two-body potential.
With sufficient control over solutions to these hierarchies, one could control
the magnitude of the error terms related to an approximation of the state of
the quantum system by a finite lattice kinetic equation studied in our earlier
work, closely connected to standard bosonic and fermionic Boltzmann--Nordheim
equations.  
\end{abstract}

\bigskip


\bigskip


\newpage

%\bigskip

\tableofcontents
\input{section1}
\input{section2}
\input{section3}

\input{section4}

\input{section5}

\input{section6}

\bigskip

\noindent \textbf{Acknowledgements}

The authors thank Peter Madsen, Phan Thành Nam, Chiara Saffirio, Manfred Salmhofer, Herbert Spohn,  and   Minh Binh Tran
%Gigliola Staffilani, and Minh Binh Tran 
for helpful discussions relevant to the project.  The work has been
supported by the Academy of Finland, via the {\em Finnish Centre of Excellence in Randomness and
Structures\/} (project Nos. 346306 and 364213).
 
The funders had no role in study design, analysis, decision to publish, or
preparation of the manuscript.

\smallskip

\noindent\textbf{Compliance with ethical standards} 
\smallskip

\noindent \textbf{Data availability statement} Data sharing not applicable to this article as no datasets
were generated or analysed during the current study.
\smallskip

\noindent \textbf{Conflict of interest} The authors declare that they have no conflict of interest.

%\input{appendix}

%\nocite{*} % Print all references regardless of whether they were cited in the poster or not
\bibliographystyle{abbrv} 
\bibliography{bibliography} % Use the  file bibliography.bib

\end{document}

%% file: section1.tex
\section{Introduction}

Mathematical foundations of kinetic theory of locally interacting classical
particles and of weakly interacting wave evolutions have seen significant focus
and advances over the last few years.  For example, the breakthrough results
of Deng, Hani, and Ma \cite{deng_full_2023, deng_long_2025}, 
were quoted as one of the
motivations for Yu Deng's Fields Medal in 2026.
These works have clarified the complex
mechanisms that lie behind cancellations leading to controllability of
accumulation of statistical correlations in these models.  They 
rely on highly refined analysis of perturbation expansions
of moments of the relevant time-evolving random variables, using tools such as
regularized correlations functions and cluster expansions.
These techniques have a rich history 
starting from the pioneering work of Landford in 1973 \cite{lanford73} for a system of interacting classical particles.  Recent developments on the topic are summarized in the proceedings article by Bodineau, Gallagher, Saint-Raymond and Simonella \cite{Bodineau2023}.

In contrast, advances in quantum kinetic theory for weakly interacting system of bosons or fermions have been more limited.  From a technical point of view, the question of controlling correlations in these systems is quite similar to controlling the evolution hierarchy related to the nonlinear Schr\"{o}dinger equation (NLS), as detailed for example in \cite{lukkarinen_not_2009}.  Since the accuracy of kinetic evolution equations can now be controlled for NLS with sufficiently well-prepared initial data \cite{buckmaster_onset_2021,DengHani_fullrange_2023,deng_full_2023,deng_long_2024}, with further techniques available if stochastic forcing is included \cite{Dymov2021,Dymov2022,Hannani_stoch_KdV_2022,dymov_kinetic_2025},
it seems plausible that also the quantum kinetic equation should be accurate in a similar setup.
One important distinction is in that the correlation functions in the quantum case are not associated with probability measures and random variables but instead come from certain linear functionals, called \defem{states}, over a non-commutative algebra of relevant observables.
Hence, many tools from probability theory are not directly applicable in the quantum case.

Examples of such tools are given by \defem{cumulants}, \defem{Wick polynomials}, and similar constructions, including the \defem{correlation errors} of the marginals of probability distributions introduced by Pulvirenti and Simonella \cite{Pulvirenti2016}.  Cumulants can be thought of as regularized moments,
and both carry the same amount of information about the system: knowing all
moments up to certain order allows determining all cumulants up to the same
order, and vice versa \cite{lukkarinen_wick_2016}.  The reason why cumulants appear as a more convenient choice
to control accuracy of kinetic theory is the main assumption needed in its
derivation: some form of sufficient statistical
independence is used in the derivations, called \defem{chaos} in this context.
In addition to independence, cumulants can also be used to measure how close the distribution of the random variables is to
a Gaussian distribution.
The main results up to day concern \defem{propagation of chaos} \cite{kac_foundations_1956,paul_size_2019,deng_propagation_2024} which corresponds to 
checking that, if the state of the system is initially sufficiently chaotic and well-prepared, it will
remain so at least up to times which are needed to see nontrivial
evolution connected to kinetic theory.   Given sufficiently strong stochastic forcing, further properties of chaos can sometimes be derived. 
In particular, cumulant hierarchy of particle energies in the stochastic Kac model has recently been used to control also \defem{generation of chaos} by the stochastic collisions in this mean field model \cite{lukkarinen_generation_2025}.  Cumulant hierarchies have also been used in \cite{dymov_kinetic_2025} for a stochastically forced NLS.

In this work, we focus on extending the main tools used in the above works in a
non-commutative setting which is required in the study of quantum kinetic
theory.  The concept of cumulants has already a counterpart in the quantum mechanical setting, called 
\defem{truncated correlation functions} \cite{brattelirobinson2}.  
Our aim is to generalize these definitions and introduce symmetric and antisymmetric cumulants for any expectation map over an algebra, as well as the corresponding Wick polynomials which serve as convenient regularization of monomials in the algebra.  In particular, we show how the Wick polynomials allow for a straightforward derivation of an evolution hierarchy of both symmetric and antisymmetric cumulants, assuming merely that the cumulants are composed of elements of the algebra satisfying a closed evolution equation.  This result generalizes those derived for (commuting) random variables in \cite{lukkarinen_wick_2016}.

We also apply these techniques to derive the cumulant evolution hierarchy for two quantum systems: mutually interacting bosons or fermions hopping on a finite $d$-dimensional lattice.  We connect these cumulant hierarchies to our earlier work \cite{lukkarinen2026finitelatticekineticequations} in which we explain 
how the commonly used quantum kinetic equations, in this context called Boltzmann--Nordheim or Ueling--Uhlenbeck equations, may be
derived from a \defem{truncated  hierarchy}, obtained from a \defem{closure assumption} which amounts to ignoring certain higher order cumulants.  Here we give an explicit expression for the related truncation error, the difference between the true and the truncated cumulant hierarchy.  
Control of this truncation error would lead to control of the accuracy of the kinetic evolution equations also for these fermionic and bosonic quantum systems.

% The current work can hence be understood in the general context of the mathematical programme of determining the precise conditions and errors for application of kinetic evolution equations, such as general Boltzmann equations.

Much of the earlier work in connecting quantum correlation functions to quantum kinetic equations relies on evolution hierarchies of moments, see for example \cite{erdos_quantum_2004,benedetto_weak-coupling_2005,lukkarinen_not_2009,Frst2013} for the relation between these hierarchies and the kinetic equations.  There are also some more complete mathematical results using mean-field type scalings  in \cite{chen_emergence_2023,Crdenas2025}, corresponding to evolution at subkinetic times.   In a recent preprint \cite{madsen2026asymptoticbehaviorkinetictime}, a fermionic lattice counterpart of an earlier 
result  \cite{lukkarinen_weakly_2011} for time-correlations at an equilibrium state of a discrete NLS is proven. Fermionic truncated correlation functions, 
together with Grassmann integral techniques, have also been used by Salmhofer \cite{Salmhofer2009} to control their decay in space at an equilibrium state for such a fermionic system.  Here, we systematize and extend these tools to a more general non-commutative setting, including also Wick polynomial regularizations.
%Non-commutative character
%becomes more prominent for Wick polynomials, as their definition will need a
%choice of a \defem{reference  order} for the arguments.
%In addition to cumulants, we also extend the definition of Wick polynomials to
%the general setup considered here.  We also show how these Wick polynomials are
%closely connected to evolution of the cumulants in weakly interacting quantum
%systems, extending the results from  [LMarcozzi].  

In fact, there already are cumulant constructions, called \defem{free cumulants} \cite{Speicher1994}, which are applicable in a general non-commutative setting, used for example in the analysis of properties of random matrices.
A general discussion about the associated \defem{free probability} can be found in
\cite{NicaSpeicher2006}.
However, it should be stressed that the cumulants we use here are \defem {not} the same as free
cumulants: even though all three types of  ``non-commuting cumulants'' 
can be defined for the systems we consider here, they
typically have different
numerical values and provide different type of control for the properties of
expectation values in the system.

What we call \defem{symmetric cumulants} is the most obvious generalization of standard cumulants to a non-commutative setting, but these do not share the main properties of free cumulants.
Their usefulness in the analysis involving bosonic creation and annihilation operators comes from the fact that these satisfy the \defem{canonical commutation relations} (CCR), and as such they are ``almost'' commuting.  In particular, the corresponding symmetric cumulants are then argument permutation invariant, apart possibly from order two --- see Proposition \ref{prop:cum_per_invariance}.
In addition, the symmetric cumulants can be used to check how close the expectation values match those coming from a bosonic quasi-free state.

The numerical values of \defem{antisymmetric cumulants} have little relation to those of either free or symmetric cumulants.  Instead, they are convenient for elements which satisfy 
\defem{canonical anticommutation relations} (CAR) such as fermionic creation and annihilation operators.
In this case, the corresponding antisymmetric cumulants are  argument permutation invariant, apart possibly from order two and up to a sign determined by the permutation.  The detailed result is given in  Proposition  \ref{prop:cum_per_invariance_F} below.
In addition, the symmetric cumulants can be used to check how close the expectation values match those coming from a fermionic quasi-free state.

The precise definition of symmetric cumulants and Wick polynomials is given in Section \ref{sec:sym}
and the antisymmetric cumulants and Wick polynomials in Section \ref{sec:anti}.
We will then also discuss in more detail their basic properties and relation to free cumulants.
In particular, to avoid misconceptions when defining and using the non-commuting objects, we develop a
somewhat heavy notational machinery in Section \ref{sec:notations} for indexing and manipulating linear combinations related to elements in the algebra.  
We also detail in Section\ \ref{sec:expectationmaps} what kind of expectation maps we consider here.  
Although any state can be used to give an expectation map, we do not require positivity of expectations maps here, only a normalization condition is needed. 

The last two sections contain our applications of these tools to time-evolution problems.
In Section\ \ref{sec:time-evolution}, we explain how the results from Sections  \ref{sec:sym} and  \ref{sec:anti} lead to evolution hierarchies of cumulants, given in Theorems \ref{thm:cumhieR}, \ref{thm:cumhie} and \ref{thm:cumhie_F}.  The main difference between these results comes from the assumed form of the rate function of the dynamical system determining the evolution of the related algebra elements.  In the first result, we assume that the rate function is a polynomial, while in the last two results it is explained how using regularized Wick polynomials in the rate function leads to fewer terms in the cumulant hierarchies.

Section \ref{sec:hie_bos_fer} specializes these methods in the two lattice systems of quantum particles discussed above.  We detail some natural assumptions for both the fermionic 
and the bosonic case which allow direct application of the previous results.  We state explicitly the antisymmetric cumulant hierarchy for fermions in \ref{sec:fermions} and the symmetric cumulant hierarchy for bosons in Section\ \ref{sec:bosons}. 
The connection of these two hierarchies to the truncated cumulant hierarchy of \cite{lukkarinen2026finitelatticekineticequations} and to quantum kinetic theory is discussed in Section\ \ref{sec:derivation}.

%\subsection{Summary of results}

%(ADD A SUMMARY?)

%% file: section2.tex
\section{Notations}
\label{sec:notations}

\subsection{Sequence notations}
\label{sec:sequence}

Given an index set $J$, i.e.\ any non-empty set, we define an {\em index sequence} to be any
finite set $I\subset\N\times J$ satisfying that if $(\ell,j)\in I$ then
$(\ell,i)\notin I$ for all $i\in J$, $i\ne j$. Thus such index sequences are in
one-to-one correspondence with graphs of maps $f:N\to J$, where $N\subset \N$
is any finite subset.   In particular, $\emptyset$ is an (empty) index sequence.
Let $\I(J)$ denote the
collection of all index sequences constructed from $J$.   Using index sequences,
we can now systematically handle combinatorics of subsequences  and permutations.  
For example, with this notation, we can use the standard subset notation $E\subset I$ to construct
a subsequence $E$ of the index sequence $I\in\I(J)$ so that the elements
$(\ell,j)\in E$ will automatically remember their position in the original
index sequence $I$. This will be important as we are interested in labeling operators 
that do not necessarily commute.

Given an index sequence $I$, the corresponding \defem{collection of position
labels} is given by $\pr_1(I)\subset \N$.  Here and in the following, the
notation $\pr_n$ denotes the map projecting the $n$:th component from an
element of a product space.  For example, the collection of indices in $I$ is
given by $\pr_2(I)= \{j\in J\, |\,(\ell,j)\in I\}$.  Note that although
$|\pr_1(I)|=|I|$ we only have $|\pr_2(I)|\le |I|$ since repetitions of an
index are allowed.

If $I,I'\in\mathcal{I}(J)$, then $I\cup I'$, $I\cap I'$ and $I\setminus I'$ denote the 
basic set operations, namely, union, intersection and set difference.   In addition,
if $I \in \I(J)$ and $\ell\in \pr_1(I)$, 
there is a unique  $j\in J$ for which $(\ell,j)\in I$, and we sometimes use the shorthand notation ``$\ell$'' to refer to
$(\ell,j)$.  In particular, if $E\subset I$, we may then denote $E\cup \ell:=E\cup\{(\ell,j)\}$, $E\cap
\ell:=E\cap\{(\ell,j)\}$ and $E\setminus \ell:=E\setminus \{(\ell,j)\}$.
Given $I\in\I(J)$ and $m\in\N$ we further define 
\begin{align}
    \label{eq:jani_notation}
    I|_{<m} = \{ (\ell,j)\in I \ : \ \ell<m\}
    \,,
    \qquad
    I|_{>m} = \{ (\ell,j)\in I \ : \ \ell>m\}
    \,.
\end{align}
Both $I|_{<m}$ and $I|_{>m}$ are subsets of $I$ and hence belong to $\I(J)$.

Note that $I\cup I'\notin\I(J)$ for $I,I'\in\mathcal{I}(J)$
if there is $(\ell,j)\in I$ such that $(\ell,j')\in I'$
for some $j'\ne j$.
This motivates the following definition of $I+I'$ that 
glues $I$ and $I'$ together in a way that results in an index sequence. For $I\in\I(J)$
and $n\in\N_0$ we first define a shift $I+n:=\{(\ell+n,j_\ell)\colon (\ell,j_\ell)\in I\}$
and set $n+I:=I+n$. Given 
$I,I'\in\I(J)$, we define the gluing of $I$ and $I'$ by 
\begin{align}
    \label{eq:glueing}
    I+I':=
    \begin{cases}
    I\cup (I'+\max_{(\ell,j_\ell)\in
    I}\ell), & I\ne\emptyset\\
    I', & I=\emptyset
    \end{cases}.
\end{align}
Therefore, $\emptyset+I=I+\emptyset=I$ for any $I$.  We also always have $|I+I'|=|I|+|I'|$ and
$I+I'\in\I(J)$.
Note that this gluing process 
is associative but not commutative.  The results from the sum in the notation
$$
\sum_{i=1}^{n} I_i= I_1+I_2+\dots+I_n
$$
can thus be defined for $I_i\in\I(J)$, $i=1,2,\ldots,n$, independently of the order in which these sums are computed.

\subsection{Polynomials and permutations}

We will use $\A$ to denote a unital algebra over  
the complex numbers $\C$. Let $\id$ denote the identity element of the algebra. 
Let $J$ be any index set and $b\colon J\to\A$ any map.  We introduce the map $b$ mainly to make it clear when the algebraic rules may and may not be applied to simplify the results from algebraic operations.

We define $b^I\in\A$ to be a product of elements corresponding to the labels in $I$.
To be more precise, we define $b^\emptyset:=\id$, and for a nonempty index sequence $I\in\I(J)$, $I\ne\emptyset$,
we define
\begin{align}
    \label{eq:prod}
    b^I:= \prod_{(\ell,j_\ell)\in I}^{}b(j_\ell),
\end{align}
where the elements $b(j_\ell)$ are ordered in the product so that the corresponding labels $\ell$ increase from left to right. For example, if $i,j,k\in J$ and
$I=\{(3,j),(8,i),(9,j),(12,j),(27,k)\}$, we have $b^I=b(j)b(i)b(j)b(j)b(k)$. Note
that the elements might not commute, and hence the order is important. 
We recall the definition in \eqref{eq:glueing} and note that now there always holds
\[
b^Ib^{I'}=b^{I+I'}\,, \qquad I,I'\in\I(J) \,.
\]

Unless stated otherwise, a \defem{label permutation} $P$ here refers to a bijection from $\N$ to itself.
The shorthand notation 
\begin{align}
    \label{eq:indexpermutation}
    P[I]:=\{(P(\ell),j )\colon (\ell,j )\in I\}
\end{align}
will be used for index sequences $I\in\I(J)$. Note that $P[I]\in \I(J)$ for $I\in \I(J)$.
We will mainly be interested in taking permutations of polynomials, i.e., permute the order of the
operators in each monomial. We note that the
shorthand notation $P(b^E)=b^{P[E]}$ would be ill-defined for permutations $P$. To see why,
consider 
$I=\{(1,i),(2,j),(3,i),(4,j)\}$ and a permutation $P$ that only swaps $3$ and $4$.  Then,
$b^{P[E_1]}=b(i)b(j)\ne b(j)b(i)=b^{P[E_2]}$ for $E_1=\{(1,i),(2,j)\}$ and
$E_2=I\setminus E_1$ while $b^{E_1}=b(i)b(j)=b^{E_2}$. 

Our solution to the above notational problem is to introduce
for each index sequence $I\in\I(J)$ the corresponding {\it element sequence} $\es{b}{I}$ defined by
\begin{align}
    \label{eq:opset}
    \es{b}{I}:=\{(\ell,b(j))\colon (\ell,j )\in I\} \subset \N \times \mathcal{A}.
\end{align}
For notational convenience, we have used the double brackets ``$\es{}{\,}$'' to indicate that $b$ acts on the second component. Similarly, single brackets  ``$[\,]$'' are used when function acts on the first component as in \eqref{eq:indexpermutation}.

Let $\ES(J;b):=\{ \es{b}{I} \ : \ I\in\I(J)\}$ be the collection of all element 
sequences constructed from $b\colon J\to\A$.  In the following, we will drop the map $b$ from the notation $\ES(J;b)$, unless it is required to stress which $b$ is used in the construction.
Often, we first have to work with element sequences $\es{b}{I}$
 and only later convert them into corresponding monomials, as elements of the algebra. 
To this end, we introduce the abstract complex vector space freely spanned by the element sequences $\ES(J)$,
 \begin{align*}
     \ves(J):=  \vspan(\ES(J)) =
     \left\{ 
         c_1 \es{b}{I_1}
         +c_2 \es{b}{I_2}
         \dots
         +c_n \es{b}{I_n}
         \ : \ 
         n\in \N, \ c_i\in\C, \ \es{b}{I_i}\in\ES(J)
     \right\} . 
 \end{align*}
 The elements of $\ves(J)$ are called {\it abstract polynomials}, and 
 $\ES(J)$ is a linearly independent subset of $\ves(J)$.
 Note that as a set $\ES(J;b)=\I(b(J))$, but the meaning of the operation ``$+$'' is defined differently in $\ES(J)\subset \ves(J)$ and $\I(b(J))$.

Given a permutation $P$ of $\N$, we extend $P$ into a linear map acting on $\ves(J)$
by introducing the shorthand notation 
\begin{align}
    \label{eq:Pves}
    P\es{b}{I}:=\es{b}{P[I]} =
    \{(P(\ell),b(j) )\colon (\ell,j )\in I\}\in\ES(J), \qquad I\in\I(J)\,,
\end{align}
for element sequences $\es{b}{I}\in\ES(J)$.  
We define $\Xi$ as the map from $\ves(J)$ to the algebra that collapses the abstract polynomials to their 
corresponding polynomials, i.e.,
\begin{align}
    \label{eq:Xi}
    \Xi\sum_{i=1} ^n c_i \es{b}{I_i} := \sum_{i=1} ^n c_i b^{I_i}\,, \qquad c_i\in\C,\ I_i\in\I(J) \,.
\end{align}
In particular, if $I=\emptyset$, we have $\es{b}{\emptyset}=\emptyset$ and $\Xi(c\emptyset)=\Xi c\es{b}{\emptyset}=cb^\emptyset
= c\id \in \mathcal{A}$. Note that $\Xi$ is linear but never one-to-one since,
given some $j\in J$ and setting $I_1=\{(1,j)\}$, $I_2=\{(2,j)\}$, we have $\es{b}{I_1}-\es{b}{I_2}\ne 0 =\es{b}{I_1}-\es{b}{I_1}$ but $\Xi(\es{b}{I_1}-\es{b}{I_2})=0_{\mathcal{A}}=\Xi0$. We stress that the zero of the abstract polynomial space is a new element added in the construction and it has to be different from $\emptyset$.
%$\Xi \es{b}{I}= \id$ for $I=\{(1,j_1)\}$ with $j_1= b^{-1}(\id)$. 
With these definitions, it becomes possible to  apply permutations on the
abstract polynomials and then 
collapse them into
polynomials of the algebra.  Namely, we have
\begin{align}
    \label{eq:XiP}
    \Xi P\sum_{i}  c_i \es{b}{I_i} = 
    \Xi \sum_{i}  c_i \es{b}{P[I_i]} =\sum_{i} c_i b^{P[I_i]} 
\end{align}
for any permutation $P$ and $\sum_i c_i \es{b}{I_i} \in\ves(J)$.

We will need products of polynomials, and therefore we need to define a product in 
$\ves(J)$ that emulates the product of polynomials. This is done by setting 
\begin{align}
    \label{eq:opseq_prod}
    \es{b}{I} \es{b}{I'}:=\es{b}{I+I'}\,, \qquad I,I'\in\I(J) \, ,
\end{align}
and then extending by linearity.
Note that, since $I+\emptyset=\emptyset+I=I$, the identity element of this product is 
$\es{b}{\emptyset}=\emptyset$.
Now for $\sum_ic_i\es{b}{I_i},\ \sum_m c'_m \es{b}{I_m'}\in\ves(J)$ we have
\begin{align}
    \left( \sum_ic_i\es{b}{I_i} \right)  \left( \sum_m c'_m \es{b}{I_m'} \right) 
   =&
   \sum_{i,m} c_i c_m' \es{b}{I_i+I_m'},
   \\
   \Xi
   \left( \sum_ic_i\es{b}{I_i} \right)  \left( \sum_m c'_m \es{b}{I_m'} \right) 
   =&
   \sum_{i,m} c_i c_m' b^{I_i+I_m'}
   =
   \left( \sum_ic_ib^{I_i} \right)  \left( \sum_m c'_m b^{I_m'} \right) \,.
\end{align}
Hence, $\Xi$ is an algebra homomorphism from $\ves(J)$ to $\mathcal{A}$. Therefore, the product \eqref{eq:opseq_prod} preserves the order,
which is what we wanted.  However, $\Xi P$ is typically not a homomorphism unless $P$ is the identity permutation.

In addition, we want to define an interlacing  product, denoted by ``$\times$'',
 that allows mixing the elements unlike the product \eqref{eq:opseq_prod}. Given $\sum_ic_i\es{b}{I_i},\ \sum_m c'_m \es{b}{I_m'}\in\ves(J)$  we define the interlacing product $\times$ on $\ves(J)$ by 
\begin{align}
    \label{eq:prod_interlace}
    \left(
    \sum_ic_i\es{b}{I_i}
    \right) 
    \times 
    \left(
    \sum_m c'_m \es{b}{I_m'}
    \right)
    = 
    \sum_{i}
    \sum_{m}
    c_i
    c'_{m} 
    \begin{cases}
        \es{b}{I_i}
        \cup
        \es{b}{I'_m}
        \,, & \pr_1(I_i)\cap \pr_1(I'_m)=\emptyset\,,
        \\
        0\,, & \text{else}\,.
    \end{cases}
\end{align}
Here, condition $\pr_1(I_i)\cap \pr_1(I'_m)=\emptyset$ ensures that the right hand side of \eqref{eq:prod_interlace} belongs to $\ves(J)$. We note that in this paper we use the product $\times$ only for polynomials for which the condition $\pr_1(I_i)\cap \pr_1(I'_m)=\emptyset$ is satisfied for all $i$ and $m$.
The interlacing product $\times$ is linear in both arguments. It is also associative and commutative, since  both the union and the intersection have these properties. We note that $\Xi$ is not a homomorphism under 
the interlacing product $\times$.

To illustrate the difference between these two products, take $I=\{(1,j)\}$ and $I'=\{(2,j')\}$. Then,
\begin{align*}
    \Xi 
    (\es{b}{I'}
    \times
    \es{b}{I})
    =
    b(j)b(j')
   = \Xi (
    \es{b}{I}
    \es{b}{I'})
    \ne 
    b(j')b(j)
    =
    \Xi (
    \es{b}{I'}
    \es{b}{I} )
    \,.
\end{align*}

\subsection{Expectation map on unital algebra}
\label{sec:expectationmaps}

We will use $\E$ to denote an \defem{expectation map} on $\A$, namely, a linear map
\begin{align}
    \label{eq:expectation_map}
    \E\colon \A\to\C 
\end{align}
that is normalized in the
sense that $\E[\id]=1$.
If $\mathcal{A}$ is a $C^*$-algebra, any \defem{state} on it would define an expectation map.  However,
unlike states, we do not require any positivity properties from an expectation map.

We say that an expectation map $\E$ is \defem{even} with respect to a subset $\bes{\subset b(J)} \subset \mathcal{A}$ if
\begin{align}
    \label{eq:Eeven}
    \E[b^I]=0
\end{align}
for all $I\in\I(J)$ such that
$|I|$ is odd and $b(j)\in \bes$ for all $(\ell,j)\in I$.
Clearly, then $\id \not\in \bes$.
An even state in a fermionic Fock space gives an example of an even expectation map with respect to the collection of creation and annihilation operators.

%% file: section3.tex
\section{Symmetric cumulants and Wick polynomials}
\label{sec:sym}
\subsection{Symmetric cumulants}

Suppose we are given an expectation map $\E\colon\A\to\C$ on a complex unital algebra $\A$.
First, let us choose $J=\A$ and $b=\text{id} \colon J\to\A$.  Employing the notations from Section\ \ref{sec:notations}, we define {\it symmetric
cumulants}
$\bcum[\es{b}{I}]$ for  $I\in\I(J)$ by setting $\bcum[\emptyset]:=0$ and then
doing induction in $|I|\ge 1$ to the recursion relation 
\begin{align}
    \label{eq:m2k}
    \E[b^I] = \sum_{\pi\in\p(I)}^{} \prod_{A\in\pi}^{} \bcum[\es{b}{A}]\,,
\end{align}
where $\p(I)$ denotes the collection of all partitions of the set $I$.
Then for  any index set $J$ and any map $b\colon J\to \A$ the equation \eqref{eq:m2k} holds for all $I\in\I(J)$.  Here, $I=\emptyset$ is a special case discussed below.
Symmetric cumulants may thus be considered
as a map $\bcum\colon \I(\A)\to\C$.
For the rest of this Section \ref{sec:sym}, we assume that 
 $J$ is any index set and $b\colon J\to \A$ any map.
Then
$\bcum[\es{b}{\cdot}]$ becomes a map from $\I(J)$ to $\C$.
The motivation for a more general choice of $b$ and  $J$ becomes apparent in the applications in Section \ref{sec:hie_bos_fer}.

Equation \eqref{eq:m2k} is called the moments-to-cumulants formula.
This definition is also used in \cite{brattelirobinson2}, where the 
name \defem{bosonic truncated correlation functions} is used over cumulants.
In the nondegerate case, the moments to cumulants formula \eqref{eq:m2k} may be rewritten as
\begin{align}
    \bcum[\es{b}{I}] = \E[b^I] - \sum_{\substack{\pi\in\p(I)\\ |\pi|>1 }}^{} \prod_{A\in\pi}^{} \bcum[\es{b}{A}]\,, \qquad I\in \I(J)\,,\ I\ne \emptyset \,.
\end{align}
If $\pi\in\p(I)$,
we call $A\in\pi$ a \defem{cluster} in the partition $\pi$.
Since we insist that also for $I=\emptyset$
the formula \eqref{eq:m2k} holds with the standard definitions of vacuous
products and sums, we
should define then $\p(\emptyset):=\{\emptyset\}$,
since for $\pi=\emptyset$ the condition $A\in \pi$ is false and thus vacuously
$\prod_{A\in \pi}^{} \left(\cdot\right):= 1$.  Hence, with this definition the
right hand side of \eqref{eq:m2k} evaluates to $1= \E[b^\emptyset]$.

We will be also interested in permuted cumulants. Recall from \eqref{eq:indexpermutation} how permutations act on index sequences. Given a
permutation $P\colon\N\to\N$ and a index sequence $I\in\I(J)$ also
$P[I]$ is in $\I(J)$ and thus
\begin{align}
    \label{eq:m2kP}
    \E[b^{P[I]}] =
    \sum_{\pi_0\in\p(P[I])}^{}\prod_{A\in\pi_0}^{} \bcum[\es{b}{A}]
    =        \sum_{\pi\in\p(I)}^{} \prod_{A\in\pi}^{} \bcum[\es{b}{P[A]}]\,,
\end{align}
where in the second equality we have changed summation variables from
$\pi_0\in\p(P[I])$ to $\pi = \{P^{-1}[A]: A\in\pi_0\}$.  This clearly
results in a bijection between $\p(P[I])$ and $\p(I)$, the inverse given
by the map $P(\pi) := \{P[A]: A\in\pi\}$ for $\pi\in\p(I)$.  Hence, the
corresponding products of cumulants are mapped as stated above.

Whenever the underlying map $b\colon J\to\A$ is clear, we will use the shorthand notation 
\begin{align}
    \label{eq:cum_notation0}
    \bcum[\es{b}{I}]=\bcum[I]\,, \qquad I\in\I(J)\,.
\end{align}
It should be mentioned that $\bcum[I]$ notices only the relative order of the position 
labels in $I$, i.e., if $f:\N\to\N$ is strictly increasing, then 
\begin{align}
    \label{eq:cumorder}
    \bcum[I] = 
    \bcum[f[I]]\,,
\end{align}
for all $I\in\I(J)$, where $f[I] = \{ (f(\ell),j)\ : \ (\ell,j)\in I\}\in\I(J)$. This follows recursively from
\eqref{eq:m2k} and $b^I=b^{f[I]}$. Therefore, the following notations 
are justified, 
\begin{align}
    \label{eq:cum_notation}
    \bcum[b(j_1),  \dots, b({j_n})] 
    =
    \bcum[j_1,  \dots, {j_n}] 
    =
    \bcum[I]
    =
    \bcum[\es{b}{I}]
    ,
\end{align}
where $I=\{(\ell_1,{j_1}),(\ell_2,{j_2}),\dots, (\ell_n,{j_n})\}$ and $\ell_1<\ell_2<\dots<\ell_n$.
The first two cumulants are 
\begin{align*}
 &\bcum[b_1] = \E[b_1]\,, \qquad \bcum[b_1, b_2]
 =  \E[b_1 b_2]
-
\E[b_1]
\E[b_2]
 \,. 
\end{align*}

\begin{remark}
    Note that the symmetric cumulants \eqref{eq:m2k} are different from the free cumulants 
    used in free probability. In \eqref{eq:m2k}, the set $\p(I)$ contains all partitions of $I$, while in the 
    definition of free cumulants one uses only the non-crossing partitions.  For
    detailed definition and properties of free cumulants see
    \cite{Speicher1994,NicaSpeicher2006}.
\end{remark}

\subsection{Symmetric Wick polynomials}
We define Wick polynomials $\bwick{b^I}$  by setting $\bwick{b^\emptyset}:=\id$
and for $I\in\I(J)$, $I\ne\emptyset$, define them recursively by requiring
\begin{align}
    \label{eq:wick}
    \bwick{b^I}:=b^I-\sum_{\emptyset\ne E\subset I}^{} \E[b^E]\bwick{b^{I\setminus E}}.
\end{align}
The lowest order examples are
\begin{align*}
 &\bwick{b_1} = b_1-\E[b_1]\id\,, \qquad 
 \bwick{b_1 b_2}
 =  b_1 b_2  
 -
 \E[b_1]{b_2}
 -
 \E[b_2]{b_1}
 +
 2\E[b_1]\E[ b_2]\id  
 -
 \E[b_1 b_2]\id  
 \,.
\end{align*}
By induction, one can prove that $\E[\bwick{b^I}]=0$ if $I\ne \emptyset$.
Also, it should be stressed that the notation borrowed from
physics is somewhat misleading: the Wick polynomial $\bwick{b^I}$ depends on
the full sequence $I$, not only on the algebra element $b^I$.

As previously discussed, we need to be careful with permutations of polynomials.
Hence, for each $I\in\I(J)$ we define the abstract symmetric Wick polynomial
$\bwick{\es{b}{I}}\in\ves(\I(J))$ by setting $\bwick{\es{b}{\emptyset}}=\es{b}{\emptyset}=\emptyset$
and for $I\ne \emptyset $ recursively by 
\begin{align}
    \label{eq:wickseq}
    \bwick{\es{b}{I}}:= 
    \es{b}{I}-\sum_{\emptyset\ne E\subset I} \E[b^E]\bwick{\es{b}{I\setminus E}}\,.
\end{align}
It is straightforward to check that
\begin{align}
    \label{eq:wickseq_is_wick}
    \Xi\bwick{\es{b}{I}}=\bwick{b^I}
\end{align}
for all $I\in \I(J)$.

\begin{remark}
    \label{rem:samelenght_wick}
   An abstract Wick polynomial defined this way contains sums of element 
   sequences of different lengths. The definition \eqref{eq:wickseq} is sufficient for this paper,
   but we note that this can be inconvenient when multiplying with an abstract
   Wick polynomial from the left. The reason being that the
   different length sequences result in different sized shifts in the
   multiplication \eqref{eq:opseq_prod}  due to the definition
   \eqref{eq:glueing} of gluing sequences. 
   For example, if $I=\{(1,j)\}$ and $I'=\{(1,j')\}$ then 
   \begin{align}
       \nonumber
       \bwick{\es{b}{I}} \es{b}{I'} 
       &=
       \left( \es{b}{I}-\E[b(j)]\es{b}{\emptyset} \right) 
        \es{b}{I'} 
        % \\ \nonumber        &
       =
        \es{b}{I+I'}-\E[b(j)]\es{b}{I'} 
        \\
       \label{eq:example_wick}
       &=
       \{(1,j),(2,j')\}-\E[b(j)]\{(1,j')\}, 
   \end{align}
   and therefore the position label of $j'$ is $2$ in the first term and $1$ in the second term.

   For this reason, it may be convenient to make the sequences have the same
   length by adding the identity element $\id$ of the underlying algebra $\A$.
   If $\id \not\in b(J)$, we can always extend $J$ and $b$, by adding a new index $j_{\id}$ to $J$
   and setting $b(j_{\id}) =\id$.
   For $I\in\I(J)$ define $\es{\id}{I}=\{(\ell,\id)\ \colon \ (\ell,j)\in I\}$ and
   %$\id^I=\{(\ell,\id)\ \colon \ (\ell,j)\in I\}$ and
   \begin{align}
        \label{eq:samelenght_wick0}
       \bwick{\es{b}{\emptyset;I}}=\es{\id}{I}\, .
   \end{align}
   
   Let $I\in\I(J)$ and $U\subset I$, $U\neq\emptyset$. We define
   $\bwick{\es{b}{U;I}}$ recursively in $|I|$ and $|U|$ by
    \begin{align}
        \label{eq:samelenght_wick}
        \bwick{\es{b}{U;I}}:= 
        \es{b}{U}\cup \es{\id}{I\setminus U}-\sum_{\emptyset\ne E\subset U} \E[b^E]
        (\bwick{\es{b}{U\setminus E;U}}\cup\es{\id}{I\setminus U}).
    \end{align}
    For example, by \eqref{eq:samelenght_wick0} for $|U|=1$ we have
    \begin{align*}
        \bwick{\es{b}{U;I}}
        &= 
        \es{b}{U}\cup \es{\id}{I\setminus U}- \E[b^{U}]
        \bwick{\es{b}{\emptyset;U}}\cup\es{\id}{I\setminus U}
        \\
        &=
        \es{b}{U}\cup \es{\id}{I\setminus U}- \E[b^{U}]
        \es{\id}{I}
    \end{align*} 
    It is straightforward to check that every "monomial" sequence $\es{b}{I'}$ in the abstract polynomial 
    $\bwick{\es{b}{U;I}}$ satisfies $\pr_1(\es{b}{I'})=\pr_1(I)$. In addition, there 
    holds 
    \begin{align}
        \label{eq:samelenght_wickseq_is_wick}
        \Xi\bwick{\es{b}{U;I}}=\bwick{b^{U}}
    \end{align}
    whenever $U\subset I\in \I(J)$.
\end{remark}

\subsection{Properties of symmetric cumulants and Wick polynomials}

We start by proving the following useful expression for cumulants and moments 
\begin{proposition}
    \label{prop:expression}
    Take $I\in\I(J)$ and fix $(\ell,j)\in I$. 
    \begin{align}
        \label{eq:expressionE}
        \E[b^{I}] 
        &= 
        \sum_{\substack{E\subset I,\\(\ell,j)\in E}}
        \E[b^{I\setminus E}]\bcum[E]
    \end{align}
    or equivalently 
    \begin{align}
        \label{eq:expressionK}
        \bcum[I] 
        &= 
        \E[b^{I}] 
        -
        \sum_{\substack{E\subsetneq I,\\(\ell,j)\in E}}
        \E[b^{I\setminus E}]\bcum[E].
    \end{align}
\end{proposition}
\begin{proof}
    Moments to cumulants formula \eqref{eq:m2k} implies
    \begin{align*}
        \E[b^{I}] 
        &= 
        \sum_{\substack{E\subset I,\\(\ell,j)\in E}}
        \bcum[E]
        \left( 
            \sum_{\pi\in\p(I\setminus E)}
            \prod_{A\in\pi}
            \bcum[A]
        \right) .
    \end{align*}
    Using moments to cumulants \eqref{eq:m2k} again gives \eqref{eq:expressionE}. Then solving 
    $\bcum[I]$ from \eqref{eq:expressionE} implies \eqref{eq:expressionK}.
\end{proof}

Our first example is a version for symmetric cumulants of what happens with deterministic random variables, i.e., with random variables which are almost surely constant.
\begin{proposition}
    \label{prop:constant}
    Let $I\in\I(J)$, $|I|\ge2$. Suppose that there is $(\ell,j)\in I$ such that 
    $b(j)=c\id$ for some $c\in\C$. Then 
    \begin{align}
        \label{eq:constant}
        \bcum[I] = 0
        \, .
    \end{align}
\end{proposition}
\begin{proof}
    We will do induction in $|I|$. Suppose $|I|=2$.
    Since $b^{I} = cb^{I\setminus\ell}$, $\E[\id]=1$
    and $\E$ is linear, the explicit definitions of the lowest order terms yield
    \begin{align}
        \label{eq:some_constant_estim}
        \bcum[I]   = 
        \E[cb^{I\setminus\ell} ]
        -
        \E[b^{I\setminus\ell}] 
        \E[c\id]
        =0
        \,.
    \end{align}
    Therefore, \eqref{eq:constant} holds for $|I|=2$. Take $n\in\N$, $n\ge2$.
    Suppose that \eqref{eq:constant} holds for any $|I|\leq n$. Consider some $I$ with 
    $|I|=n+1$. Then the induction hypothesis 
    implies that $\bcum[E]=0$ in 
    \eqref{eq:expressionK} if $2\leq|E|\leq n$, and hence 
    \eqref{eq:some_constant_estim} holds also for this $I$.
    Thus, induction implies \eqref{eq:constant}. 
\end{proof}

\begin{proposition}
    The symmetric cumulants are multilinear in the sense that 
    \begin{align}
        \nonumber
        \bcum[b(j_1),\dots,  \alpha b(i)+\beta b(k) ,\dots, b({j_n})] 
        &=
         \alpha
        \bcum[b(j_1),\dots,   b(i) ,\dots, b({j_n})] 
        \\
        \label{eq:bcum_multilinear}
        &\quad+
         \beta 
        \bcum[b(j_1),\dots,  b(k) ,\dots, b({j_n})] 
    \end{align}
    for all $j_1,\dots,j_n,i,k\in J$ and $\beta,\alpha\in\C$.
    Similarly, the symmetric Wick polynomials are multilinear.
\end{proposition}
\begin{proof}
The multilinearity \eqref{eq:bcum_multilinear} follows recursively from
\eqref{eq:expressionK} together with the  multilinearity of $\E[b^I]$. Similarly we obtain multilinearity of the 
symmetric wick polynomials from \eqref{eq:wick}.
\end{proof}

For random variables, if a sequence contains one random variable which
  is independent from the rest, then the corresponding cumulant is zero.
  Although the meaning of ``independence'' is more difficult to capture in
  non-commutative setting, let us present below two factorization properties
  which are connected to vanishing of certain symmetric cumulants.  One
  special aspect about the results below is that they do not require full
  independence, it suffices that the cumulants are zero for subsequences of the
  reference sequence $I$.  The first case is motivated by properties of
  Gaussian distributions and the second case of independence.

\begin{proposition}
    \label{prop:independence}
     Fix $I\in\mathcal{I}(J)$ such that  $|I|\ge2$. Fix $(\ell,j_\ell)\in I$. 
     Then let 
     \[
         \mathcal{E}(n,I,\ell)=\{ E\subset I \colon |E|\ge n,\ (\ell,j_\ell)\in E\}.
     \]
     Then the following two conditions are equivalent:
     \begin{enumerate}
        \item[(1)]
        There holds
         \begin{align}
             \label{eq:E3_independence}
             \E[b^E]=&\E[b(j_\ell)] \E[b^{E\setminus \ell}] + 
             \sum_{\substack{p\subset E \\
             |p|=2,  (\ell,j_\ell)\in p}} \bcum[p]\E[b^{E\setminus p}] 
         \end{align}
         for all $E\in\mathcal{E}(3,I,\ell)$.
     \item[(1')] 
         There holds
         \begin{align}
             \label{eq:k3_independence}
             \bcum[E] =& 0 \quad 
         \end{align}
         for all $E\in\mathcal{E}(3,I,\ell)$.
     \end{enumerate}
\end{proposition}
\begin{remark}
\label{rem:indepence}
     Note that Proposition \ref{prop:independence} implies that 
     $\bcum[E] = 0$ for all $E\in\mathcal{E}(2,I,\ell)$ if and only if $\E[b^E]= \E[b(j_\ell)] \E[b^{E\setminus \ell}]$ for all $E\in\mathcal{E}(2,I,\ell)$.
 \end{remark}
\begin{proof}
    By \eqref{eq:expressionE} we always have for $E\subset I$ that 
    \begin{align}
        \label{eq:proof_independence_sym}
        \E[b^E]=&\E[b(j_\ell)] \E[b^{E\setminus \ell}] + 
             \sum_{\substack{p\subset E \\
             |p|=2,  (\ell,j_\ell)\in p}} \bcum[p]\E[b^{E\setminus p}]
             +
             \sum_{\substack{U\subset E \\
             |U|\ge3,  (\ell,j_\ell)\in U}} \bcum[U]\E[b^{E\setminus U}].
    \end{align}
    Hence, the condition \eqref{eq:k3_independence} implies \eqref{eq:E3_independence}. 
    Suppose that \eqref{eq:E3_independence} holds. Then \eqref{eq:proof_independence_sym}
    implies for $n=3$ that  
    \begin{align}
        \label{eq:proof_independence_sym2}
             \sum_{\substack{U\subset E \\
             |U|\ge n,  (\ell,j_\ell)\in U}} \bcum[U]\E[b^{E\setminus U}]
             =
             0
             \quad \text{for all} \quad E\in\mathcal{E}(3,I,\ell).
    \end{align}
    Suppose \eqref{eq:proof_independence_sym2} holds for $n=k\ge3$.  Take any
    $E\in \mathcal{E}(3,I,\ell)$ with $|E|=k$. Then
    \eqref{eq:proof_independence_sym2} for $n=k$ implies
    $\bcum[E]=\bcum[E]\E[b^{E\setminus E}]=0$, since $U=E$ is the only subset
    of $E$ satisfying $|U|\ge k$. Therefore \eqref{eq:proof_independence_sym2}
    holds for $n=k+1$. Hence, by induction
    we obtain that the condition \eqref{eq:k3_independence} holds.
\end{proof}

\begin{proposition}
     Fix $I\in\mathcal{I}(J)$ such that  $|I|\ge2$. Fix a pair $p\subset I$, $|p|=2$. 
     Then let for $n\ge 2$,
     \[
     \mathcal{E}(I,p) = \{E\subset I\colon  E\cap p\ne \emptyset\}\,, \qquad
         \mathcal{E}_1(n,I,p) = \{E\subset I\colon |E|\ge n, |E\cap p|=1\}.
     \]
     Then the following two conditions are equivalent:
     \begin{enumerate}
        \item[(1)] There holds
         \begin{align}
             \label{eq:Epair_independence}
             \E[b^E]=&\E[b^{p\cap E}] \E[b^{E\setminus p}] +
             \sum_{\substack{U\subset E\\ p\subset U, |U|\ge 3}}
             \bcum[U]\E[b^{E\setminus U}] 
         \end{align}
         for all $E\in\mathcal{E}(I,p)$.
     \item[(1')]
         There holds
         \begin{align}
             \label{eq:kpair_independence}
             \bcum[E] =& 0 
         \end{align}
         for all $E\in\mathcal{E}_1(2,I,p)$.
     \end{enumerate}
\end{proposition}

\begin{proof}
    Given some $E\in\mathcal{E}(I,p)$, we first use moments to cumulants formula \eqref{eq:m2k},
    and split the partitions into three cases: (1) there are clusters whose union equals $p\cap E$,
    (2) $p$ is a proper subset of a cluster, and (3) the remaing cases, which can be identified by the fact that at least one
    cluster in the partition belongs to $\mathcal{E}_1(2,E,p)$.  We obtain
    \begin{align}
        \nonumber
         \E[b^E]
         &=
         \left( 
             \bcum[p\cap E]+\ind(p\subset E)\prod_{(\ell,j_\ell)\in p}\bcum[j_\ell]
         \right) 
         \left(
             \sum_{\pi\in\p(E\setminus p)}
             \prod_{A\in\pi}
             \bcum[A]
         \right)
         \\
        \nonumber
         &\quad+
         \sum_{\substack{U\subset E\\ p\subset U,|U|\ge3}}
         \bcum[U]
         \left( 
             \sum_{\pi\in \p(E\setminus U)}
             \prod_{A\in\pi}
             \bcum[A]
         \right) 
         +
             \sum_{\pi\in \p(E)}
         \ind({\mathcal{E}_1(2,E,p)\cap \pi\ne\emptyset})
             \prod_{A\in\pi}
             \bcum[A]
         \\
         &=
         \E[b^{p\cap E}]
         \E[b^{E\setminus p}]
         +
         \sum_{\substack{U\subset E\\ p\subset U,|U|\ge3}}
         \bcum[U]
         \E[b^{E\setminus U}]
         +
             \sum_{\pi\in \p(E)}
             \ind({\mathcal{E}_1(2,E,p)\cap \pi\ne\emptyset})
             \prod_{A\in\pi}
             \bcum[A]
        \label{eq:proof_pair_independence_sym}
    \end{align}
    where at the last step we used moments to cumulants \eqref{eq:m2k} to rewrite the brackets 
    as expectations. Hence, condition \eqref{eq:kpair_independence} implies condition \eqref{eq:Epair_independence}.

    Suppose that condition \eqref{eq:Epair_independence}
    holds. Then \eqref{eq:proof_pair_independence_sym} implies for $n=2$ that 
    \begin{align}
        \label{eq:proof_pair_independence_sym2}
             \sum_{\pi\in \p(E)}
             \ind({\mathcal{E}_1(n,E,p)\cap \pi\ne\emptyset})
             \prod_{A\in\pi}
             \bcum[A]
         =0 \quad \text{for all} \quad E\in \mathcal{E}( I,p).
    \end{align}
    Suppose that \eqref{eq:proof_pair_independence_sym2} holds for $n=k\ge2$
    and $\bcum[E]=0$ for $E\in  \mathcal{E}_1(2, I,p)$, $|E| < k$.  Take any
    $E\in  \mathcal{E}_1(2, I,p)\subset\mathcal{E}(I,p)$, $|E|=k$. Then
    \eqref{eq:proof_pair_independence_sym2} for $n=k$ implies $\bcum[E]=0$,
    since $\mathcal{E}_1(k,E,p)=\{E\}$. This means that $\bcum[E]=0$ for any
    $E\in  \mathcal{E}_1(2, I,p)$, $|E|\leq k$. Therefore,
    \eqref{eq:proof_pair_independence_sym2} holds for $n=k+1$. Thus, by induction the condition
    \eqref{eq:kpair_independence} holds.
\end{proof}

The following two result point out the main benefit of using the symmetric
cumulants for  commuting and canonically commuting elements in the
algebra: the cumulants are  then completely commuting, apart possibly from
order two.

\begin{proposition}
    \label{prop:cum_per_invariance}
    Suppose $I\in\I(J)$ is an index sequence for elements which satisfies
    the following generalized canonical commutation relations: if $a,a'\in
    \pr_2(\es{b}{I})$, then there is a scalar $c(a,a')\in \C$ such that $[a,a']=c(a,a')\id$.
    Suppose $P$ is a permutation of $I$. If $|I|\ne2$, then
    \begin{align}
        \label{eq:cum_per_invariance}
        \bcum[I]=\bcum[P[I]]
        \,.
    \end{align}
    If the elements are pairwise commuting, $c(a,a')=0$ for all $a,a'\in
    \pr_2(\es{b}{I})$, then \eqref{eq:cum_per_invariance_F} holds also if $|I|=2$.
\end{proposition}

\begin{proof}
    Note that the case $|I|=1$ is trivial. If $c(a,a')=0$, then $aa'=a'a$ and therefore 
    \begin{align}
        \bcum[a,a']  
        = \E[aa']-\E[a]\E[a']
        =
        \bcum[a',a]  \,,
    \end{align}
    where shorthand notation \eqref{eq:cum_notation} was used. Hence,
    \eqref{eq:cum_per_invariance_F} holds also for $|I|=2$ in the 
    case $c(a,a')=0$.
    We will now focus on the case $|I|\ge 3$.

    By \eqref{eq:cumorder}, without loss of generality, we may assume that $\pr_1
    I=\{1,2,\dots,|I|\}$ and $P(\pr_1I)=\pr_1I$.
    The remaining claim for even $|I|\ge 4$
    has been proven in \cite[Theorem A.1]{lukkarinen_not_2009}. We will give a proof that works for any $|I|\ge3$.
    Note that since $|I|<\infty$ and all permutations of finite sets are compositions of swaps, it is
    sufficient to prove \eqref{eq:cum_per_invariance} for all swaps. Fix $m\in\{1,2,\dots,|I|-1\}$. 
    Let $P$ be the permutation that swaps the $m$th and $(m+1)$th element, namely, 
    $P(m)=m+1$, $P(m+1)=m$ and $P(\ell)=\ell$, otherwise.

    We will prove \eqref{eq:cum_per_invariance} by using induction in $|I|$.
    For the first case $|I|=3$ we have $I=\{(1,j_1),(2,j_2),(3,j_3)\}$. 
        Let $i$ be the integer for
    which $\{i,m,m+1\}=\{1,2,3\}$. Then moments to cumulants implies
    \begin{align}
        \E[b^I]-\E[b^{P[I]}] =& \bcum[I]-\bcum[P[I]]
        + \bcum[{j_m},{j_{m+1}}]\bcum[j_i]
        - \bcum[{j_{m+1}},{j_m}]\bcum[j_i],
    \end{align}
    where the notation \eqref{eq:cum_notation} is used.
    The commutation relation implies
    $\E[b^{P[I]}]=\E[b^I]+c_{m+1,m}\E[b(j_i)]$ and $\bcum[{j_{m+1}},{j_m}]=
    \bcum[{j_m},{j_{m+1}}]+c_{m+1,m}$. Here
    $c_{\ell,\ell'}=c(b(j_{\ell}),b(j_{\ell'}))$. Since $\bcum[j_i]=\E[b(j_i)]$, we
    obtain $\bcum[I]=\bcum[P[I]] $. This proves that
    \eqref{eq:cum_per_invariance} holds for all swaps when $|I|=3$

   Fix $n\ge3$ and suppose that \eqref{eq:cum_per_invariance} holds for all swaps when
   $3\leq|I|\leq n$. Suppose $|I|=n+1$. 
   Let $A_0=\{(m,j_m),(m+1,j_{m+1})\}$. Then moments-to-cumulants together with \eqref{eq:cum_per_invariance}
   for $|A|<n+1$  implies 
   \begin{align*}
       \E[b^I]-\E[b^{P[I]}]
       &= \bcum[I]
        +
        \sum_{\substack{\pi\in\p(I)\\|\pi|>1}}
        \ind(A_0\in \pi)
        \bcum[{j_m},{j_{m+1}}]
        \prod_{\substack{A\in\pi\\A\ne A_0}}\bcum[A]
       \\
        &
         -\bcum[P[I]]
        -
        \sum_{\substack{\pi\in\p(P[I])\\|\pi|>1}}
        \ind(A_0\in \pi)
        \bcum[{j_{m+1}},{j_m}]
        \prod_{\substack{A\in\pi\\A\ne A_0}}\bcum[A]
        \\
        &= \bcum[I]-\bcum[P[I]]
       -c_{m+1,m} \E[b^{I\setminus m\setminus m+1}],
   \end{align*}
   where in the last step we used $\bcum[{j_{m+1}},{j_m}]=
   \bcum[{j_m},{j_{m+1}}]+c_{m+1,m}$, the fact that $P$ swaps only $m$ and $m+1$ and 
    moments to cumulants to obtain $\E[b^{I\setminus m\setminus m+1}]$. Using 
    the commutation relation $[b(i_{m+1}),b(i_m)]=c_{m+1,m}\id$ gives $\E[b^{P[I]}]
    =\E[b^I]+c_{m+1,m}\E[b^{I\setminus m\setminus m+1}]$.
    Hence, $\bcum[\es{b}{I}]=\bcum[P\es{b}{I}]$. This proves \eqref{eq:cum_per_invariance}
    for all swaps which implies it for all permutations.
\end{proof}

\begin{proposition}
    \label{prop:reg_m2k} \textup{(Regularized symmetric moments to cumulants formula)}
    Fix $L\in\N_0$ and take index sets $I',I_1,\dots,I_L\in\I(J)$.
    Let $I=I'+\sum_{l=1}^{L}I_l$. Then
    for any permutation $P$ there holds
    \begin{align}
        \label{eq:regm2k}
        \E\left[\Xi \biggl(\es{b}{P[I']}\times\Big(\bigtimes_{l=1}^L\bwick{\es{b}{P[I_l+n_l]}}\Big)\biggr)\right]=\sum_{\pi\in\p(I)}^{}
        \prod_{A\in\pi}^{}\ind(A\not\subset I_l+n_l\ \forall l) \bcum[P[A]],
    \end{align}
    where $n_l=\max(\pr_1 (I'+\sum_{i=1}^{l-1}I_{i}))$ so that $I_l+n_l\subset I$. In particular, if $P$ is the identity permutation, 
    then 
     \begin{align}
        \E\left[{b}^{I'}\prod_{l=1}^L\bwick{{b}^{I_l}}\right]=\sum_{\pi\in\p(I)}^{}
        \prod_{A\in\pi}^{}\ind(A\not\subset I_l+n_l\ \forall l) \bcum[A].
    \end{align}
\end{proposition}

\begin{remark}
    Recall how $\Xi $ acts  from \eqref{eq:Xi}. Recall also the definition of $\times$ in (\ref{eq:prod_interlace}) and note that each $\pr_1(I')$, $\pr_1(I_l+n_l)$, $\dots,$ $ \pr_1(I_L+n_L)$ are pairwise disjoint.
\end{remark}

\begin{proof}
    We will do induction in $L$. We note that the case $L=0$ holds as it is the 
    standard moments to cumulants formula \eqref{eq:m2k} for the sequence $P[I']$.

    Take $m\in\N$. Suppose that \eqref{eq:regm2k} holds for $L < m$. We will prove that 
    \eqref{eq:regm2k} holds for $L=m$ by doing indunction in $|I_1|$.  If $|I_1|=0$, then $I_1=\emptyset$, and therefore 
    $I=I'+\sum_{l=1}^{m}I_l=I'+\sum_{l=2}^{m}I_l$ and $\es{b}{I'}\prod_{l=1}^L\bwick{\es{b}{I_l}}=\es{b}{I'}\prod_{l=2}^m\bwick{\es{b}{I_l}}$.
    Hence, 
    by \eqref{eq:regm2k} for $L=m-1$, the identity \eqref{eq:regm2k} is
    true for $L=m$ and $|I_1|=0$. 

    Take $k\in\N$ suppose that \eqref{eq:regm2k} holds for $L=m$ and $|I_1|<k$.
    Let $|I_1|=k$.
    For a given partition $\pi\in\p(I)$ we define 
    \[
        \pi_1:=\{A_1\in\pi\colon A_1\not\subset I_1+n_1\}
        \text{ and  }
        \pi_0:=\pi\setminus\pi_1.
    \]
    Then for $E=\cup\pi_0$ we have $E\subset I_1+n_1$ and $\pi_0\in\p(E)$, $\pi_1\in\p(I\setminus E)$. For a
    fixed $\pi$ the decomposition $\pi=\pi_0\cup\pi_1$ is unique, and hence 
    $$
    1=\sum_{E\subset I_1+n_1}
    \sum_{\pi_0\in\p(E)}\sum_{\pi_1\in\p(I\setminus
    E)}\ind(\pi=\pi_0\cup\pi_1)\prod_{A_1\in\pi_1}\ind(A_1\not\subset I_1+n_1).
    $$
    This identity, 
    together with \eqref{eq:opseq_prod}, the regularized
    moments-to-cumulants formula \eqref{eq:regm2k} for $L=m-1$, and 
    the fact $\es{b}{P[I']}\times \es{b}{P[I_1+n_1]}=\es{b}{P[I'+I_1]}$, implies
    \begin{align*}
        &\E\left[\Xi \, \es{b}{P[I']}\times \es{b}{P[I_1+n_1]} \times \bigg(
        \bigtimes_{l=2}^m\bwick{\es{b}{P[I_l+n_l]}}
        \bigg)
        \right]
        \\
        &=\E\left[\Xi \, \es{b}{P[I'+I_1]}\times 
        \bigg(
        \bigtimes_{l=2}^m\bwick{\es{b}{P[I_l+n_l]}}
        \bigg)
        \right]
        \\
        &=\sum_{\pi\in\p(I)}^{} 
        \prod_{A\in\pi}
        \ind(A\not\subset I_l+n_l \ \forall l=2,\dots,m)
        \bcum[P[A]]
        \\
        &=
        \sum_{\pi\in\p(I)}
        \sum_{E\subset I_1+n_1}
        \sum_{\pi_0\in\p(E)}
        \sum_{\pi_1\in\p(I\setminus E)}
        \ind(\pi=\pi_0\cup\pi_1)
        \left( 
            \prod_{A_1\in\pi_1}
            \ind(A_1\not\subset I_1+n_1)
        \right) 
        \\
        &\qquad\times
        \prod_{A\in\pi}
        \ind(A\not\subset I_l+n_l \ \forall l=2,\dots,m)
        \bcum[P[A]]
        \\
        &=
        \sum_{E\subset I_1+n_1}
        \sum_{\pi_0\in\p(E)}
        \sum_{\pi_1\in\p(I\setminus E)}
        \left( 
        \prod_{A_0\in\pi_0} \bcum[P[{A_0}]]
        \right) 
        \left( 
        \prod_{A_1\in\pi_1}
        \ind(A_1\not\subset I_l+n_l \ \forall l)
        \bcum[P[A_1]]
        \right),
    \end{align*}
    where in the last step we used the indicator function $\ind(\pi=\pi_0\cup\pi_1)$ to compute the $\pi$ sum,
    and in the $A_1$ product we merged the remaining indicator functions. In the $A_0$ product we 
    used $E\subset I_1+n_1$ to compute that $\ind(A_0\not\subset I_l+n_l \
    \forall l=2,\dots,m)=1$ for $A_0\in \pi_0$.

    Splitting the sum into $E=\emptyset$ and $E\ne\emptyset$ yields
    \begin{align*}
        &\E\left[\Xi \, \es{b}{P[I']}\times \es{b}{P[I_1+n_1]} 
        \times
        \bigg(
        \bigtimes_{l=2}^m\bwick{\es{b}{P[I_l+n_l]}}
        \bigg)
        \right]
        \\
        &=
        \sum_{\pi\in\p(I)}
        \prod_{A\in\pi}
        \ind(A\not\subset I_l+n_l \ \forall l)
        \bcum[P[A]]
        \\
         &+
         \sum_{\emptyset\ne E\subset I_1+n_1}
         \left( 
        \sum_{\pi_0\in\p(E)}
        \prod_{A_0\in\pi_0} \bcum[P[{A_0}]]
         \right) 
         \left( 
        \sum_{\pi_1\in\p(I\setminus E)}
        \prod_{A_1\in\pi_1}
        \ind(A_1\not\subset I_l+n_l \ \forall l)
        \bcum[P[A_1]]
         \right) 
        \\
        =&
        \sum_{\pi\in\p(I)}
        \prod_{A\in\pi}
        \ind(A\not\subset I_l+n_l \ \forall l)
        \bcum[P[A]]
        \\
         &+
         \sum_{\emptyset\ne E\subset I_1+n_1}
         \E[ b^{P[E]}]
         \E\left[\Xi \, \es{b}{P[I']}\times \bwick{\es{b}{P[(I_1+n_1)\setminus
         E]}}\times
         \bigg(
         \bigtimes_{l=2}^m\bwick{\es{b}{P[I_l+n_l]}}
         \bigg)
         \right],
    \end{align*}
    where we used the standard moments to cumulants formula \eqref{eq:m2k} and 
    the induction hypothesis in $|I_1|$. Hence, using \eqref{eq:wickseq} for
    $\es{b}{P[I_1+n_1]}$  implies \eqref{eq:regm2k} for $L=m$ and $|I_1|=k$. Hence,
    induction in $|I_1|$ and  induction in $L$ imply that \eqref{eq:regm2k}
    holds.
    \end{proof}

Let us finally state how cumulants of products are evaluated.  In the following
Proposition \eqref{prop:iteratedcum}, we give the explicit formulae only for
the simplest cases where only one of the arguments
is replaced by a monomial or a Wick polynomial.  More general cumulants of
monomials -formulas can be obtained from these by iteration.  

In order to give the above mentioned formula, we define a permutation that will be also 
used in the rest of the paper. Given index sequences $I,I'\in\I(J)$ and $\ell\in\pr_1(I)$, 
we define $P_\ell^{I',I}\colon \N\to\N$ to be a permutation such that when acting on 
subsets of $I'+I\setminus\ell$ it inserts indices in $I'$ to be at the position where $\ell$ 
was in $I$. Recall from \eqref{eq:indexpermutation} how permutations act on index sequences.

\begin{remark}
    \label{rem:insert_permutation}
    The above mentioned permutation $P_\ell^{I',I}$, which acts on 
    subsets of $I'+I\setminus\ell$ by inserting indices in $I'$ to be at the position where $\ell$ 
    was in $I$, may explicitly be constructed as follows.
    First, for given $\ell,n'\in\N$, we define a permutation 
    $\tp_\ell^{n'}\colon\N\to\N$,
    \begin{align}
            \label{eq:perP}
            \tp_\ell^{n'}(m)=
            \begin{cases}
                m+\ell-1, & \ m\leq n'  
                \\
                m-n', & n'<m\leq \ell-1+n'
                \\
                m, &m\ge \ell+n'
            \end{cases}.
    \end{align}
    Note that $\tp_\ell^{n'}$ maps the first $n'$ integers to be after $\ell-1+n'$, 
    for example when acting on the sequence $(1,2,3,4,5,6,\dots)$ componentwise by $\tp_3^2$ we get 
    $(3,4,1,2,5,6,\dots)$.

    Take $I,I'\in\I(J)$ and $\ell\in\pr_1(I)$. 
    Let $f\colon\N\to\N$ be a bijection such that
    $f(\pr_1(I'+I\setminus\ell))=\{1,2,\dots,|I'+I\setminus\ell|\}$
    and that  $f(m)<f(n)$ for all $m,n\in \pr_1(I'+I\setminus\ell)$, $m<n$.
    Then 
    \begin{align}
        \label{eq:unnoralized_perP}
        P_\ell ^{I',I} = \tp_\ell^{|I'|}\circ f
        \,.
    \end{align}
     Note that if $\pr_1
    I=\{1,2,\dots,|I|\}$ and $\pr_1 I'=\{1,2,3,\dots,|I'|\}$, we can take
    $f=\text{id}$, and then $P_\ell ^{I',I} = \tp_\ell^{|I'|}$.
    We will be only interested on how $P_\ell^{I',I}$ acts on subsets of $I'+I\setminus\ell$ and 
    this does not depend on the choice of $f$.
\end{remark}

\begin{proposition}
    \label{prop:iteratedcum} \textup{(Cumulants of products)}
    Suppose $I,I'\in \I(J)$ and $(\ell,j)\in I$. 
    Let the permutation $P_\ell^{I',I}$ be as in Remark \eqref{rem:insert_permutation}.
    Suppose $\tilde{j}\in J$  and define
    $$\tilde{I}_\ell :=I\setminus\{(\ell,j)\}\cup\{(\ell, \tilde{j})\}.$$
    We consider two options: (1) 
    $b(\tilde{j})= b^{I'}$ or (2) 
    $b(\tilde{j})=\bwick{b^{I'}}$. 
    
    Then, in case (1) we have $b^{\tilde{I}_\ell} =
    b^{I|_{<\ell}}{b^{I'}}b^{I|_{>\ell}} =
    b^{P_\ell^{I',I}(I'+I\setminus\ell)}$ and in case (2) $b^{\tilde{I}_\ell} =
    b^{I|_{<\ell}}\bwick{b^{I'}}b^{I|_{>\ell}}$.  In addition, in case (1) 
    \begin{align}
        \label{eq:cofprod_B1}
       \bcum[{\tilde{I}_\ell}]
       = \sum_{\pi\in\p(I'+I\setminus\ell)}^{}
       \prod_{A\in\pi}^{}\ind(A\cap I' \ne \emptyset) \bcum[P_\ell^{I',I}[A]].
    \end{align}
    and in case (2), 
    \begin{align}
        \label{eq:cofprod_B2}
       \bcum[{\tilde{I}_\ell}]
       = \sum_{\pi\in\p(I'+I\setminus\ell)}^{}
       \prod_{A\in\pi}^{}\ind(A\cap I' \ne \emptyset)\ind(A\not\subset I') \bcum[P_\ell^{I',I}[A]].
    \end{align}
\end{proposition}
\begin{remark}
     Note that if $J$ does not have such element $\tilde{j}$, then one can take $\tilde{j}\not\in J$ and extend $b$ to $J\cup\{\tilde{j}\}$ by setting in case (1) 
    $b(\tilde{j})=b^{I'}$ or in case (2) 
    $b(\tilde{j})=\bwick{b^{I'}}$.
\end{remark}

\begin{proof}
    Let $P=P_\ell^{I',I}$ and $n_0=\max(\pr_1(I'))$.
    Proposition \ref{prop:reg_m2k} implies that in the case (1) the RHS of \eqref{eq:cofprod_B1} 
    equals to 
    $$\E[\Xi \left( \es{b}{P[I']}\times\bwick{\es{b}{P[I\setminus\ell+n_0]}}\right)]=:\zeta$$ 
    and in the case (2) the RHS of \eqref{eq:cofprod_B2} 
    equals to 
    $$\E[\Xi \left( \bwick{\es{b}{P[I']}}\times \bwick{\es{b}{I\setminus\ell+n_0]}}\right)]=:\zeta.$$ Here, 
    we introduced $\zeta$ to denote the appropriate expectation in each of the two cases.

    Let $I'_1=\{(1,\tilde{j})\}$ and $P_1=P_\ell^{I'_1,I}$. In case (1) we have $\Xi\, \es{b}{P[I']}=b^{I'}=\Xi\, \es{b}{I'_1}$ and 
    in case (2) we have $\Xi\, \bwick{\es{b}{P[I']}}=\bwick{b^{I'}}=\Xi\, \es{b}{I'_1}$.
    Therefore, in case (1) we have
    \begin{align*}
        \Xi \left( \es{b}{P[I']}\times\bwick{\es{b}{P[I\setminus\ell+n_0]}} \right)
        =
        \Xi \left( \es{b}{P_1[I'_1]}\times\bwick{\es{b}{P_1[I\setminus\ell+1]}}\right)
    \end{align*}
    and in case (2) we have
    \begin{align*}
         \Xi \left( \bwick{\es{b}{P[I']}}\times\bwick{\es{b}{I\setminus\ell+n_0]}}]\right)
        =
        \Xi \left( \es{b}{P_1[I'_1]}\times\bwick{\es{b}{P_1[I\setminus\ell+1]}}\right)
        \,.
    \end{align*}
    Hence, in both case (1) and (2) we have 
    \begin{align*}
        \zeta
        &=
        \E[\Xi\left(  \es{b}{P_1[I'_1]}\times\bwick{\es{b}{P_1[I\setminus\ell+1]}}\right)]
      =
        \bcum[P_1[I'_1+I\setminus\ell]]
        \,,
    \end{align*}
    where at the last step we used Proposition \eqref{prop:reg_m2k}. Since
    $P_1[I'_1+I\setminus\ell]= \tilde{I}_\ell$ the result holds in
    both cases.
\end{proof}

%% file: section4.tex
\section{Antisymmetric cumulants and Wick polynomials}
\label{sec:anti}

\subsection{Antisymmetric cumulants}\label{sec:deffcum}

Suppose we are given an expectation map $\E\colon\A\to\C$ on a complex unital algebra $\A$.
For now, suppose $J=\A$ and $b=\text{id} \colon J\to\A$. We define
{\it antisymmetric cumulants} $\fcum[\es{b}{I}]$ for  $I\in\I(J)$
by setting $\bcum[\emptyset]:=0$ and then
doing induction in $|I|\ge 1$ to the recursion relation
\begin{align}
    \label{eq:m2k_Fodd}
    \E[b^I] =
    \sum_{\pi\in\p'_2(I)}^{}\sgn(\pi) \prod_{A\in\pi}^{} \fcum[\es{b}{A}]\,,
\end{align}
where $\p'_2(I)$ denotes the collection of partitions of $I$ which contain
\defem{at most one} odd element, i.e., for every $\pi\in\p_2'$ there is at most one $A_1\in\pi $ with $|A_1|$ odd while $|A|$ is even for every other $A\in\pi$, $A\ne A_1$. Here
$\sgn(\pi)$ is $+1$ or $-1$ according to whether the permutation of $I$ which corresponds to $I\mapsto
(A_1,A_2,\ldots,A_{|\pi|})$
is even or odd. 
We call $A\in\pi$ a \defem{cluster} in the partition $\pi$. 
Here $\pi = \{A_1,A_2,\dots,A_n\}$, $n=|\pi|$, and any ordering
of the clusters $A_i$ of the partition $\pi$ may be used without affecting the
resulting sign (recall that at most one cluster $A_i$ is then odd).   
More precisely,
$\sgn(\pi)$ is the sign of the permutation $P:\pr_1(I)\to\pr_1(I)$ 
constructed by requiring that $P$ preserves relative order for elements in each cluster $A_i$ and 
it orders the clusters in to the chosen order $A_1,\dots,A_n$. Preserving
relative order in $A_i$ means that $P(\ell_1)<P(\ell_2)$ for all
$\ell_1,\ell_2\in \pr_1(A_i)$, $\ell_1<\ell_2$, and ordering the clusters means 
that for each two clusters $A_i,A_m$, $i<m$, there holds $P(\ell_i)<P(\ell_m)$ for 
all $\ell_i\in \pr_1(A_i)$, $\ell_m\in \pr_1(A_m)$.

We will sometimes use the shorthand
notation $\sgn(A_1,A_2,\dots,A_n)$ for $\sgn(\pi)$ if
$\pi = \{A_1,A_2,\dots,A_n\}$.  In particular, this sign is uniquely defined by given the order
of the clusters $A_k$ in the
sequence even if $\pi$ is a general partition, i.e., even if it contains odd clusters.
For example, if $E$ and $I$ are even and nonempty sets for which $E\subsetneq
I$, then also  $I\setminus E$ must be even and nonempty, and thus
$\sgn(\pi)=\sgn(E,I\setminus E)=
\sgn(I\setminus E,E)$ for $\pi=\{E,I\setminus E\}\in \p_2(I)$.
In addition, to these notations we use the following definitions for the
analogous degenerate cases: we set $\sgn(\emptyset,I) := 1$ and
$\sgn(I,\emptyset) := 1$ (note that both cases may be thought of as degenerate
examples of an identity permutation).

Then for  any index set $J$ and any map $b\colon J\to \A$ the equation \eqref{eq:m2k_Fodd} holds for all $I\in\I(J)$.
The antiymmetric cumulants may be considered
as a map $\fcum\colon \I(\A)\to\C$ and 
$\fcum[\es{b}{\cdot}]$ as a map from $\I(J)$ to $\C$.
For the rest of this Section \ref{sec:anti} we assume that 
 $J$ is any index set and $b\colon J\to \A$ any map.
 The motivation for more general choice of $b$ and  $J$ becomes apparent in the applications in Section \ref{sec:hie_bos_fer}.

The equation \eqref{eq:m2k_Fodd} is called the moments-to-cumulants formula.
Assuming an even $I$,
the definition \eqref{eq:m2k_F} is also used in \cite{brattelirobinson2}, where the
name fermionic truncated correlation functions is used over cumulants. 
Note that for even $I$, the moments-to-cumulants formula \eqref{eq:m2k_Fodd} reduces to 
\begin{align}
    \label{eq:m2k_F}
    \E[b^I] = \sum_{\pi\in\p_2(I)}^{}\sgn(\pi) \prod_{A\in\pi}^{} \fcum[\es{b}{A}]\,,
\end{align}
where $\p_2(I)$ denotes the collection of partitions of the set $I$ containing
only even sets. Note that for even $|I|$ we have $\p_2'(I)=\p_2(I)$.
In addition, if $\E$ is an even
expectation map with respect to $\bes\subset b(J)\subset\A$, then one can check by induction that
$\fcum[\es{b}{I}]=0$ for all odd $I$ with
$\pr_2(I)\subset \bes$.
One often 
rewrites the equation \eqref{eq:m2k_Fodd} as
\begin{align}
      \label{eq:m2k_Fall}
    \fcum[\es{b}{I}]=\E[b^I] - 
    \sum_{\substack{\pi\in\p'_2(I)\\ |\pi|>1}}
    \sgn(\pi) 
    \prod_{A\in\pi}
    \fcum[\es{b}{A}]\,, \qquad I \ne \emptyset\, ,
\end{align}
in order to more easily see and use the recursion relation.

If we insist that also for $I=\emptyset$
the formula \eqref{eq:m2k_Fodd} holds with the standard definitions of vacuous
products and sums, we
should define $\p'_2(\emptyset):=\{\emptyset\}$
 and
$\sgn(\emptyset):=1$,
since for $\pi=\emptyset$ the condition $A\in \pi$ is false and thus vacuously
$\prod_{A\in \pi}^{} \left(\cdot\right):= 1$.  Hence, with these definitions the
right hand side of \eqref{eq:m2k_Fodd} evaluates to $1= \E[b^\emptyset]$.

Just like in the symmetric case, we would like to consider also here permuted
cumulants. Recall from \eqref{eq:indexpermutation} how permutations act on index sequences.  Given a permutation $P\colon\N\to\N$ and a  index sequence $I\in\I(J)$ also
$P[I]$ is in $\I(J)$ and thus
\begin{align}
    \label{eq:m2k_F_P}
    \E[b^{P[I]}] =
    \sum_{\pi_0\in\p'_2(P[I])}^{}\sgn(\pi_0) \prod_{A\in\pi_0}^{}
    \fcum[\es{b}{A}] =        \sum_{\pi\in\p'_2(I)}^{}\sgn(\pi;P)
    \prod_{A\in\pi}^{} \fcum[\es{b}{P[A]}]\,,
\end{align}
where in the second equality we have changed summation variables from
$\pi_0\in\p'_2(P[I])$ to $\pi = \{P^{-1}[A]: A\in\pi_0\}$.  This clearly
results in a bijection between $\p'_2(P[I])$ and $\p'_2(I)$, the inverse given
by the map $P(\pi) := \{P[A]: A\in\pi\}$ for $\pi\in\p'_2(I)$.  Hence, the
corresponding products of cumulants are mapped as stated above.  However, the
relative sign can change in this operation, and in the above formula we have
defined
\[
 \sgn(\pi;P) := \sgn(P(\pi)) \,, \qquad \pi\in\p'_2(I)\,.
\]
For example, if $P$ is the permutation $(1234)\mapsto (1324)$ and
$\pi=\{\{1,2\},\{3,4\}\}$ we have $\sgn(\pi)=1$ but
$P(\pi)=\{\{1,3\},\{2,4\}\}$ and thus $\sgn(\pi;P)=-1$.

Whenever the underlying map $b\colon J\to\A$ is clear, we will use the shorthand notation 
\begin{align}
    \label{eq:cum_notation0_fer}
    \fcum[\es{b}{I}]=\fcum[I]\,,\qquad I\in\I(J)\,.
\end{align}
Also in the antisymmetric case, $\fcum[I]$ depends only on the relative order of the position 
labels in $I$, i.e., if $f:\N\to\N$ is strictly increasing, then 
\begin{align}
    \label{eq:fer_cumorder}
    \fcum[I] = 
    \fcum[{f[I]}],
\end{align}
for all $I\in\I(J)$, where $f[I] = \{ (f(\ell),j)\ : \ (\ell,j)\in I\}\in\I(J)$. This follows recursively from
\eqref{eq:m2k_Fodd} and the fact that $b^I=b^{f[I]}$.
Therefore, the following notations 
are justified, 
\begin{align}
    \label{eq:cum_notation_fer}
    \fcum[b(j_1),  \dots, b({j_n})] 
    =
    \fcum[j_1,  \dots, {j_n}] 
    =
    \fcum[I]
    =
    \fcum[\es{b}{I}]
    ,
\end{align}
where $I=\{(\ell_1,{j_1}),(\ell_2,{j_2}),\dots, (\ell_n,{j_n})\}$ and $\ell_1<\ell_2<\dots<\ell_n$.

The lowest order cumulants are thus
\begin{align*}
 &\fcum[b_1] = \E[b_1]\,, \qquad \fcum[b_1, b_2]
 =  \E[b_1 b_2]\,, \\
 &
 \fcum[b_1, b_2, b_3] =
\E[b_1 b_2 b_3]
 - \E[b_1 b_2] \E[b_3]
 + \E[b_1 b_3] \E[b_2]
 - \E[b_2 b_3] \E[b_1]\,,\\
 &
 \fcum[b_1, b_2, b_3,b_4] =
  \E[b_1 b_2 b_3 b_4]
 - \E[b_1 b_2]\E[b_3 b_4]
 + \E[b_1 b_3]\E[b_2 b_4]
 - \E[b_1 b_4]\E[b_2 b_3]
\,.
\end{align*}

\subsection{Antisymmetric Wick polynomials}

We define antisymmetric Wick polynomials $\fwick{b^I}$ 
 by setting $\bwick{b^\emptyset}:=\id$
and for $I\in\I(J)$, $I\ne\emptyset$, define them recursively by requiring
\begin{align}
    \label{eq:fer_wick}
    \fwick{b^I}:=b^I- 
    \sum_{\substack{\emptyset \ne E\subset I,\\ |E|\text{ even }}}
    \sgn(E,I\setminus E)
    \E[b^E]\fwick{b^{I\setminus E}}\,.
\end{align}
The same formula is used for both odd and even $|I|$,
but in either case we only allow removal of even clusters.
The lowest order examples are
\begin{align*}
 &\fwick{b_1} = b_1\,, \qquad \fwick{b_1 b_2}
 =  b_1 b_2 - \E[b_1 b_2]\id \,, \\
 &
 \fwick{b_1 b_2 b_3} =
 b_1 b_2 b_3 
 - \E[b_1 b_2] b_3
 +  \E[b_1 b_3] b_2
 - \E[b_2 b_3] b_1
 \,.
\end{align*}
By induction, one can prove that $\E[\fwick{b^I}]=0$ if $I\ne \emptyset$
and  $|I|$ is even.
However, $\E[\fwick{b^I}]$ might not be zero if $|I|$ is odd; we will later
prove that then $\E[\fwick{b^I}]=\fcum[I]$. Also, it should be stressed that the notation borrowed from
physics is somewhat misleading: the Wick polynomial $\fwick{b^I}$ depends on
the full sequence $I$, not only on the algebra element $b^I$.

As previously discussed, we need to be careful with permutations of
polynomials.    Hence, for each
$I\in\I(J)$ we define the abstract antisymmetric Wick polynomial
$\fwick{\es{b}{I}}\in\ves(\I(J))$ by setting $\fwick{\es{b}{\emptyset}}=\es{b}{\emptyset}=\emptyset$
and for $I\ne \emptyset$ recursively by 
\begin{align}
    \label{eq:fer_wickseq}
    \fwick{\es{b}{I}}:= \es{b}{I}-
    \sum_{\substack{\emptyset\ne E\subset  I,\\ |E|\text{ even }}}
    \sgn(E,I\setminus E) \E[b^E]\fwick{\es{b}{I\setminus E}}\,.
\end{align}
It is straightforward to check that
\begin{align}
    \label{eq:fer_wickseq_is_wick}
    \Xi\fwick{\es{b}{I}}=\fwick{b^I}
\end{align}
for all $I\in \I(J)$.

\begin{remark}
    \label{rem:samelenght_wick_F}
    Similar alternative construction of abstract antisymmetric Wick polynomials can 
    be done as in Remark \ref{rem:samelenght_wick} was done for abstract symmetric 
    polynomials. One just has to replace the recursion relation
    \eqref{eq:samelenght_wick} by
    \begin{align}
        \label{eq:samelenght_wick_F}
        \fwick{\es{b}{U;I}}:= 
        \es{b}{U}\cup \es{\id}{I\setminus U}-
         \sum_{\substack{\emptyset\ne E\subset  U,\\ |E|\text{ even
    }}}\sgn(E,U\setminus E)
        \E[b^E]
        (\fwick{\es{b}{U\setminus E;U}}\cup\es{\id}{I\setminus U}).
    \end{align}
\end{remark}

\subsection{Properties of antisymmetric cumulants and Wick polynomials}

We start by proving the following lemma for the change of the sign of a partition after splitting it into two parts.
\begin{lemma}
    \label{lem:composition}
    Take $I\in\I(J)$ and $\pi\in\p'_2(I)$. Let $\pi_1\subset\pi$. Define $\pi_0=\pi\setminus\pi_1$ and
    $E=\cup\pi_0$. Then $\pi_0\in \p'_2(E)$, $\pi_1\in \p'_2(I\setminus E)$
    and
    \begin{align}
        \label{eq:composition}
            \sgn(\pi)=\sgn(E,I\setminus
            E)\sgn(\pi_0)\sgn(\pi_1).
    \end{align}
    Here, we recall that $\sgn(\emptyset,I) = 1$
    and $\sgn(I,\emptyset) = 1$.
\end{lemma}
\begin{proof}
As a first step, let us point out that the formula is also valid if either
$\pi_0$ or $\pi_1$ is empty.  For example, if $\pi_1=\emptyset$, then
$\pi_0=\pi$, $E=I$ and $I\setminus E = \emptyset$.  Therefore, $\pi_0\in
\p'_2(I)$, $\pi_1\in\p'_2(\emptyset)=\{\emptyset\}$, $\sgn(\pi_1)=1$ and
$\sgn(\pi_0)=\sgn(\pi)$ while $\sgn(I,\emptyset)=1$.  If $\pi_0=\emptyset$, we
have $\pi_1=\pi$, and similar conclusions hold, with the roles of these two
sets reversed and using $\sgn(\emptyset,I)=1$.

Now if both $\pi_0$ and $\pi_1$ are nonempty,
the sets $E$ or $I\setminus E$ are also nonempty, and there is at most one odd
cluster in either of them since there is at most one odd cluster in $\pi$.  We
can make a reference permutation of $\pi$, $I\mapsto  (A)_{A\in\pi}$, as a
composition of the permutations $I\mapsto (E,I\setminus E)$, $E\mapsto
(A)_{A\in\pi_0}$,  and $I\setminus E\mapsto  (A)_{A\in\pi_1}$.  The result
follows since the sign of compositions of permutations is the product of their
individual signs. 
\end{proof}

Next we prove the following useful expression for cumulants and moments.
\begin{proposition}\label{th:fcumiterformulas}
    Take $I\in\I(J)$ and fix $(\ell,j)\in I$.
    If $|I|$ is odd, we have
    \begin{align}
        \label{eq:expressionE_oddfer}
        \E[b^I] 
        &= 
        \sum_{\substack{E\subset I,\\(\ell,j)\in E}}
        \sgn(E,I\setminus E)
        \E[b^{I\setminus E}]\fcum[E]
        =   \sum_{\substack{E\subset I,\\|E|\text{ odd}}}
        \sgn(E,I\setminus E)
        \E[b^{I\setminus E}]\fcum[E]\,,      
    \end{align}
    and if $|I|$ is even,
    \begin{align}
        \label{eq:expressionE_fer}
        \E[b^I]
        &=
        \sum_{\substack{E\subset I,\\{(\ell,j)\in E},\, |E|\,\text{even}}}
        \sgn(E,I\setminus E)
        \E[b^{I\setminus E}]\fcum[E]\,.
    \end{align}
  In particular, for all $I$,
    \begin{align}
        \label{eq:expressionK_fer}
        \fcum[I] 
        &= 
        \E[b^I] 
        -
        \sum_{\substack{E\subsetneq I,\\(\ell,j)\in E}} \ind(|I| \text{ is odd or } |E| \text{ is even})
        \sgn(E,I\setminus E)
        \E[b^{I\setminus E}]\fcum[E].
    \end{align}
\end{proposition}
\begin{proof}
We begin with the moments to cumulants formula \eqref{eq:m2k_Fodd}.  Suppose first
that $|I|$ is odd.  Then for each $\pi\in \p'_2(I)$ there is a unique $E\in \pi$ which has an odd $|E|$ and then $I\setminus E$ is even.
Thus we can add a unit factor
$1=\sum_{E\subset I}\sum_{\pi_1\in\p_2(I\setminus E)} \ind(|E|\text{ odd})
\ind(\pi=\pi_0\cup \pi_1)$, where $\pi_0:=\{E\}$, into the summand.  Since
$\sgn(\pi_0)=1$, we can then use  Lemma \ref{lem:composition} to conclude that
\[
  \E[b^I] = \sum_{\substack{E\subset I,\\|E|\text{ odd}}} \sgn(E,I\setminus E)\fcum[E]  \sum_{\pi_1\in\p_2(I\setminus E)}
            \sgn(\pi_1) \prod_{A\in\pi_1}
            \fcum[A]\,.
\]
Using \eqref{eq:m2k_F} yields the righmost expression in \eqref{eq:expressionE_oddfer}.

Let us then consider general $I$.  Given $(\ell,j)\in I$ and $\pi\in \p'_2(I)$ there is again a unique $E\in \pi$ with $(\ell,j)\in E$.
Continuing as above yields
    \begin{align*}
        \E[b^I] 
&=
        \sum_{\substack{E\subset I,\\(\ell,j)\in E}}
        \sgn(E,I\setminus E)
        \fcum[E]
        \left( 
            \sum_{\pi_1\in\p'_2(I\setminus E)}
            \sgn(\pi_1) \prod_{A\in\pi_1}
            \fcum[A]
        \sum_{\pi\in\p'_2(I)} \ind(\pi =\pi_1 \cup \{E\})
        \right) \, .
    \end{align*}
    There are then four cases to consider, depending on the number of elements in $E$ and $I$.
   If $I$ is even and $E$ is odd, then there are no partitions $\pi\in\p'_2(I)$ with $E\in \pi$, so the term evaluates to zero.  If both $E$ and $I$ are even, also $I\setminus E$ is even, and thus $\p'_2(I)=\p_2(I)$ and $\p'_2(I\setminus E)
   = \p_2(I\setminus E)$.  We find that
   $\sum_{\pi\in\p'_2(I)} \ind(\pi =\pi_1 \cup \{E\})=1$ and the term in parentheses evaluates to $\E[b^{I\setminus E} ]$, by  \eqref{eq:m2k_F}.
    If both $I$  and $E$ are odd, $I\setminus E$ is even and
    $\sum_{\pi\in\p'_2(I)} \ind(\pi =\pi_1 \cup \{E\})=1$ for any $\pi_1\in
    \p'_2(I\setminus E)
   = \p_2(I\setminus E)$.  Hence, again the term in parentheses evaluates to $\E[b^{I\setminus E} ]$, by  \eqref{eq:m2k_F}.  Finally, if $I$ is odd and $E$ is even, there is exactly one odd cluster in $\pi$ which will belong to $\pi\setminus \pi_0$.  Now we can use \eqref{eq:m2k_Fodd} to conclude that the term still evaluates to $\E[b^{I\setminus E} ]$. Then solving
    $\fcum[I]$ from \eqref{eq:expressionE_oddfer} or \eqref{eq:expressionE_fer} implies \eqref{eq:expressionK_fer}.
\end{proof}

\begin{proposition}
    The antisymmetric cumulants are multilinear in the sense that
    \begin{align}
        \nonumber
        \fcum[b(j_1),\dots,  \alpha b(i)+\beta b(k) ,\dots, b({j_n})] 
        &=
         \alpha
        \fcum[b(j_1),\dots,   b(i) ,\dots, b({j_n})] 
        \\
        \label{eq:fcum_multilinear}
        &\quad+
         \beta 
        \fcum[b(j_1),\dots,  b(k) ,\dots, b({j_n})] 
    \end{align}
    for all $j_1,\dots,j_n,i,k\in J$ and $\beta,\alpha\in\C$.
    Similarly, the antisymmetric Wick polynomials are multilinear.
\end{proposition}
\begin{proof}
The multilinearity \eqref{eq:fcum_multilinear} follows recursively from
\eqref{eq:expressionK_fer} together with the  multilinearity of $\E[b^I]$. Similarly we obtain multilinearity of the 
antisymmetric Wick polynomials from \eqref{eq:fer_wick}.
\end{proof}

We next present how the results related to ``independence'' and ``Gaussianity'' are modified 
for antisymmetic cumulants.
\iffalse
  For real random variables, if a collection contains one random variable which
  is independent from the rest, then the corresponding cumulant is zero.
  Although the meaning of ``independence'' is more difficult to capture in
  non-commutative setting, let us present below two factorization properties
  which are connected to vanishing of certain antisymmetric cumulants.  One
  special aspect about the results below is that they do not require full
  independence, it suffices that the cumulants are zero for subsequences of the
  reference sequence $I$.  The first case is motivated by properties of
  Gaussian distributions and the second case of independence.
\fi
\begin{proposition}
    \label{prop:independence_fer}
     Fix $I\in\mathcal{I}(J)$ and choose $(\ell,j)\in I$.
     Let, for $n\in \N$,
     \[
         \mathcal{E}(n,I,\ell)=\{ E\subset I \colon |E|\ge n,\, (\ell,j)\in E\}.
     \]
     Denote $I_\ell := \{(\ell,j)\}$.
     Then the following two conditions are equivalent:
     \begin{enumerate}
        \item[(1)]  For all $E\in\mathcal{E}(3,I,\ell)$, we have
         \begin{align}
             \label{eq:ferE3_independence}
             & \E[b^E]=
             \ind(|E|\,\text{odd})
             \sgn(I_\ell,E\setminus I_\ell)
             \E[b(j)]\E[b^{E\setminus I_\ell}]
             + \sum_{\substack{I_\ell\subset p\subset E ,\\
             |p|=2}}\sgn(p,E\setminus p) \E[b^p]\E[b^{E\setminus p}] \,.
         \end{align}
     \item[(1')]
     For all $E\in\mathcal{E}(3,I,\ell)$, we have
         \begin{align}
             \label{eq:ferk3_independence}
            & \fcum[E] = 0 \,.
         \end{align}
     \end{enumerate}
 \end{proposition}
 \begin{remark}
     Note that Proposition \ref{prop:independence_fer} implies that 
     $\fcum[E] = 0$ for all $E\in\mathcal{E}(2,I,\ell)$ if and only if $\E[b^E]= 0$ for all $E\in\mathcal{E}(2,I,\ell)$.  The proof of this statement is very similar to the one given below, and we will skip its details.
 \end{remark}
\begin{proof}
If $n>|I|$, we have
$\mathcal{E}(n,I,\ell)=\emptyset$ and all statements assuming $E\in \mathcal{E}(n,I,\ell)$ hold vacuously.  Hence, we may assume $|I|\ge 3$.

By \eqref{eq:expressionK_fer}, if  $(\ell,j_\ell)\in E\subset I$, then
\begin{align}
        \label{eq:proof_fer_independence_sym}
        \E[b^E]
        &=
        \sum_{\substack{F\subset E,\\(\ell,j)\in F}} \ind(|E| \text{ is odd or } |F| \text{ is even})
        \sgn(F,E\setminus F)
        \E[b^{E\setminus F}]\fcum[F]\,.
    \end{align}
 In the sum, if $|F|\ge n$, then $F\in \mathcal{E}(n,I,\ell)$.
 Hence, condition \eqref{eq:ferk3_independence} implies \eqref{eq:ferE3_independence}.

    Suppose that \eqref{eq:ferE3_independence} holds. Then \eqref{eq:proof_fer_independence_sym}
    implies for $n=3$ that
    \begin{align}
        \label{eq:proof_fer_independence_sym2}
        \sum_{U\in\mathcal{E}(n,E,\ell)}
        \ind(|E| \text{ is odd or } |U| \text{ is even})
        \sgn(U,E\setminus U)
        \E[b^{E\setminus U}]\fcum[U]
        =0\,,
        \quad \text{for all }
        E\in \mathcal{E}(3,I,\ell).
    \end{align}
    Suppose that $k\ge 3$ and that \eqref{eq:proof_fer_independence_sym2} holds for $n=k$ and that \eqref{eq:ferk3_independence} holds for every $E\in\mathcal{E}(3,I,\ell)$ with $|E|<k$ which is trivially true for the base case $k=3$. Assume $
E\in \mathcal{E}(3,I,\ell)$ with $|E|=k$. Then
$\mathcal{E}(k,E,\ell)=\{E\}$.  Hence, in
\eqref{eq:proof_fer_independence_sym2}
necessarily $U=E$ which
implies that $\fcum[E]=0$.
Hence, \eqref{eq:ferk3_independence}
holds if $|E|\le k$ and thus  \eqref{eq:proof_fer_independence_sym2} holds for $n=k+1$. Hence, by induction
the condition \eqref{eq:ferk3_independence} holds.
\end{proof}

Below is a similar results which shows that antisymmetic cumulants also can be used to check if it is possible to ``factor out a pair'' from an index sequence.
\begin{proposition}
     Fix $I\in\mathcal{I}(J)$ with $|I|\ge 2$, and choose a pair $p\subset I$, $|p|=2$.
     Let 
     \[
     \mathcal{E}(I,p) = \{E\subset I\colon  E\cap p\ne \emptyset\}\,, \qquad
         \mathcal{E}_1(I,p) = \{E\subset I\colon  |E\cap p|=1\}.
     \]
     Then the following two conditions are equivalent:
     \begin{enumerate}
        \item[(1)] For all $E\in \mathcal{E}(I,p)$,
         \begin{align}
             \label{eq:ferEpair_independence}
             \E[b^E]=&
         \sum_{p\subset U\subset E} \ind(|E| \text{ is odd or } |U| \text{ is even})
        \sgn(U,E\setminus U)
        \E[b^{E\setminus U}]\fcum[U]\,,
         \end{align}
        where the right hand side is vacuously zero, in case $p\not \subset E$.
     \item [(1')]
      For all $E\in\mathcal{E}_1(I,p)$, we have
         \begin{align}
             \label{eq:ferkpair_independence}
             \fcum[E] =& 0 \,.
         \end{align}
     \end{enumerate}
\end{proposition}
\begin{proof}
Suppose $E\in \mathcal{E}(I,p)$ and fix
$(\ell,j)\in E\cap p$.  Denote
$I_\ell := \{(\ell,j)\}$.
Then,
by \eqref{eq:expressionK_fer},
\begin{align}
        \label{eq:proof_ferpair_independence}
        \E[b^E]
        &=
        \sum_{\substack{F\subset E,\\(\ell,j)\in F}} \ind(|E| \text{ is odd or } |F| \text{ is even})
        \sgn(F,E\setminus F)
        \E[b^{E\setminus F}]\fcum[F]\,.
    \end{align}
If in this sum $p\not \subset F$, we have $F\in  \mathcal{E}_1(I,p)$.
Otherwise, $p\subset F\subset E$, and we obtain
\begin{align}\label{eq:proof_pair_fer_independence_sym}
        \E[b^E]
        &=
        \sum_{p\subset U\subset E} \ind(|E| \text{ is odd or } |U| \text{ is even})
        \sgn(U,E\setminus U)
        \E[b^{E\setminus U}]\fcum[U]
      \nonumber \\ & \qquad
      +
        \sum_{I_\ell \subset F\subset E, F\in  \mathcal{E}_1(I,p)} \ind(|E| \text{ is odd or } |F| \text{ is even})
        \sgn(F,E\setminus F)
        \E[b^{E\setminus F}]\fcum[F]\,
        ,
        \end{align}
 Hence, the condition
    \eqref{eq:ferkpair_independence} implies that the condition
    \eqref{eq:ferEpair_independence} holds.

    Suppose that  the condition \eqref{eq:ferEpair_independence} holds. Applying this to $E\in\mathcal{E}_1(I,p)\subset\mathcal{E}(I,p)$, we obtain
      from \eqref{eq:proof_pair_fer_independence_sym} that
    for $n=1$
    \begin{align}
        \label{eq:proof_pair_fer_independence_sym2}
        & \sum_{F\in  \mathcal{E}_1(I,p),\, F\subset E,\, |F|\ge n} \ind(|E| \text{ is odd or } |F| \text{ is even})
        \sgn(F,E\setminus F)
        \E[b^{E\setminus F}]\fcum[F] = 0\,,
        \nonumber \\ &
        \quad
        \text{for all }
        E\in\mathcal{E}_1(I,p)\,.
    \end{align}
   Suppose that $k\ge 1$ and that \eqref{eq:proof_pair_fer_independence_sym2} holds for $n=k$ and that
   \eqref{eq:ferkpair_independence} holds for every
   $E\in\mathcal{E}_1(I,p)$ with $|E|<k$ which is trivially true for the base case $k=1$. Assume $
E\in \mathcal{E}_1(I,p)$ with $|E|=k$. Then, only the term $F=E$ appears on the left hand side of \eqref{eq:proof_pair_fer_independence_sym2}, and thus
 $\fcum[E]=0$.
Hence, \eqref{eq:ferkpair_independence}
holds if $|E|\le k$ and thus  \eqref{eq:proof_pair_fer_independence_sym2} holds for $n=k+1$. Hence, by induction
the condition \eqref{eq:ferkpair_independence} holds always.
\end{proof}

The following two results point out the main benefit of using the antisymmetric
cumulants for  anticommuting and canonically anticommuting elements in the
algebra: the cumulants are  then completely anticommuting, apart possibly from
order two.

\iffalse
\begin{proposition}
    \label{prop:cum_per_invariance_F}
    Suppose $I\in\I(J)$ is an index sequence for elements which satisfies
    the following generalized canonical anticommutation relations: if $a,a'\in
    \pr_2(\es{b}{I})$, then there is a scalar $c(a,a')\in \C$ such that $\{a,a'\}=c(a,a')\id$.
    Suppose $P:\N\to\N$ is a permutation of $I$, in the sense
    that $P$ leaves $\N\setminus \pr_1(I)$ invariant. If $|I|\ne2$, then
    \begin{align}
        \label{eq:cum_per_invariance_F}
        \fcum[I]=\sgn(P)\fcum[P[I]]
        \,,
    \end{align}
    where $\sgn(P)$ denotes the sign of the finite permutation
    $P|_I$.

    If the elements are pairwise anticommuting, $c(a,a')=0$ for all $a,a'\in
    \pr_2(\es{b}{I})$, then \eqref{eq:cum_per_invariance_F} holds also if $|I|=2$.
\end{proposition}
\begin{remark}\label{rem:signsofpermutI}
 In the above statement, the value of $\sgn(P)$ is computed from the
 restriction $P|_I$ and not from the map $P$ itself.  More precisely, given any
 $\pr_1(I)$ there is a unique increasing function $f$ such that it maps the
 standard enumerating set $M:=\{1,2,\ldots\,|I|\}$ bijectively onto
 $\pr_1(I)$.  By assumption, $Q:=P|_I:I\to I$ is a bijection, and thus
 $P':=f^{-1}\circ Q \circ f$ is a permutation of $M$.  In the above, we define
 $\sgn(P):=\sgn(P')$.  Note that by \eqref{eq:fer_cumorder}, then also
 $\fcum[I]=\fcum[f(M)]$ and $\fcum[P[I]]=\fcum[f(P'[M])]$.
\end{remark}
\fi

\begin{proposition}
    \label{prop:cum_per_invariance_F}
    Suppose $I\in\I(J)$ is an index sequence for elements which satisfies
    the following generalized canonical anticommutation relations: if $a,a'\in
    \pr_2(\es{b}{I})$, then there is a scalar $c(a,a')\in \C$ such that $\{a,a'\}=c(a,a')\id$.
    Suppose $P:\N\to\N$ is a permutation, then
    \begin{align}
        \label{eq:cum_per_invariance_F}
        \fcum[I]=\sgn(P|_I)\fcum[P[I]]
        \,,
    \end{align}
    where $\sgn(P|_I)$ denotes the sign of the permutation $P|_I\colon \pr_1(I)\to\pr_1(I)$, specified in Remark \ref{rem:signsofpermutI}.

    If the elements are pairwise anticommuting, $c(a,a')=0$ for all $a,a'\in
    \pr_2(\es{b}{I})$, then \eqref{eq:cum_per_invariance_F} holds also if $|I|=2$.
\end{proposition}
\begin{remark}\label{rem:signsofpermutI}
    In the above statement, the value $\sgn(P|_I)$ is computed as 
    follows. Given any
     $\pr_1(I)\subset\N$ there is a unique increasing function $f\colon P (\pr_1(I)) \to \pr_1(I)$. 
     %such that $f\circ P (\pr_1(I)) = \pr_1(I)$. 
     Therefore, $P|_I:=f\circ P|_{\pr_1(I)}$ 
     is a bijection on the finite set $\pr_1(I)$ and $\sgn(P|_I)$ denotes its sign. 
      Note that by \eqref{eq:fer_cumorder}, then also
     %$\fcum[I]=\fcum[f[I]]$ and 
     $\fcum[P[I]]=\fcum[P|_I[I]]$.
\end{remark}

\begin{proof}
The proof is an adaptation of the argument from  the proof of  \cite[Theorem
A.1]{lukkarinen_not_2009}.  It is sufficient to consider sequences $I$ with the
standard enumerating set, assuming $\pr_1(I)=\{1,2,\ldots,|I|\}$, since then
there is a increasing function to map the result to $I$ with general
enumerating sets, as explained in Remark \ref{rem:signsofpermutI}.

If $|I|=1$, only identity permutation is allowed and the statement is obvious.
If $|I|=2$, we have $\es{b}{I}=\{(1,a),(2,a')\}$ and,
then, $\fcum[I]=\E[aa']=-\E[a'a]+c(a,a')=-\fcum[P[I]]+c(a,a')$ for the
permutation $P$ which swaps the elements, i.e., $(12)\mapsto (21)$.  Since this
is the only nontrivial permutation and $\sgn(P)=-1$, we have proved the result
assuming $|I|=2$ for the fully anticommuting case for which $c(a,a')=0$.  There
is nothing to prove in the general statement for these $I$.

It remains to prove the general statement for $I$ such that $|I|\ge 3$.
We do this by induction in $k=|I|$.
Assume thus that $k\ge 3$ and $P$ is a corresponding permutation.
Since any permutation can be written as a product of swaps of neighbouring
elements, it suffices to prove the result for any $m$, $1\le m\le k-1$, using a
permutation $P$ for which $P(m)=m+1$, $P(m+1)=m$, and $P(n)=n$ for $n\ne
m,m+1$.  Then $\sgn(P)=-1$ and
$b(j_m)b(j_{m+1})=-b(j_{m+1})b(j_m)+ c \id$, with $c=c(b(j_m),b(j_{m+1}))\in
\C$ given. Here,  $j_\ell\in J$ is the 
index for which $(\ell,j_\ell)\in I$. Denoting $p=\{(m,j_m),(m+1,j_{m+1})\}\subset I$, we obtain
\begin{align}\label{eq:basicCAR}
  \E[b^I] = -\E[b^{P[I]}] + c \E[b^{I\setminus p}]\,.
\end{align}

For the base case, $|I|=3$, we considering first $m=1$.  Then, explicitly
\begin{align*}
 & \fcum[b_1, b_2, b_3] +  \fcum[b_2, b_1, b_3]  =
\E[b_1 b_2 b_3] + \E[b_2 b_1 b_3]
\\ & \qquad
 - (\E[b_1 b_2]+\E[b_2 b_1]) \E[b_3]
 + \E[b_1 b_3] \E[b_2] +  \E[b_2 b_3] \E[b_1]
 - \E[b_2 b_3] \E[b_1] - \E[b_1 b_3] \E[b_2]
\\ & \quad
 = c \E[b_3] - c \E[b_3]=0\,,
\end{align*}
where the shorthand notation $b_\ell=b(j_\ell)$ was used.
Similarly, one can prove the case with $m=2$ by cheking that $\fcum[b_1, b_2,
b_3] +  \fcum[b_1, b_3, b_2] =0$.  Hence, for $k=3$ the statement holds.

Suppose then that $k\ge 4$ and the result holds whenever $|I|<k$.
Assume $|I|=k$. We obtain from  \eqref{eq:expressionK_fer}:
 \[
   \fcum[I] = \E[b^I]
        -
        \sum_{\substack{E\subsetneq I,\\(m,j_m)\in E}}
        \ind(|I| \text{ is odd or } |E| \text{ is even})
        \sgn(E,I\setminus E)
        \E[b^{I\setminus E}]\fcum[E]\,.
  \]
Consider some $E$ in the above sum.  Then $|E|<k$, and our goal is first to show that, if $E\ne p$, then
\begin{align}\label{eq:antisymmbasestep}
  \sgn(E,I\setminus E)
        \E[b^{I\setminus E}]\fcum[E]
        =
         -\sgn(P[E],P[I]\setminus P[E])
        \E[b^{P[I]\setminus P[E]}]\fcum[P[E]]\,.
\end{align}
If $m+1\not \in E$, $\fcum[P[E]]=\fcum[E]$, since then $P|_E$ can be collapsed to identity using \eqref{eq:fer_cumorder}; for instance, clearly $b^E = b^{P[E]}$.  However, $\sgn(E,I\setminus E)=-\sgn(P[E],P[I]\setminus P[E])$ since a permutation for the latter split is obtained by first acting on $P[I]$ with $P=P^{-1}$ and then using the permutation $I\mapsto (E,I\setminus E)$.  Thus,  \eqref{eq:antisymmbasestep} holds in this case.

Otherwise, $p\subset E$, so $2\le |E|<k$.  If $E\ne p$, then $|E|\ge 3$, and we can use the induction assumption and conclude that $\fcum[E]=-\fcum[P[E]]$.  In this case,
$P$ leaves $I\setminus E$ invariant and $P|_E$ is a permutation of $E$.
Hence, $\sgn(E,I\setminus E)=\sgn(P[E],P[I]\setminus P[E])$.
We conclude that \eqref{eq:antisymmbasestep} holds also in this case.

Finally, if $E=p$, then $\sgn(E,I\setminus E)=+1=\sgn(P[E],P[I]\setminus P[E])$, and
\[
 \fcum[E] = \E[b^p] = -\E[b^{P[p]}] + c = -\fcum[P[E]] + c \,.
\]
Therefore, for all allowed $E$
\[
  \sgn(E,I\setminus E)
        \E[b^{I\setminus E}]\fcum[E]
        =
         -\sgn(P[E],P[I]\setminus P[E])
        \E[b^{P[I]\setminus P[E]}]\fcum[P[E]] + \ind(E=p)c
        \E[b^{I\setminus p}]\,,
\]
which implies that
\[
 \fcum[I] + \fcum[P[I]]
 = \E[b^I] + \E[b^{P[I]}] - c \E[b^{I\setminus J}]
       \sum_{\substack{E\subsetneq I,\\(m,j_m)\in E}}\ind(|I| \text{ is odd or } |E| \text{ is even}) \ind(E=J)\,.
\]
Since $|I|\ge 4$, we  have $p\subsetneq I$ and $|p|$ is even, and thus the sum over  the indicator function evaluates to one for any such $I$. Therefore, by \eqref{eq:basicCAR},
$\fcum[I]=- \fcum[P[I]]$.   This completes the induction step and the proof of the Proposition.
\end{proof}

\begin{proposition}
    \label{prop:reg_m2k_F} \textup{(Regularized antisymmetric moments to cumulants formula)}
    Fix $L\in\N$ and take index sets $I',I_l\in\mathcal{I}, \ l=1,\dots,L$.
    Let $I=I'+\sum_{l=1}^{L}I_l$ and define $n_l=\max(\pr_1 (I'+\sum_{i=1}^{l-1}I_{i}))$ 
    so that $I_l+n_l\subset I$. 
    Suppose we are given a
    permutation $P\colon\N\to\N$. If
    \begin{align}
        \label{eq:regP1}
        \sgn(P[E],P[I]\setminus P[E])
        &=
        \sgn(P[E],P[I_l+n_l]\setminus P[E])
    \end{align}
    for all $E\subset I_l+n_l$, with even $|E|\ge 2$, and for all $l$, then
    \begin{align}
        \label{eq:regm2k_F}
       & \E\left[\Xi \biggl( \es{b}{P[I']}\times \Big(\bigtimes_{l=1}^L\fwick{\es{b}{P[I_l+n_l]}}\Big) \biggr)\right]
        \nonumber \\ & \quad
        =\sum_{\pi\in\p'_2(I)}^{}
        \sgn(\pi;P)
        \prod_{A\in\pi}^{}\ind(|A|\text{ odd, or }A\not\subset I_l+n_l\ \forall l) \fcum[P[A]]\,.
    \end{align}
    
    In particular, \eqref{eq:regP1} is satisfied and \eqref{eq:regm2k_F} holds if $P$ has the following property: for any $l$ and any
    $\ell_1,\ell_2\in \pr_1(I_l+n_l)$ with $P(\ell_1)<P(\ell_2)$, there is an even number of
    elements from different sets between them, i.e., 
    \begin{align}
        \label{eq:regP2}
    |\{\ell\in \pr_1(I):
    P(\ell_1)<\ell<P(\ell_2),\, \ell\not \in \pr_1(I_l+n_l) \}|
    \end{align}
    is even.  
    Moreover, if $P$ is the identity permutation, then
    \begin{align}
        \label{eq:identity_regm2k_F}
        \E\left[ {b}^{I'}\prod_{l=1}^L\fwick{{b}^{I_l}}\right]=\sum_{\pi\in\p'_2(I)}^{}
        \sgn(\pi)
        \prod_{A\in\pi}^{}\ind(|A|\text{ odd, or }A\not\subset I_l+n_l\ \forall l) \fcum[A].
    \end{align}
\end{proposition}

\begin{remark}
First, we point out that, as a corollary, we obtain that $\E[\fwick{b^I}]=\fcum[I]$ if $|I|$ is odd. Recall the definition of $\times$ in (\ref{eq:prod_interlace}) and note that each $\pr_1(I')$, $\pr_1(I_l+n_l)$, $\dots,$ $\pr_1(I_L+n_L)$ are pairwise disjoint.
    Recall also how $\Xi$ acts from
    \eqref{eq:Xi} and the definition of the relative sign, $\sgn(\pi;P)$, of a
    partition $\pi\in\p_2(I)$ after a permutation $P$ of $I$.  The assumption
    on $P$ is used in the proof to find an even permutation to collect all
    elements of $P[I_l+n_l]$, for any $l$, into a consecutive subsequence.

    For general permutations $P$, no such simple expression can be expected in this antisymmetric case.  
    For example, take $I'=\{(1,j'_1),(2,j'_2)\}$, $I_1=\{(1,j_1),(2,j_2)\}$, 
    and choose $P$ such that $P(2)=3$, $P(3)=2$, and $P(i)=i$, otherwise.
    Then $I=\{(1,j'_1),(2,j'_2),(3,j_1),(4,j_2)\}$, 
    $P[I]=\{(1,j'_1),(2,j_1),(3,j'_2),(4,j_2)\}$, $P[I']=\{(1,j'_1),(3,j'_2)\}$,
    $P[I_1+n_1]=\{(2,j_1),(4,j_2)\}$.  Denoting $b_1=b(j_1)$, etc., we can evaluate the left hand side of (\ref{eq:regm2k_F}) to $\E[b'_1 b_1 b'_2 b_2]-\E[b'_1 b'_2] \E[b_1 b_2]$,
    while the right hand side yields
    \[
     \kappa[b'_1, b_1, b'_2, b_2] +\kappa[b'_1, b_1]\kappa[b'_2, b_2] + 
     \kappa[b'_1, b_2]\kappa[b_1,b'_2] = \E[b'_1 b_1 b'_2 b_2]+\E[b'_1 b'_2] \E[b_1 b_2]\,.
    \]
    Hence, in this case the Wick polynomial reqularization picks the wrong sign to cancel out the term which has a cluster internal to $I_1+n_1$.  This example shows that, in contrast to the symmetric case, it is not possible to do the regularization correctly without letting it to depend on $P$ and $I$ instead of just on $P[I_1+n_1]$.
\end{remark}

\begin{proof}
    We will do induction in $L$. We note that the case $L=0$ holds as it is the 
    standard moments to cumulants formula \eqref{eq:m2k_F_P} for the sequence
    $P[I']$, as detailed in Sec.\ \ref{sec:deffcum}.  All characteristic
    functions evaluate then to one since they satisfy the second condition
    vacuously.

    Take $m\in\N$. Suppose that \eqref{eq:regm2k_F} holds for $L < m$. We will prove that 
    \eqref{eq:regm2k_F} holds for $L=m$ by doing indunction in $|I_1| \ge 0$.
    If $|I_1|=0$, then $I_1=\emptyset$, and therefore
    $I=I'+\sum_{l=1}^{m}I_l=I'+\sum_{l=2}^{m}I_l$ and
    $\es{b}{I'}\prod_{l=1}^L\fwick{\es{b}{I_l}}=\es{b}{I'}\prod_{l=2}^m\fwick{\es{b}{I_l}}$.
    Hence, 
    by \eqref{eq:regm2k_F} for $L=m-1$, the identity \eqref{eq:regm2k_F} is
    true for $L=m$ and $|I_1|=0$. 

    Take $k\in\N$, and suppose that \eqref{eq:regm2k_F} holds for $L=m$ and $|I_1|<k$.
    Let $|I_1|=k$.
    For a given partition $\pi\in\p'_2(I)$ we define
    \[
        \pi_1:=\{A_1\in\pi\colon |A_1|\text{ odd, or }A_1\not\subset I_1+n_1\}
        \text{ and  }
        \pi_0:=\pi\setminus\pi_1.
    \]
    Now $\pi_0$ is either empty or contains only even size clusters which are
    subsets of $I_1+n_1$.
    Denote $E=\cup\pi_0$.  Then $|E|$ is even and  $E\subset I_1+n_1$.  If
    $|I|$ is odd, the only odd cluster belongs to $\pi_1$, and thus
    $\pi_0\in\p_2(E)$, $\pi_1\in\p'_2(I\setminus E)$. For a fixed $\pi$ the
    decomposition $\pi=\pi_0\cup\pi_1$ is unique, and hence 
    \[
    1=
    \sum_{\substack{E\subset I_1+n_1
    \\ |E| \text{ even}}}
    \sum_{\pi_0\in\p_2(E)}\sum_{\pi_1\in\p'_2(I\setminus
    E)}\ind(\pi=\pi_0\cup\pi_1)\prod_{A_1\in\pi_1}\ind(|A_1|\text{ odd, or
    }A_1\not\subset I_1+n_1).
    \]
    Next, we aim to use  Lemma \ref{lem:composition} to compute the sign of the
    corresponding split.  We recall that $\sgn(\pi;P)=\sgn(P(\pi))$, and now
    $P(\pi)=P(\pi_0)\cup P(\pi_1)$, thus we conclude that
    \begin{align*}
     &\sgn(\pi;P)=\sgn(P(\pi_0))
    \sgn(P(\pi_1))\sgn(P[E],P[I]\setminus P[E])
    \\& \quad =\sgn(P[E],P[I]\setminus P[E]) \sgn(\pi_0;P)\sgn(\pi_1;P)\,.
    \end{align*}
    Therefore, we can apply these results in the
    regularized moments-to-cumulants formula \eqref{eq:regm2k_F} for $L=m-1$ together with the fact $\es{b}{P[I']}\times\es{b}{P[I_1+n_1]}=\es{b}{P[I'+I_1]}$, and obtain
    \begin{align*}
        &\E\left[\Xi  \, \es{b}{P[I']}\times\es{b}{P[I_1+n_1]} \times \Big(\bigtimes_{l=2}^m\fwick{\es{b}{P[I_l+n_l]}}\Big)\right]
        \\
        &=\E\left[\Xi \, \es{b}{P[I'+I_1]}\times 
        \Big(
        \bigtimes_{l=2}^m\fwick{\es{b}{P[I_l+n_l]}}
        \Big)\right]
        \\
        &=\sum_{\pi\in\p'_2(I)}^{}
        \sgn(\pi;P)
        \prod_{A\in\pi}^{}\ind(|A|\text{ odd, or }A\not\subset I_l+n_l\ \forall l=2,\dots,m) \fcum[P[A]]
        \\
        &=
        \sum_{\pi\in\p'_2(I)}^{}
        \sum_{\substack{E\subset I_1+n_1\\ |E| \text{ even}}}
        \sgn(P[E],P[I]\setminus P[E])
        \sum_{\pi_0\in\p_2(E)}
         \sgn(\pi_0;P)
        \\
        &\qquad\times
         \sum_{\pi_1\in\p'_2(I\setminus E)}
       \sgn(\pi_1;P)
        \ind(\pi=\pi_0\cup\pi_1)
\left(
            \prod_{A_1\in\pi_1}
            \ind(|A_1|\text{ odd, or }A_1\not\subset I_1+n_1)
        \right)
        \\ &\qquad\quad\times
        \prod_{A\in\pi}
        \ind(|A|\text{ odd, or }A\not\subset I_l+n_l \ \forall l=2,\dots,m)
        \fcum[P[A]]
        \\
        &=
        \sum_{\substack{E\subset I_1+n_1\\ |E| \text{ even}}}
        \sgn(P[E],P[I]\setminus P[E])
        \sum_{\pi_0\in\p_2(E)}
        \sgn(\pi_0;P)
        \sum_{\pi_1\in\p'_2(I\setminus E)}
        \sgn(\pi_1;P)
        \prod_{A_0\in\pi_0} \fcum[P[A_0]]
        \\
        &\qquad\times
        \left( 
        \prod_{A_1\in\pi_1}
        \ind(|A_1|\text{ odd, or }A_1\not\subset I_l+n_l \ \forall l)
        \fcum[P[A_1]]
        \right),
    \end{align*}
    where in the last step we used the indicator function $\ind(\pi=\pi_0\cup\pi_1)$ to compute the $\pi$ sum,
    and in the $A_1$ product we merged the remaining indicator functions. In the $A_0$ product, we 
    have used $E\subset I_1+n_1$ to conclude that $A_0\not\subset I_l+n_l \
    \forall l=2,\dots,m$ if $A_0\in \pi_0$.

    Splitting the sum into $E=\emptyset$ and $E\ne\emptyset$ yields
    \begin{align*}
        &\E\left[\Xi \, \es{b}{P[I']}\times\es{b}{P[I_1+n_1]}
        \times\Big(\bigtimes_{l=2}^m\fwick{\es{b}{P[I_l+n_l]}}\Big)\right]
        \\ &\quad
        =\sum_{\pi\in\p'_2(I)}^{}
        \sgn(\pi;P)
        \prod_{A\in\pi}^{}\ind(|A|\text{ odd, or }A\not\subset I_l+n_l\ \forall l) \fcum[P[A]]
        \\
         &\qquad +
       \sum_{\substack{\emptyset\ne E\subset I_1+n_1\\ |E| \text{ even}}}
        \sgn(P[E],P[I]\setminus P[E])
         \left( 
        \sum_{\pi_0\in\p_2(E)}
        \sgn(\pi_0;P)
        \prod_{A_0\in\pi_0} \fcum[P[{A_0}]]
         \right) 
         \\
        &\qquad\quad \times
        \left( 
        \sum_{\pi_1\in\p'_2(I\setminus E)}
        \sgn(\pi_1;P)
        \prod_{A_1\in\pi_1}
        \ind(|A_1|\text{ odd, or }A_1\not\subset I_l+n_l \ \forall l)
        \fcum[A_1]
         \right) 
        \\
        & \quad
       =\sum_{\pi\in\p'_2(I)}^{}
        \sgn(\pi;P)
        \prod_{A\in\pi}^{}\ind(|A|\text{ odd, or }A\not\subset I_l+n_l\ \forall
        l) \fcum[P[A]]
        \\
         & \qquad +
        \sum_{\substack{\emptyset\ne E\subset I_1+n_1\\ |E| \text{ even}}}
        \sgn(P[E],P[I]\setminus P[E])
         \E[ b^{P[E]}]
         \\
        &\qquad\quad \times
        \E\left[\Xi \, \es{b}{P[I']}\times\fwick{\es{b}{P[(I_1+n_1)\setminus
        E]}}\times \Big(\bigtimes_{l=2}^m\fwick{\es{b}{P[I_l+n_l]}}\Big)\right],
    \end{align*}
    where we used the standard moments to cumulants formula \eqref{eq:m2k_F}
    and the induction hypothesis in $|I_1|$. Note that given $E\subset I_1+n_1$
    where $|E|\ge 2$ is even, and recalling the assumption \eqref{eq:regP1} of allowed $P$,
    the sign of the permutation $P[I]\mapsto (P[E],P[I]\setminus P[E])$ is the
    same as that of $P[I_1+n_1] \mapsto (P[E],P[I_1+n_1 ]\setminus P[E])$.
    Hence, using linearity of $\E[\Xi\,\cdot]$ and \eqref{eq:fer_wickseq} for
    $\es{b}{P[I_1+n_1]}$ implies \eqref{eq:regm2k_F} for $L=m$ and $|I_1|=k$.
    Hence, induction in $|I_1|$ and induction in $L$ imply that
    \eqref{eq:regm2k_F} holds as stated in the Proposition.

    Suppose now that \eqref{eq:regP2} holds. Then we can use an even permutation to collect all
    elements of $P[I_1+n_1 ]$ together, after which a permutation separates
    $P[E]$, and as an even set, it can be moved to the beginning and the
    remaining elements of $P[I_1+n_1 ]\setminus P[E]$ moved back to their
    original positions in $P[I]\setminus P[E]$ using an even permutation.
    Therefore, \eqref{eq:regP2} implies \eqref{eq:regP1}.

    Note that for any $l$ and any
    $\ell_1,\ell_2\in \pr_1(I_l+n_l)$ with $\ell_1<\ell_2$, there are zero 
    elements of different sets between them.
    Hence, the identity permutation satisfies
    \eqref{eq:regP2}, and we obtain
        \eqref{eq:identity_regm2k_F}.
 \end{proof}

Let us finally state how cumulants of products are evaluated.  We give the
explicit formulae only for the simplest cases where only one of the arguments
is replaced by a monomial or a Wick polynomial.  More general cumulants of
monomials -formulas can be obtained from these by iteration.  The main
difference to the symmetric case is that now only insertion of odd monomials
leads to simplications in the resulting cumulant expansions.

\begin{proposition}
    \label{prop:iteratedcum_F} \textup{(Cumulants of products)}
    Suppose $I,I'\in \I(J)$ and $(\ell,j)\in I$.  
    Assume that $|I'|$ is odd.
    Let the permutation $P_\ell^{I',I}$ be as in Remark \eqref{rem:insert_permutation}.
    Suppose $\tilde{j}\in J$ and define
    $$\tilde{I}_\ell :=I\setminus\{(\ell,j)\}\cup\{(\ell, \tilde{j})\},$$
    We consider two options: (1) 
    $b(\tilde{j})=b^{I'}$ or (2) 
    $b(\tilde{j})=\fwick{b^{I'}}$. 
    
    Then, in case (1) we have $b^{\tilde{I}_\ell} =
    b^{I|_{<\ell}}{b^{I'}}b^{I|_{>\ell}} =
    b^{P_\ell^{I',I}(I'+I\setminus\ell)}$ and in case (2) $b^{\tilde{I}_\ell} = 
    b^{I|_{<\ell}}\fwick{b^{I'}}b^{I|_{>\ell}}$.  In addition, in case (1) 
    \begin{align}
        \label{eq:cofprod_F1}
       \fcum[{\tilde{I}_\ell}]
       = \sum_{\pi\in\p'_2(I'+I\setminus\ell)}^{}
        \sgn(\pi;P_\ell^{I',I})
        \prod_{A\in\pi}^{}\ind(A\cap I' \ne \emptyset) \fcum[P_\ell^{I',I}[{A}]].
    \end{align}
    and in case (2), 
    \begin{align}
        \label{eq:cofprod_F2}
       \fcum[{\tilde{I}_\ell}]
       = \sum_{\pi\in\p'_2(I'+I\setminus\ell)}^{}
        \sgn(\pi;P_\ell^{I',I})
        \prod_{A\in\pi}^{}\ind(A\cap I' \ne \emptyset)\ind(|A|\text{ odd, or }A\not\subset I') \fcum[P_\ell^{I',I}[{A}]].
    \end{align}
\end{proposition}

\begin{remark}
    Note that if $J$ does not have such element $\tilde{j}$, then one can take $\tilde{j}\not\in J$ and extend $b$ to $J\cup\{\tilde{j}\}$ by setting in case (1) 
    $b(\tilde{j})=b^{I'}$ or in case (2) 
    $b(\tilde{j})=\fwick{b^{I'}}$.  As in Proposition \ref{prop:reg_m2k_F}, there is a difference between the symmetric and antisymmetric results, namely, it is necessary to add the assumption that $|I'|$ is odd.  
    
    For example, take $I=\{(1,j_1),(2,j_2),(3,j_3)\}$, $I'=\{(1,j'_1),(2,j'_2)\}$, and consider case (1) for $\ell=2$.  Then $I'+I\setminus \ell=\{(1,j'_1),(2,j'_2),(3,j_1),(5,j_3)\}$ and $P_\ell^{I',I}$ maps $(1,2,3,5)$ to $ (2,3,1,4)$. Hence, $P_\ell^{I',I}(I'+I\setminus \ell)=\{(1,j_1),(2,j'_1),(3,j'_2),(4,j_3)\}$.  Denoting $b_1=b(j_1)$, etc., we have $b(\tilde{j})=b'_1 b'_2$.  Assume, for simplicity, that the expectation map is even.  Then,
    \[
    \fcum[\tilde{I}_2] = \fcum[b_1,b'_1 b'_2,b_3] = \E[b_1 b'_1 b'_2 b_3]  + \E[b_1 b_3]\E[b'_1 b'_2]\,,
    \]
    while the right hand side of \eqref{eq:cofprod_F1} yields
    \[
   \fcum[b_1,b'_1,b'_2,b_3] + \fcum[b_1,b'_1] \fcum[b'_2,b_3] -  \fcum[b_1,b'_2]\fcum[b'_1,b_3] 
    = \E[b_1 b'_1 b'_2 b_3]  - \E[b_1 b_3]\E[b'_1 b'_2]\,.
    \]
    Hence, the identity fails even in this simple case with even $|I'|$.
%%% Another example:    
%For example, take $I=\{(1,j_1),(2,j_2)\}$, $I'=\{(1,j'_1),(2,j'_2)\}$, and consider case (1) for $%\ell=1$.  Then $P_\ell^{I',I}$ is an identity permutation  and $I'+I\setminus \ell=\{(1,j'_1),(2,j'_2),(3,j_2)\}$.  Denoting $b_1=b(j_1)$, etc., we have $b(\tilde{j})=b'_1 b'_2$.  Now
  %  \[
  %  \fcum[\tilde{I}_1] = \fcum[b'_1 b'_2,b_2] = \E[b'_1 b'_2 b_2] 
  %  \]
  %  while the right hand side of \eqref{eq:cofprod_F1} yields
  %  \[
  %  \fcum[b'_1,b'_2,b_2] - \fcum[b'_1,b_2] \fcum[b'_2] + \fcum[b'_2,b_2] \fcum[b'_1] 
  %  =  \E[b'_1 b'_2 b_2] -  \fcum[b'_1,b'_2] \fcum[b_2]\,.
  %  \]
    \end{remark}

\begin{proof}
    Let $P=P_\ell^{I',I}$, $n'=|I'|$ and $n_0=\max\pr_1(I')$. Define $S=I'+I\setminus\ell$.
    We proceed by induction in $n=|I|$.  If $n=1$, we have $S=I'$, $P=\text{id}$ and
 $\tilde{I}_\ell = \{(\ell,\tilde{j})\}$. Hence, in case (1),
 \[
  \fcum[{\tilde{I}_\ell}] =\E[b^{I'}] = \sum_{\pi\in\p'_2(S)}^{}
        \sgn(\pi;P)
        \prod_{A\in\pi}^{}\ind(A\cap I' \ne \emptyset) \fcum[P[{A}]]\,
 \]
 and in case (2),
 $\fcum[{\tilde{I}_\ell}] =\E[\fwick{b^{I'}}] = \ind(|I'|\text{ odd})
 \fcum[{I'}]$, which satisfies \eqref{eq:cofprod_F2}.  Therefore, all base
 cases are true.

 Suppose then that $n\ge 2$ and 
  the result has been proven whenever $|I|<n$.  Assume $|I|=n$.  We only
  consider the second, more complicated case in detail, the proof of the first
  case follows by minor modifications from this as explained later.  The expression
  \eqref{eq:expressionK_fer} implies that 
        \begin{align*}
     &  \fcum[{\tilde{I}_\ell}]
       =  \E[b^{\tilde{I}_\ell}] 
        -
        \sum_{\substack{E\subsetneq \tilde{I}_\ell,\\(\ell,\tilde{j})\in E}} \ind(|I| \text{ is odd or } |E| \text{ is even})
        \sgn(E,\tilde{I}_\ell\setminus E)
        \E[b^{\tilde{I}_\ell\setminus E}]\fcum[{E}]\,.
    \end{align*}
  Note that $|\tilde{I}_\ell|=|I|$.  
  The induction assumption
  yields
  \begin{align*}
     &  \fcum[{\tilde{I}_\ell}]
       =  \E[b^{I|_{<\ell}}\fwick{b^{I'}}b^{I|_{>\ell}}] 
        -
        \sum_{\substack{E\subsetneq \tilde{I}_\ell,\\(\ell,\tilde{j})\in E}}
        \ind(|I| \text{ is odd or } |E| \text{ is even})
        \sgn(E,\tilde{I}_\ell\setminus E)
        \E[b^{\tilde{I}_\ell\setminus E}]
       \\ & \qquad \times
        \sum_{\pi\in\p'_2(I'+E\setminus \ell)}^{}
        \sgn(\pi;P)
        \prod_{A\in\pi}^{}\ind(A\cap I' \ne \emptyset)\ind(|A|\text{ odd, or
        }A\not\subset I') \fcum[P[{A}]]\,,
    \end{align*}
    where we have used $\sgn(\pi;P_\ell^{I',E})=\sgn(\pi;P)$
    and $\fcum[P_\ell^{I',E}[{A}]]=\fcum[P[{A}]]$.
   On the other hand, by Proposition \ref{prop:reg_m2k_F} we have
   \begin{align*}
     &  \E[b^{I|_{<\ell}}\fwick{b^{I'}}b^{I|_{>\ell}}] 
      =\sum_{\pi\in\p'_2(S)}^{}
        \sgn(\pi;P)
        \prod_{A\in\pi}^{}\ind(|A|\text{ odd, or }A\not\subset I') \fcum[P[{A}]]\,.
 \end{align*}

 For a given $\pi\in\p'_2(S)$ and $E\subset \tilde{I}_\ell$,
 $(\ell,\tilde{j})\in E$, define $E_0=I'+E\setminus\ell$. Then $(\tilde{I}_\ell
 \setminus E)+n_0 =S\setminus E_0$.
 Note that there is unique such $E$ for which $E_0=\cup\{ A\in \pi : A\cap I' \ne \emptyset\}$.
 Therefore,
 \begin{align}
    1 =  
       \sum_{\substack{E\subset \tilde{I}_\ell,\\(\ell,\tilde{j})\in E}}
        \sum_{\pi_0 \in \p'_2(E_0)}       \sum_{\pi_1 \in \p'_2(S\setminus E_0)}
      \ind(\pi = \pi_0\cup \pi_1)
         \prod_{A\in\pi_0}^{}\ind(A\cap I' \ne \emptyset)
         \,.
 \end{align}
 Hence, Proposition \ref{lem:composition} for $\sgn(\pi;P)=\sgn(P(\pi))$ implies
  \begin{align*}
     &  \E[b^{I|_{<\ell}}\fwick{b^{I'}}b^{I|_{>\ell}}] 
      = \sum_{\pi\in\p'_2(S)}^{}
       \sum_{\substack{E\subset \tilde{I}_\ell,\\(\ell,\tilde{j})\in E}}
        \sum_{\pi_0 \in \p'_2(E_0)}       \sum_{\pi_1 \in \p'_2(S\setminus E_0)}
      \ind(\pi = \pi_0\cup \pi_1)
         \prod_{A\in\pi_0}^{}\ind(A\cap I' \ne \emptyset)
    \\ & \qquad  \times    
          \sgn(P[E_0],P[S\setminus E_0])\sgn(\pi_0;P)\sgn(\pi_1;P)
         \prod_{A\in\pi}^{}\ind(|A|\text{ odd, or }A\not\subset I') \fcum[P[{A}]]\,.
 \end{align*}
 As in the proof of Proposition \ref{th:fcumiterformulas}, here
 $\sum_{\pi\in\p'_2(S)}^{} \ind(\pi = \pi_0\cup \pi_1)$ evaluates to zero if
 $|S|=n+n'-1$ is even and $|E_0|=|E|-1+n'$ is odd, and to one, otherwise.  To each $A\in
 \pi_1$, we have $A\cap I' =\emptyset$, so that $\ind(|A|\text{ odd, or
 }A\not\subset I')=1$.
 Summing over $\pi_1$ hence yields a factor
 $\E[b^{S\setminus E_0}] = 
\E[b^{\tilde{I}_\ell \setminus E}]$.  We obtain
  \begin{align*}
     &  \E[b^{I|_{<\ell}}\fwick{b^{I'}}b^{I|_{>\ell}}] 
      = \sum_{\substack{E\subset \tilde{I}_\ell,\\(\ell,\tilde{j})\in E}} \E[b^{\tilde{I}_\ell \setminus E}]
        \sum_{\pi_0 \in \p'_2(E_0)}    \ind(n+n'-1 \text{ odd, or }|E|-1+n'\text{ even})
    \\ & \qquad  \times    
         \sgn(P[E_0],P[S\setminus E_0]) \sgn(\pi_0;P)
         \prod_{A\in\pi_0}^{}\ind(A\cap I' \ne \emptyset)\ind(|A|\text{ odd, or }A\not\subset I') \fcum[P[{A}]]\,.
 \end{align*}
Since $n'$ is odd, the characteristic function is equal to 
$ \ind(n\text{ odd, or }|E|\text{ even})$.  Also the fact that $n'$ is odd implies
$\sgn(P[E_0],P[S\setminus E_0]) = \sgn(E,\tilde{I}_\ell \setminus E)$.  Hence,
the lower order terms cancel
everything apart from the term $E=\tilde{I}_\ell$ for which $|E|=n$ and
$E_0=S=I'+I\setminus\ell$.  We find that 
\begin{align*}
     &  \fcum[{\tilde{I}_\ell}]
       =   \sum_{\pi_0 \in \p'_2(I'+I\setminus\ell)}  
    \sgn(\pi_0;P)
    \prod_{A\in\pi_0}^{}\ind(A\cap I' \ne \emptyset)\ind(|A|\text{ odd, or }A\not\subset I') \fcum[P[{A}]]\,,
\end{align*}
which completes the induction step in this case.

The proof of case (1) is now obtained by removing every indicator function
\\$\ind(|A|\text{ odd, or }A\not\subset I')$ from the above expressions and
replacing $\E[b^{I|_{<\ell}}\fwick{b^{I'}}b^{I|_{>\ell}}]$ by
  \begin{align*}
     &  \E[b^{I|_{<\ell}}{b^{I'}}b^{I|_{>\ell}}]
       =  \sum_{\pi\in\p'_2(S)}^{}
        \sgn(\pi;P)
        \prod_{A\in\pi}^{} \fcum[P[{A}]]\,.
    \end{align*}
This concludes the proof of both cases of the Proposition.
\end{proof}

%% file: section5.tex
\section{Time evolution of symmetric and antisymmetric cumulants}
\label{sec:time-evolution}

We will now consider the time evolution of cumulants when the algebra elements $b$ satisfy an evolution equation with a polynomial rate function.
The following two definitions aim for maximal generality, and we will give concrete examples from quantum systems in the next Section.

\begin{assumption}
    \label{ass:b_time}
    Suppose $U\subset\R$ is a non-degenerate time interval.
    Assume that for each $t\in U$ we have a complex
    unital algebra $\A_t$ and a map $b_t\colon J\to \A_t$, where the
    domain $J$ is independent of time.  Suppose that  $\E\colon \cup_{t\in U}
    \A_t\to \C$ is such that for each $t\in U$ the restriction $\E|_{\A_t}$ is
    an expectation map on $\A_t$.
\end{assumption}

\begin{assumption}
    \label{ass:derivative}
    Suppose that Assumption \eqref{ass:b_time} holds.
    Take 
    an index sequence $I\in\I(J)$, $I\ne\emptyset$.
    Fix $t\in U$.
    Suppose that for any $E\subset I$ the map $s\mapsto \E[b_s^E]$ is
    differentiable at $t$ with derivative $\ddt \E[b_t^E]$.
    Assume that there is a linear map $\partial$ from the subalgebra $\A_t^{(I)}$ generated 
    by $b_t(\pr_2 I)\subset \A_t$ to $\A_t$ such that
    the map $\partial$ determines the derivatives, 
    $\ddt \E[b_t^E] = \E[\partial (b_t^E)]$ for all $E\subset I$, and 
    satisfies the product rule 
    \[
        \partial (AB)= (\partial A) B+
        A \partial B 
    \]
    for all $A,B\in\A_t^{(I)}$.
\end{assumption}

Let us begin with a Lemma which shows why the above conditions lead to simple
evolution relations of both symmetric and antisymmetric cumulans.

\begin{lemma}\label{th:partialtosum}
  Suppose that Assumption  \ref{ass:b_time} holds and let the corresponding
  symmetric and antisymmetric cumulants at time $t$, i.e., those defined using
  the expectation map $\E|_{\mathcal{A}_t}$ and the map $b_t$, be denoted by
  $\bcum[\es{b_t}{I}]$ and $\fcum[\es{b_t}{I}]$, for $I\in \I(J)$.   Suppose
  also that Assumption \ref{ass:derivative} holds
    for some given $t\in U$ and $I\in\I(J)$.   
    
  Then, 
  \begin{align}\label{eq:partialbI}
    \partial(b_t^I) = \sum_{(\ell,j)\in I} b_t^{{I}|_{<\ell}}
    \partial b_t(j)
    b_t^{{I}|_{>\ell}}
  \end{align}
  and
  \begin{align}\label{eq:ddtcumforpartial}
    \ddt\bcum[\es{b_t}{I}]=  \sum_{(\ell,j)\in I} \bcum[\es{\tilde{b}_t}{\tilde{I}_\ell}]\,,\qquad
    \ddt\fcum[\es{b_t}{I}]=  \sum_{(\ell,j)\in I} \fcum[\es{\tilde{b}_t}{\tilde{I}_\ell}]\,,
  \end{align}
  where $\tilde{b}_t$ is an extension of $b_t$
  to $J\cup \{\tilde{j}_\ell \colon (\ell,j_\ell)\in I\}$ defined by $\tilde{b}_t(\tilde{j}_\ell)=\partial b_t(j_\ell)$.
  Furthermore, given $(\ell,j_\ell)\in I$ and $A\subset I$ let $\tilde{A}_\ell$ denote the set  
    $A\setminus{\{(\ell,j_\ell)\}}\cup\{(\ell, \tilde{j}_\ell)\}$.
\end{lemma}
\begin{proof}
We prove all three statements via an induction in $|E|\ge 1$ for $E\subset I$.
Note that then $E$ also satisfies Assumption \ref{ass:derivative} for the given
$t$. If $|E|=1$, there is a unique $\ell\in \N$ such that $E=\{(\ell,j)\}$ and
$\tilde{E}_\ell=\{(\ell, \tilde{j})\}$. Then,
$\partial(b_t^E)=\partial(b_t(j))=b_t(\tilde{j})= b_t^{\tilde{E}_\ell}$, so
\eqref{eq:partialbI} holds for $E$.  Since then
$\bcum[\es{b_t}{E}]=\E[b_t(j)]=\fcum[\es{b_t}{E}]$, we find that both
equations in \eqref{eq:ddtcumforpartial} hold for $E$, by taking expectation of
\eqref{eq:partialbI}.
 
 Suppose then that given $2\le k\le |I|$, the results is true whenever $|E|<k$.  
 If $(\ell,j)\in E$ with $\ell=\min (\pr_1(E))$, we have $\partial(b_t^E) = \partial(b_t(j) b_t^{E\setminus \ell})
 = \partial(b_t(j)) b_t^{E\setminus l} + b_t(j) \partial(b_t^{E\setminus
 \ell})$.  The first term is of the correct form and the induction
 assumption will produce the terms of other elements in $E$.  Hence,
 \eqref{eq:partialbI} holds for $E$ which completes the induction step for that
 equation.
 
 We only give the detailed proof for the corresponding antisymmetric cumulant
 since the same proof works for the symmetric case, after mild modifications.
 Differentiating \eqref{eq:m2k_Fall}, we find
 \begin{align*}
    & \ddt \fcum[\es{b_t}{E}]=\E[\partial(b_t^E)] - 
    \sum_{\substack{\pi\in\p'_2(E)\\ |\pi|>1}}
    \sgn(\pi) \sum_{A\in\pi} \ddt \fcum[\es{b_t}{A}] 
    \prod_{A'\in\pi,\, A'\ne A}
    \fcum[\es{b_t}{A'}]\,.
 \end{align*}
 Then, using the induction assumption and \eqref{eq:partialbI} yields the formula
 \begin{align*}
    &
    \sum_{(\ell,j)\in I} \E[\tilde{b}_t^{\tilde{E}_\ell}]
   - 
    \sum_{\substack{\pi\in\p'_2(E)\\ |\pi|>1}}
    \sgn(\pi) \sum_{A\in\pi} \sum_{(\ell,j)\in E} \ind((\ell,j)\in A) \fcum[\es{\tilde{b}_t}{\tilde{A}_\ell}]
    \prod_{A'\in\pi,\, A'\ne A}
    \fcum[\es{b_t}{A'}]
    \\ & \quad =
    \sum_{(\ell,j)\in E} \biggl( \E[\tilde{b}_t^{\tilde{E}_\ell}]
   - 
    \sum_{\substack{\pi\in\p'_2(E)\\ |\pi|>1}}
    \sgn(\pi) \fcum[\es{\tilde{b}_t}{\widetilde{A(\ell)}_\ell}]
    \prod_{A'\in\pi,\, A'\ne A(\ell)}
    \fcum[\es{b_t}{A'}]\biggr)
    \\ & \quad =
    \sum_{(\ell,j)\in E} \biggl( \E[\tilde{b}_t^{\tilde{E}_\ell}]
   - 
    \sum_{\substack{\pi\in\p'_2(\tilde{E}_\ell)\\ |\pi|>1}}
    \sgn(\pi)
    \prod_{A\in\pi}
    \fcum[\es{\tilde{b}_t}{A}]\biggr)\,.  
 \end{align*}
where in the second expression $A(\ell)$ denotes the unique $A\in \pi$ for which
$(\ell,j)\in A$.  By \eqref{eq:m2k_Fall}, we then find that
\eqref{eq:ddtcumforpartial} holds for the antisymmetric cumulant with $E$.
This completes the induction step and hence the proof of the antisymmetric
case.

The proof for the symmetric case is obtained by replacing $\fcum$ by $\bcum$, $
\sgn(\pi)$ by $1$, and $\p'_2$ by $\mathcal{P}$ in the above computations.
  \end{proof}

As a corollary of Lemma \ref{th:partialtosum} and the earlier moment
substitution results, Propositions \ref{prop:iteratedcum} and
\ref{prop:iteratedcum_F}, and using multilinearity of the cumulants, we now
obtain the following hierarchies for evolution of symmetric and antisymmetric
cumulants.  In the first Theorem, we assume that the rate operator $\partial$
is a polynomial, while in the next two Theorems we detail how the evolution
equation will get simplified if this is converted to a rate function using the
appropriate Wick polynomials.
Note that given one of the three rate functions, one can always construct some
coefficients for the other two rate functions, using the general defining
relations, i.e., since $b_t^{I'} = \sum_{E\subset I'} \E[b_t^E]
\bwick{b_t^{I\setminus E}}
= \sum_{E\subset I',\, E\, \text{even}} \sgn(E,I'\setminus E)\E[b_t^E] \fwick{b_t^{I\setminus E}}$.

The proofs of the following three results follow by checking that all the
assumptions for applying the relevant cumulants of moments Proposition are
valid.
\begin{theorem}
    \label{thm:cumhieR}
     Consider some fixed $t\in U$ and $I\in\I(J)$ so that Assumptions
     \ref{ass:b_time} and \ref{ass:derivative} hold.
     We denote by $\bcum[\es{b_t}{I}]$ and by $\fcum[\es{b_t}{I}]$ the
    symmetric and antisymmetric cumulants constructed with
    $\E|_{\mathcal{A}_t}$ and $b_t$.  Suppose that for each
    $(\ell,j_\ell)\in I$
    we are given a finite subset $\I_{j_\ell}$ of $\I(J)$ and
    scalars $R^{I'}_{j_\ell}(t)\in \C$, for every $I'\in\I_{j_\ell}$,
    such that
    \begin{align}
        \label{eq:nonreq_a_evo}
        \partial b_t(j_\ell) = \sum_{I'\in\I_{j_\ell}}R^{I'}_{j_\ell}(t)b^{I'}_t\,.
    \end{align}
Then, there holds
    \begin{align}
    \label{eq:cumhieRbos}
     &\ddt\bcum[\es{b_t}{I}]=
    \sum_{(\ell,j_\ell)\in I}^{}
    \sum_{I'\in\I_{j_\ell}}
    R^{I'}_{j_\ell}(t)
    \sum_{\pi\in\p(I'+I\setminus \ell)}
    \prod_{A\in \pi}
    \ind(A\cap I' \ne \emptyset)
    \bcum[\es{b_t}{P_\ell^{I',I}[A]}]\,,
    \end{align}
    where $P_\ell^{I',I}$ is as in Remark \ref{rem:insert_permutation}, namely, it is a permutation of $I'+I\setminus \ell$ that makes
    room and inserts indices in $I'$ to be at the position where $\ell$ was in
    $I$.

    If, in addition, $|I'|$ is odd for
    every $I'\in \I_{j_\ell}$, $(\ell,j_\ell)\in I$, then there also holds
    \begin{align}
    \nonumber
    & \ddt\fcum[\es{b_t}{I}]=
    \sum_{(\ell,j_\ell)\in I}^{}
    \sum_{I'\in\I_{j_\ell}}
    R^{I'}_{j_\ell}(t)
    \\ & \qquad \times
    \sum_{\pi\in\p(I'+I\setminus \ell)}
        \sgn(\pi;P_\ell^{I',I})
    \prod_{A\in \pi}
    \ind(A\cap I' \ne \emptyset)
    \fcum[\es{b_t}{P_\ell^{I',I}[A]}]\,.
    \label{eq:cumhieRfer}
    \end{align}
\end{theorem}

\begin{theorem}
    \label{thm:cumhie}
    Consider some fixed $t\in U$ and $I\in\I(J)$ so that Assumptions  \ref{ass:b_time} and \ref{ass:derivative} hold.
    We denote by $\bcum[\es{b_t}{I}]$ the
    symmetric cumulants constructed with $\E|_{\mathcal{A}_t}$ and $b_t$.
    Suppose that for each $(\ell,j_\ell)\in I$
    we are given a finite subset $\I_{j_\ell}$ of $\I(J)$ and
    scalars $M^{I'}_{j_\ell}(t)\in \C$, for every $I'\in\I_{j_\ell}$,
    such that
    \begin{align}
        \label{eq:general_a_evo}
        \partial b_t(j_\ell) = \sum_{I'\in\I_{j_\ell}}M^{I'}_{j_\ell}(t)\bwick{b^{I'}_t}\,.
    \end{align}
    Define the corresponding permutations $P_\ell^{I',I}$ as in Remark \ref{rem:insert_permutation}.
    Then, there holds
    \begin{align}
    \label{eq:cumhie}
    \ddt\bcum[\es{b_t}{I}]=
    \sum_{(\ell,j_\ell)\in I}^{}
    \sum_{I'\in\I_{j_\ell}}
    M^{I'}_{j_\ell}(t)
    \sum_{\pi\in\p(I'+I\setminus \ell)}
    \prod_{A\in \pi}
    \ind(A\not\subset I', A\not\subset (I\setminus \ell+|I'|))
    \bcum[\es{b_t}{P_\ell^{I',I}[A]}],
    \end{align}
\end{theorem}

\begin{theorem}
    \label{thm:cumhie_F}
  Consider some fixed $t\in U$ and $I\in\I(J)$ so that Assumptions  \ref{ass:b_time} and \ref{ass:derivative} hold.
    We denote by $\fcum[\es{b_t}{I}]$ the
    antisymmetric cumulants constructed with $\E|_{\mathcal{A}_t}$ and $b_t$.
    Suppose that for each $(\ell,j_\ell)\in I$
    we are given a finite subset $\I_{j_\ell}$ of $\I(J)$ and
    scalars $M^{I'}_{j_\ell}(t)\in \C$, for every $I'\in\I_{j_\ell}$.
    such that
    $|I'|$ is odd for all $I'\in\I_{j_\ell}$ and
\begin{align}
        \label{eq:general_a_evo_F}
        \partial b_t(j_\ell) = \sum_{I'\in\I_{j_\ell}}M^{I'}_{j_\ell}(t)\fwick{b^{I'}_t}\,.
    \end{align}
    Define the corresponding permutations $P_\ell^{I',I}$ as in Remark \ref{rem:insert_permutation}.
    Then, there holds
    \begin{align}
    \label{eq:cumhie_F}
   & \ddt\fcum[\es{b_t}{I}]=
    \sum_{(\ell,j_\ell)\in I}^{}
    \sum_{I'\in\I_{j_\ell}}
    M^{I'}_{j_\ell}(t)
    \sum_{\pi\in\p'_2(I'+I\setminus \ell)}
    \sgn(\pi;P_\ell^{I',I})
  \nonumber \\ & \qquad \times
     \prod_{A\in\pi}^{}
     \ind(|A| \text{ odd, or } A\not\subset I', A\not\subset (I\setminus \ell+|I'|))
     \fcum[\es{b_t}{P_\ell^{I',I}[A]}]\,.
    \end{align}
\end{theorem}

%% file: section6.tex
\section{Cumulant hierarchies for bosons and fermions}
\label{sec:hie_bos_fer}

As an application of the previous methods, 
we consider in this Section evolution of two quantum systems, interacting bosons or fermions hopping on a finite $d$-dimensional lattice. 
In addition to hopping, we allow for weak interactions between the particles, determined by a fixed potential function.  For example, such a fermionic lattice model would arise physically as a
description of a fluid of electrons evolving in a crystal background potential.  For more details about such extensions, we refer to
\cite{lukkarinen_not_2009}, for instance, Remark 2.3 therein summarizes the changes and new properties which arise for
fermions and bosons moving not on a lattice, but in the continuum, $\R^d$.
We use the earlier constructions of symmetric and antisymmetric cumulants, and their evolution hierarchies, to connect the properties of such systems to the associated kinetic evolution equations.  In particular, we show how a version of the closure assumptions commonly used in physical applications for such evolution hierarchies, leads to an evolution equation whose scaling limit, at least formally, converges to a solution of the appropriate, bosonic or fermionic, Boltzmann--Nordheim kinetic evolution equation.

For simplicity, we take the lattice to be determined by  $\lattice\subset \Z^d$, with periodic boundary conditions and common side length, $L\in\N$, $L\ge 2$, and we ignore the spin degree of freedom of the particles.  
Then, $|\lattice|=L^d$.  Again for convenience, we centre the lattice at the origin of $\Z^d$ in such a way that the \defem{dual lattice} of $\Lambda$ is given by  
$
\duallattice \coloneqq \frac{1}{L}\lattice\subset \left(-1/2,1/2\right]^d =: \T^d.
$
For any function $f$ from the lattice $\lattice$ to a vector space, we set its finite discrete Fourier transform to be the finite linear combination
\begin{align}
    \label{eq:FT}
     \FT{f}(k) \coloneqq \sum_{y\in \Lambda} \rme^{-\ci 2\pi k\cdot y} f(y)\,, \qquad k\in \Lambda^*\,.
\end{align}
In this discrete case, an inverse transformation is also possible pointwise by finite linear combination.

We use the Fock space representation to describe our quantum particle systems.
The notations and definitions follow those of \cite{cadamuro_kinetic_2018} and \cite{brattelirobinson2}.
The one-particle Hilbert space is now $\h:=\C^{\lattice}$, which leads to the Fock space
$\mathcal{F}=\bigoplus_{n=0}^\infty\hh_n$ where the 
corresponding $n$-particle sectors are $\hh_n=\h^{\otimes n} \simeq \C^{\lattice^n}$.
We use $\norm{\cdot}$ and $\innerprod{\cdot}{\cdot}$ to denote, respectively,
the corresponding $\ell^2$ norm and inner product in all of these spaces.
We use the notation where $\innerprod{\cdot}{\cdot}$ is linear in the
first argument and anti-linear in the second. Let $P_+^{(n)},P_-^{(n)}\colon
\hh_n\to\hh_n$ denote,
respectively,  the orthogonal 
projections onto the subspaces of symmetric and antisymmetric functions, namely,
\begin{align}
    \label{eq:projection}
    (P_\sigma^{(n)}\psi)(x_1,\dots,x_n)=\frac{1}{n!}\sum_{\pi\in S_n}^{}
    \varepsilon_\sigma(\pi)
    \psi(x_{\pi(1)},\dots,x_{\pi(n)}) 
    \qquad 
\end{align}
for $\sigma\in\{+,-\}$ and $\psi\in\hh_n$ and $ x_1,\dots,x_n\in\lattice$.
Here $S_n$ denotes the group of permutations of the set $\{1,2,\dots,n\}$, we define $\varepsilon_+(\pi)=1$,
and $\varepsilon_-(\pi)$ is $1$ if $\pi$ is even and $-1$ if $\pi$ is odd.
We denote the corresponding bosonic and fermionic Fock spaces by $\F_+=P_+\F$ and $\F_-=P_-\F$, respectively, where
$P_\sigma =\bigoplus_{n=0}^\infty P_\sigma^{(n)}$.

For $f\in\h$ we denote the bosonic (fermionic)
annihilation and creation operators on the bosonic (fermionic) Fock space by $a(f)$ and 
$a^*(f)$, respectively. 
We use the shorthand notations $a(x)=a(f_x)$ and $a^*(x)=a^*(f_x)$
for $x\in\lattice$ and $f_x(y)=\ind(x=y)$.
In the fermionic case $a(f)$ and $a^*(f)$ are bounded linear operators 
on the fermionic Fock space $\F_-$ while in the bosonic case they are unbounded linear operators on 
the Bosonic Fock space $\F_+$ \cite{brattelirobinson2}. To construct the cumulant hierarchies we will need 
arbitrary products of these operators. It suffices that we find a subspace $\D\subset \F$ such that 
$a(f)\D\subset \D$ and $a^*(f)\D\subset\D$.
For fermions we can take $\D=\F_-$ as the operators are bounded. For bosons we take 
$\D=\cup_{r>0}D(\rme^{rN})$, where
\begin{align}
    \label{eq:}
    D(\rme^{rN}) = \{\psi\in\F_+\colon \sum_{n=0}^\infty (\rme^{rn}\norm{\psi_n})^2<\infty)\,.
\end{align}
Here $N$ is the number operator $N\psi=(n\psi_n)_{n=0}^\infty$. Note that
\begin{align}
    \label{eq:N}
    N = \sum_{x\in\lattice} a^*(x)a(x).
\end{align}
Detailed justification for the choice of $\D$ in the bosonic case will be discussed 
in the Subsection \ref{sec:bosons} below.

The main motivation for considering such a domain $\D$ for bosons is that then the canonical commutation relations (CCR) are ordinary algebraic relations between linear operators $\D\to \D$ and, hence, arbitrary products of the creation and annihilation operators may be considered without changing the domain.  In particular,
using the fact that the vectors $f_x$, $x\in \lattice$, form an orthonormal collection, the bosonic annihilation and creation operators with the above domains satisfy the following CCR,
for any choice of $x,y\in\lattice$,
\begin{align}
    \label{eq:com}
    [a(x),a^*(y)]=\ind(x=y)\id , \quad [a(x),a(y)]=0=[a^*(x),a^*(y)]\,.
\end{align}

The fermionic creation and annihilation operators belong to the $C^*$-algebra of bounded operators, since for them we can choose $\D=\F_-$.
They
satisfy the canonical anticommutation relations (CAR), namely, we have
for all $x,y\in\lattice$,
\begin{align}
    \label{eq:com_fer}
    \{a(x),a^*(y)\}=\ind(x=y)\id, \quad \{a(x),a(y)\}=0=\{a^*(x),a^*(y)\}\,,
\end{align}
where $\{A,B\}=AB+BA$ denotes the anticommutator of two elements of the algebra.  For fermions, the two operators are related by standard adjoint of bounded operators: we have $(a(x))^*=a^*(x)$ and $(a^*(x))^*=a(x)$.

Our
Hamiltonian  $H:D(H)\to \F$ is a self adjoint operator on the Fock space $\F$ 
such that $H\D\subset \D \subset D(H)$, and 
\begin{align}
\label{eq:hamiltonian}
 H|_\D = \sum_{x,y\in \Lambda} \alpha(x-y) a^*(x) a(y) +
 \lambda \frac{1}{2} \sum_{x,y\in \Lambda;\,x\ne y} V(x-y) a^*(x) a^*(y)
 a(y) a(x)\,,
\end{align}
where $\alpha:\Lambda \to \R$ is called the \defem{harmonic hopping potential}
while $V:\Lambda \to \R$ is called the \defem{interaction potential} with a
coupling $\lambda>0$.  Periodic boundary conditions are implemented by assuming
that the difference ``$x-y$'' is computed using periodic arithmetic, i.e., it denotes
$(x-y)\bmod \Lambda$. We will 
assume that both potentials are symmetric under lattice reflections, e.g., that
$V(-y)=V(y)$ for all $y\in \Lambda$, and, for convenience, we set $V(0)=0$ so that the condition $x\ne y$ can be removed from the above sum.

We will discuss the domains of $H$, $a(f)$ and $a^*(f)$
separately for bosons and fermions in the following subsections, where we will also prove that $\rme^{-\ci t H}\D \subset \D$ for either choice of $\D$ and any $t\in \R$.  However, we recall that the above Hamiltonians have $\hh_n$ as an invariant subspace, corresponding to the projection of the $n$-particle Hamiltonian defined by
\begin{align}\label{eq:Hndef}
 (H_n\psi)(x_1,\ldots,x_n) = \sum_{j=1}^n \sum_{y\in \lattice} \alpha(x_j-y)\psi(x|_{x_j\to y}) + \frac{\lambda}{2} \sum_{j,j'=1}^n V(x_j-x_{j'})\psi(x)\,.
\end{align}
This is a bounded operator on $\hh_n$, and we have
$ H=\bigoplus_{n=0}^\infty H_n$, both on the complete Fock space, as well as on the subspaces $\F_\pm$.  $D(H)$ contains all $\psi=(\psi_n)_n\in \F$ such that $\sum_{n=0}^\infty \norm{H_n \psi_n}^2<\infty$.
More details about these statements can be found from \cite{cadamuro_kinetic_2018} and \cite{brattelirobinson2}.

We are interested in the time evolved annihilation and creation operators
and their Fourier transforms
\begin{align}
    \label{eq:a}
    a_t(x,-)&:=e^{\ci tH}a(x)e^{-\ci tH},
    \qquad
    a_t(x,+):=e^{\ci tH}a^*(x)e^{-\ci tH},
    \\
    \label{eq:aft}
    \fta_t(k,-)&:=\sum_{x\in\lattice} e^{-i2\pi x\cdot k}a_t(x,-),
    \qquad
    \fta_t(k,+):=\sum_{x\in\lattice} e^{-i2\pi x\cdot k}a_t(x,+)
\end{align}
as functions $\D\to \D$ for each  $x\in\lattice$, $k\in\duallattice$,   and $t\in\R$. 
By assumption and since $\D$ is a vector space, linear combinations and products of these operators can be understood in the usual sense of operators on a fixed vector space. 
We point out that then also the time-evolved operators will satisfy the appropriate CCR or CAR relations.  For example, for bosonic operators $a_t(k,\pm)$ we know that
\begin{align}
    \label{eq:comk}
    & [a_t(k,-),a_t(k',+)]=|\lattice| \ind(k+k'=0 \bmod\duallattice),\
    \nonumber \\ &
    [a_t(k,-),a_t(k',-)]=0=[a_t(x,+),a_t(y,+)],
\end{align}
for all $k,k'\in\duallattice$ and $t\in\R$, while
the corresponding fermionic operators satisfy
\begin{align}
    \label{eq:comk_fer}
   & \{a_t(k,-),a_t(k',+)\}=|\lattice| \ind(k+k'=0 \bmod\duallattice),
    \nonumber \\ &
    \{a_t(k,-),a_t(k',-)\}=0=\{a_t(x,+),a_t(y,+)\}\,.
\end{align}
Based on these relations, it is now straightforward to check
that $[H,N]\psi=0$ for all $\psi\in \D$, i.e., the number operator
$N|_\D$ commutes with the Hamiltonian in both fermionic and bosonic case.

We say that a function $\Psi\colon \R\to \D$ is differentiable if the derivative 
\begin{align}
    \ddt  \Psi_t =\lim_{h \to 0}  (\Psi_{t+h}-\Psi_t )/h
\end{align}
exists in $\D$ for all $t\in\R$, using the topology inherited from the appropriate Fock space. If for each $t\in\R$, $A_t$ is an operator on $\D$, 
we say that $t\mapsto A_t$ is (strongly) differentiable if $t\mapsto A_t\psi$ is differentiable for all $\psi\in \D$. 
If $t\mapsto A_t$ is differentiable, we define $\ddt A_t:\D\to\D$ by 
\begin{align}
    (\ddt A_t) \psi =  
    \ddt (A_t \psi)\,,
\end{align}
which implies that $\ddt A_t$ is a linear map, i.e., an operator on $\D$.
Note that $\ddt$ satisfies the product rule, i.e. $\ddt(A_tB_t)=(\ddt A_t)B_t+A_t\ddt B_t$.

We will now compute the time derivatives of $a_t(x,\sigma)$ and $\fta_t(k,\sigma)$. Since $\D\subset D(H)$, the limit defining the derivative of $t\mapsto\rme^{\pm\ci t H}$ exists in the Fock space, and 
$\ddt (\rme^{\pm\ci t H}\psi)= \pm\ci H\rme^{\pm\ci t H}\psi$ for $\psi\in\D$.
Since $H\D \subset \D$, $a(x)\D \subset \D$ and $\rme^{\pm \ci t H}\D \subset \D$, we find that $t\mapsto a_t(x,-)$ is differentiable and the 
product rule implies $\ddt a_t(x,-)=\ci \rme^{\ci tH}[H,a(x)]\rme^{-\ci tH}$. 
By a similar argument, we get $\ddt a_t(x,+)=\ci \rme^{\ci tH}[H,a^*(x)]\rme^{-\ci tH}$.
Computing the commutators and using the CCR algebra properties for bosons  \eqref{eq:com} and CAR algebra properties for fermions \eqref{eq:com_fer}, yields for either case the expression
\begin{align}
    \label{eq:atxevol}
    &\ddt a_t(x,\sigma) = \ci \sigma \sum_{y\in \Lambda} \alpha(x-y)
    a_t(y,\sigma) +\ci\sigma \lambda  \sum_{y\in \Lambda} V(x-y)
 \begin{cases}
  a_t(y,+) a_t(y,-)a_t(x,-)\,, & \sigma=- \\
  a_t(x,+) a_t(y,+)a_t(y,-)\,, & \sigma=+
 \end{cases}
 \,.
 \end{align}
 Taking Fourier transform implies
 \begin{align}
    \nonumber
    &\ddt \fta_t(k,\sigma)=\ci\sigma\omega(k)\fta_t(k,\sigma)
    \\
    &+
    \label{eq:FTa_evol}
    \ci\sigma\lambda
    \int_{(\duallattice)^3}dk_1dk_2dk_3
    \delta_{L}(k-k_1-k_2-k_3)
    \FT{V}(k_2+k_{2+\sigma1})
    \fta_t(k_1,+)\fta_t(k_2,\sigma)\fta_t(k_3,-)\,,
\end{align}
where $\delta_L(k)=|\lattice|\ind(k{ = }0 \mod\duallattice)$ and $\omega=
\FT{\alpha}$ is the \defem{dispersion relation} of the interaction $\alpha$.
Here $\int_\duallattice\rmd k$ is a shorthand notation for
$|\duallattice|^{-1}\sum_{k\in\duallattice}$.
Note that by the assumptions on $\alpha$ and $V$, also both $\omega$ and
$\FT{V}$ are real 
valued and symmetric functions on $\duallattice$.

A general discussion about states of CCR and CAR algebras can be found from \cite{brattelirobinson2}.
Here, we consider a specific case of a tracial state on the appropriate Fock space, i.e., 
we assume that the initial state $\mean{\cdot}$ is given by 
a density operator:  there is a non-negative, trace class operator $\rho$ on the Fock space which satisfies $\tr \rho=1$, and $\mean{A}:=\tr(A\rho)$ for all bounded operators $A$ on the Fock space.
Diagonalizing, we thus can find an orthonormal sequence $(\varphi^{(i)})_i^\infty$ of 
vectors in the appropriate Fock space and real numbers $p_i\ge0$, $\sum_{i=1}^\infty p_i=1$ such that 
\begin{align}
    \label{eq:state}
    \mean{A} = \sum_{i=1}^\infty p_i \innerprod{A\varphi^{(i)}}{\varphi^{(i)}}
\end{align}
for all bounded operators $A$ on the Fock space $\F_\pm$.

After these common preliminaries,
we will continue the discussion separately for  fermions in Subsection
\ref{sec:fermions} and cover bosons in Subsection \ref{sec:bosons}.  
% We will do this in the case when the initial state is both translation and gauge invariant. 

\subsection{Fermions}
\label{sec:fermions}

We will now focus on the fermionic case.   
The fermionic annihilation and creation operators are bounded operators on the fermionic Fock space $\F_-$, and thus we can take $\D=\F_-$, as explained earlier, and this obviously is an invariant subspace for $\rme^{-\ci t H}$. More precisely, in the fermionc case we have 
\begin{align}
    \label{eq:a_bounded}
    \norm{a^\#(f)\psi} 
    \leq 
    \norm{f} 
    \norm{\psi} 
\end{align}
for all $f\in\h$ and $\psi\in\F_-$ , where $a^\#(f)$ denotes either $a(f)$ or $a^*(f)$.
Hence, $\norm{a(x,\sigma)}\le 1$ for $x\in \Lambda$, $\sigma = \pm$.
It follows that $H$ is a bounded, self-adjoint operator on the Fock space, and thus the corresponding semigroup $\rme^{-\ci t H}$ is unitary, and commutes
with $N$.  Therefore, $\norm{a_t(x,\sigma)}\le 1$ for all $t\in \R$, and thus $\norm{\FT{a}_t(k,\sigma)}\le |\Lambda|$ for all allowed values.

The idea is to use the antisymmetric cumulant hierarchy \eqref{eq:cumhieRfer}
separately for either of the fields $a$ and $\fta$.  Therefore the first
question is about Assumption \ref{ass:b_time}, namely, what are the objects $U,
  \A_t, \E$ and $b_t:J\to\A_t$ in these two cases.  
 In both cases we take $U=\R$ and for each $t\in\R$ let $\A_t=\A$, where $\A$ is the complex unital algebra of bounded linear operators on $\F_-$. We denote the null and identity elements of 
$\A$ by $0$ and $\id=\text{id}_\D$, respectively.
Each $a_t(x,\sigma)$ and $\fta_t(k,\sigma)$ belong to $\A$ as 
bounded operators.
We use the tracial state \eqref{eq:state}, to define $\E$, namely, we
set $\E[A]=\mean{A}=\tr(\rho A)$ for all $A\in\cup_t\A_t=\A$.
Then, the map $A\mapsto \E[A]$ is linear on the vector space $\A$ and
satisfies $\E[\id]=1$.  Hence, $\E$ is an expectation map on
$\A$.  

In the case of the field $a$ we will use $b_t=a_t\colon \lattice\times\{\pm\}\to\A$ and for 
the field $\fta$ we will take $b_t=\fta_t\colon \duallattice\times\{\pm\}\to\A$.
We have now checked that these definitions yields a collection which
satisfies Assumption \ref{ass:b_time}.

The main reason to consider antisymmetric cumulants  of the fields $a$ and
$\fta$ for fermions is that these operators satisfy CAR
relations and, hence, we can apply  Proposition
\ref{prop:cum_per_invariance_F} to their \defem{antisymmetric} cumulants
which are now permutation invariant up to a sign, apart from second order
cumulants.  We could equally well study the symmetric cumulants of these operators but no
such simplifying property is expected for these.  In addition, if the quantum
evolution leads to states whose number distribution is close to the
Fermi--Dirac distribution, as would be expected for fermions near equilbrium,
then this condition can be characterized by assuming $\fcum[\es{a_t}{I}]\approx 0$ if
$I\in\I(\lattice\times\{\pm\})$, $|I|>2$.  Meanwhile, then the corresponding symmetric cumulants would not be
small.  These properties are analogous to standard cumulants of  real-valued
random variables which are zero from order $3$ onwards if and only if the joint
distribution of the random variables is Gaussian.

Indeed, one important goal for an analysis of the cumulant hierarchy would be
to study the equilibriation of fermionic system.  We motivate this by pointing
out a connection of these hierarchies to quantum kinetic theory in Section\
\ref{sec:derivation}.

Finally, we need to show that the time-evolution of these cumulants also
satisfies the requirements of Assumption \ref{ass:derivative}.
Note that for $A_0\in\A$ and $t\in\R$ we have 
$A_t=\rme^{\ci t H}A_0\rme^{-\ci t H}\in \A$ and $\ddt A_t = \ci [H,A_t]$.
For every $t\in \R$ and $A\in\A$ we define
%\[
% \partial A := \ci \rme^{\ci t H}[H, A_0]\rme^{-\ci t H}
% = \dds \rme^{\ci s H}A_0\rme^{-\ci s H}\big|_{s=t}
% \,.
%\]
\[
 \partial A := \ci [H, A]
 \,.
\]
Note that, since $H=H|_\D\in \A$, also $[H,A]\in \A$ and thus the above
definition yields a map $\partial :\A \to \A$.  If $A,B\in  \A$, we
have $[H,A B]=[H,A] B + A [H,B]$, and it follows that $\partial(A
B) = \partial(A) B+A \partial(B)$ for all $A,B\in \A$.  Hence, the required
product rule is satisfied by $\partial$ in the entire algebra.

Now, take any non-empty
$I\in\I(\lattice\times\{\pm\})$.
Given any $(x,\sigma)\in \lattice\times\{\pm\}$, we already know that
 the map $t\mapsto a_t(x,\sigma)$ is differentiable.
Note that $\ddt a_t(x,\sigma) = \partial a_t(x,\sigma)$.
  The right hand side of \eqref{eq:atxevol} is differentiable by the product
  rule, so we find that the map $t\mapsto a_t(x,\sigma)\psi$ is smooth, in
  particular, it is continuously differentiable.  Therefore, by the product 
  rule, also the map $t\mapsto a_t^I \psi$
    is smooth and satisfies
\[
 \ddt(a_t^I\psi) = \sum_{(\ell,(x_\ell,\sigma_\ell))\in  I} a_t^{I|_{<\ell}}\partial (a_t(x_\ell,\sigma_\ell))  a_t^{I|_{>\ell}}\psi
 = \partial(a_t^I)\psi\,.
\]
It follows that
\[
 \innerprod{\psi}{a_t^I\psi}
 = \innerprod{\psi}{a_0^I\psi}
 + \int_0^t\! \rmd s\,\innerprod{\psi}{\partial(a_s^I)\psi}\,,
\]
where the integrand is bounded by $\norm{\psi}^2|I| (\norm{\alpha}_1+\lambda
\norm{V}_1)$.  This can be used to justify relying on Fubini's theorem in
\eqref{eq:state}, and to conclude that then
\[
 \mean{a_t^I} =\mean{a_0^I} + \int_0^t\! \rmd s\,  \mean{\partial(a_s^I)}\,,
\]
where the integrand is a continuous function of $s$.
Hence, we can now 
conclude that the map $t\mapsto \E[a_t^I]=\mean{a_t^I}$ is continuously differentiable and
satisfies $\ddt \E[a_t^I] = \E[\partial(a_t^I)]$ for any $I\in
\I(\lattice\times\{\pm\})$.  

We conclude that this construction of $\partial$ and $\E$ satisfies Assumption
\ref{ass:derivative} for any $t\in \R$ and non-empty $I \in
\I(\lattice\times\{\pm\})$.  In addition, choosing any  $j =(x,\sigma)\in \lattice\times\{\pm\}$,
we have \[
  \partial a_t(j) =
  \sum_{I'\in\I_{j}}R^{I'}_{j}(t)b^{I'}_t\,
\]
with $\I_{(x,\sigma)}:=\{I^1(y;x,\sigma),I^3(y;x,\sigma)\}_{y\in \lattice}$,
where 
\begin{align}
\label{eq:I1}
I^1(y;x,\pm)&=\{(1,(y,\pm))\}\,,  
\\
\label{eq:I3plus}
I^3(y;x,-)&=\{(1,(y,+)),(2,(y,-)),(3,(x,-))\}\,,
\\
\label{eq:I3minus}
  I^3(y;x,)&=\{(1,(x,+)),(2,(y,+)),(3,(y,-))\}\,.
\end{align}
  The respective coefficients are $R^{I^1(y;x,\sigma)}_{(x,\sigma)}(t)
  := \ci \sigma \alpha(x-y)$ and  $R^{I^3(y;x,\sigma)}_{(x,\sigma)}(t)
  := \ci \sigma \lambda V(x-y)$.
Therefore, also the assumption
\eqref{eq:nonreq_a_evo} is satisfied for this construction, and every $I'\in
\I_{(x,\sigma)}$ is odd.  By Theorem \ref{thm:cumhieR}, the 
antisymmetic cumulants $\fcum[\es{a_t}{I}]$,
$I\in\I(\lattice\times\{\pm\})$ satisfy \eqref{eq:cumhieRfer}.

Similarly, for $\fta_t(k,\sigma)$, $(k,\sigma)\in \duallattice\times\{\pm\}$, we have by \eqref{eq:FTa_evol}
that there are such $R^{I'}_{j_\ell}(t)$ and
$\I_{j_\ell}$ that
each $I'\in \I_{(x,\sigma)}$ is odd  and \eqref{eq:nonreq_a_evo} holds.  The
related derivatives still correspond to smooth functions in time.  Hence,
repeating the above argument, also the antisymmetric cumulants $\fcum[\es{\fta_t}{I}]$,
$I\in\I(\duallattice\times\{\pm\})$, satisfy an evolution hiearchy, as determined in
Theorem \ref{thm:cumhieR}.
Let us collect the above results into a Theorem containing also the explicit
forms of the two hierarchies.  
\begin{theorem}
    \label{thm:fermion_hie}
   Consider a tracial state satisfying \eqref{eq:state}, and use
   $\E[\cdot]=     \mean{\cdot}$ to construct antisymmetric cumulants.
    Then both
    $(\fcum[\es{a_t}{I}])_{I\in\I(\lattice\times\{\pm\})}$ and
    $(\fcum[\es{\fta_t}{I}])_{I\in\I(\duallattice\times\{\pm\})}$ satisfy the
    antisymmetric cumulant hierarchy \eqref{eq:cumhie_F}.  Explicitly, then
    for $t\in\R$ and any  $I\in\I(\lattice\times\{\pm\})$ there holds
    \begin{align}
        \label{eq:ferhie}
        &\partial_t
        \fcum[\es{a_t}{I}]
        =
        \sum_{(\ell,(x,\sigma))\in I}^{}
        \ci \sigma \sum_{y\in \Lambda}
        \Bigg(
        \alpha(x-y)
        \fcum[\es{a_t}{I^{\ell}_{x\to y}}]
        +
        \lambda
         V(x-y)
        \nonumber \\&
        \qquad \quad \times
    \sum_{\pi\in\p'_2(I'+I\setminus \ell)}
        \sgn(\pi; P^{I',I}_\ell)
        \prod_{A\in \pi}
        \ind(A\cap   I'\ne \emptyset)
        \fcum[\es{a_t}{P_\ell^{I',I}(A)}]
       \Biggr|_{I'= I^3(y;x,\sigma)}  \Bigg)
        \,,
    \end{align}
    where  $P_\ell^{I',I}$ is as in Remark \ref{rem:insert_permutation},
    $I^{\ell}_{x\to y} := I\setminus \{(\ell,(x,\sigma))\}\cup \{(\ell,(y,\sigma))\}$, and
    $I^3(y;x,\sigma)$ is defined in \eqref{eq:I3plus}--\eqref{eq:I3minus}.

    In addition,  for $t\in\R$ and any  $I\in\I(\duallattice\times\{\pm\})$ there holds
    \begin{align}
        \label{eq:ferfthie}
        &\partial_t\fcum[\es{\fta_t}{I}]=
        \ci\Omega(I)
        \fcum[\es{\fta_t}{I}]
+  \ci \lambda
        \sum_{(\ell,(k,\sigma))\in I}^{}
        \int_{(\duallattice)^3}
        \rmd k_{1}
        \rmd k_{2}
        \rmd k_{3}
        \delta_{L}(k-k_1-k_2-k_3)
        \sigma \FT{V}(k_2+k_{2+\sigma })
        \nonumber \\&
        \qquad \quad \times
       \sum_{\pi\in\p'_2(I'+I\setminus \ell)}
        \sgn(\pi; P^{I',I}_\ell)
        \prod_{A\in \pi}
        \ind(A\cap   I'\ne \emptyset)
        \fcum[\es{a_t}{P_\ell^{I',I}(A)}]
       \Biggr|_{I'= \hat I^3(k;\sigma)}
    \end{align}
   where $\Omega(I):=\sum_{(\ell,(k,\sigma))\in I}^{}
        \sigma \omega(k)$, and
   $\hat I^3(k;\sigma):=\{(1,(k_1,+)),(2,(k_2,\sigma)),(3,(k_3,-))\}$.
\end{theorem}
\begin{remark}
    Note that the rate function on the right hand side of both \eqref{eq:ferhie} and \eqref{eq:ferfthie} is at most a third order polynomial in cumulants since in both case $|I'|=3$, independently of $I$.  The term containing the longest cumulant is obtained when $|\pi|=1$ which gives a term proportional to $\fcum[\es{a_t}{P_\ell^{I',I}(I'+I\setminus \ell)}]$ where the length of the sequence is $|I|+2$.
\end{remark}

It is possible to use a Wick polynomial regularization to simplify  $\partial a_t$ and, hence, the evolution hierarchy.  By definition,
\[
 b_1= \fwick{b_1} \,, \quad
 b_1 b_2 b_3 = \fwick{b_1 b_2 b_3}
 + \E[b_1 b_2] \fwick{b_3}
 -  \E[b_1 b_3] \fwick{b_2}
 + \E[b_2 b_3] \fwick{b_1}\,,
\]
where $\E[b_1 b_2]=\fcum[b_1,b_2]$.  With the shorthand notation
$C_{\sigma\sigma'}(x,x';t) := \fcum[a_t(x,\sigma),a_t(x',\sigma')]$, we thus have
\[
  \partial b_t(x,\sigma) =
  \sum_{I'\in\I_{(x,\sigma)}}M^{I'}_{(x,\sigma)}(t)\fwick{b^{I'}_t}\,
\]
with otherwise same definitions of $I'\in \I_{(x,\sigma)}$ as before, except now we need to include also sets $I^1(y;x,-\sigma)=\{(1,(y,-\sigma))\}$, $y\in \lattice$, in $\I_{(x,\sigma)}$.
Then we set
$M^{I^3(y;x,\sigma)}_{(x,\sigma)}(t):= \ci \sigma \lambda V(x-y)$, as before, but change the first order terms to $M^{I^1(y;x,\sigma')}_{(x,\sigma)}(t):=\ci \sigma \tilde{\alpha}_{\sigma\sigma'}(x,y;t)$, where
\begin{align*}
 &\tilde{\alpha}_{\sigma\sigma}(x,y;t) :=
\alpha(x-y) + \ind(x=y) \lambda \sum_{x'\in \lattice} V(x-x') C_{+-}(x',x';t)
\\ & \qquad
 - \lambda V(x-y) \times \begin{cases}
   C_{+-}(y,x;t)
  \,, & \sigma=-\,, \\
  C_{+-}(x,y;t)\,, & \sigma=+\,,
 \end{cases}
\end{align*}
and
\begin{align*}
 &\tilde{\alpha}_{\sigma,-\sigma}(x,y;t) :=
\lambda V(x-y) \times \begin{cases}
   C_{--}(y,x;t)
  \,, & \sigma=-\,, \\
  C_{++}(x,y;t)\,, & \sigma=+\,.
 \end{cases}
\end{align*}

With these definitions, we can simplify the corresponding hierarchy as follows: By Theorem \ref{thm:cumhie_F}, we then have for the two lowest orders
 \begin{align*}
& \partial_t \E[a_t(x,\sigma)]
 = \ci \sigma
 \sum_{y\in \Lambda, \sigma' = \pm}
        \tilde{\alpha}_{\sigma\sigma'}(x,y;t)\E[a_t(y,\sigma')]
        + \ci  \lambda \sigma \sum_{y\in \Lambda}
        V(x-y)  \fcum[\es{a_t}{I^3(y;x,\sigma)}]
        \,,
\\
 & \partial_t C_{\sigma\sigma'}(x,x';t)
 = \ci \sigma
 \sum_{y\in \Lambda, \sigma'' = \pm}
        \tilde{\alpha}_{\sigma\sigma''}(x,y;t)
        C_{\sigma''\sigma'}(y,x';t)
 + \ci \sigma'
 \sum_{y\in \Lambda, \sigma'' = \pm}
        \tilde{\alpha}_{\sigma'\sigma''}(x',y;t)
        C_{\sigma\sigma''}(x,y;t)
 \\ & \quad
 +   \ci  \lambda \sum_{y\in \Lambda}
         \left(\sigma V(x-y)  \fcum[\es{a_t}{I^3(y;x,\sigma)+(x',\sigma')}]
          + \sigma' V(x'-y)\fcum[\es{a_t}{(x,\sigma)+I^3(y;x',\sigma')}]\right)
 \end{align*}
 where the second order equation is quadratic in $C$ due to definition of $\tilde{\alpha}$.
For $|I|>2$, we have
 \begin{align}
        \label{eq:ferhieWick}
        &\partial_t
        \fcum[\es{a_t}{I}]
        =
        \sum_{(\ell,(x,\sigma))\in I}^{}
        \ci \sigma \sum_{y\in \Lambda, \sigma' = \pm}
        \tilde{\alpha}_{\sigma\sigma'}(x,y;t)
        \fcum[\es{a_t}{I^{\ell}_{(x,\sigma)\to (y,\sigma')}}]
                \nonumber \\&
        \qquad
        + \sum_{(\ell,(x,\sigma))\in I}^{}
        \ci \sigma \lambda \sum_{y\in \Lambda}
         V(x-y) \sum_{\pi\in\p'_2(I'+I\setminus \ell)}
        \sgn(\pi; P^{I',I}_\ell)
        \nonumber \\&
        \qquad \quad \times
          \prod_{A\in \pi}
        \ind(|A| \text{ odd, or } A \not \subset   I',\, A\not \subset I\setminus \ell +|I'|)
        \fcum[\es{a_t}{P_\ell^{I',I}(A)}]
       \Biggr|_{I'= I^3(y;x,\sigma)}
        \,.
    \end{align}
If $|I|$ is even, we have also $|I'+I\setminus \ell|$ even, and hence none of the clusters in a partition in $\pi\in \p'_2(I'+I\setminus \ell)$ can be odd.  In this case, it it thus possible to replace the characteristic function above by a simpler form, by $\ind(A \not \subset   I',\, A\not \subset I\setminus \ell +|I'|)$.

Analogous manipulations are possible for the hierarchy with $I\in\I(\duallattice\times\{\pm\})$ which will be studied in more detail later, in Section \ref{sec:derivation}.

\subsection{Bosons}
\label{sec:bosons}

We will now focus on the bosonic case. The idea is to use the symmetric cumulant hierarchy \eqref{eq:cumhieRbos}
separately for either of the fields $a$ and $\fta$.
The bosonic annihilation and creation operators are unbounded, and we have
\begin{align}
    \label{eq:a_unbounded}
    \norm{a^\#(f)\psi_n} 
    \leq
     (n+1)^{1/2}
    \norm{f}
    \norm{\psi_n}
\end{align}
for all $\psi_n\in \hh_n$ and $f\in\h$, where $a^\#(f)$ denotes either $a(f)$
or $a^*(f)$ and we have identified $\psi_n$ with $\psi'\in\F$,
$\psi'_m=\ind(n=m)\psi_n$ \cite{brattelirobinson2}.
If considered to be unbounded operators on $\F_+$, the natural domains for their products will depend on how many and which operators are in the product, and one has to worry about shrinking the domain too much, e.g., into a zero subspace.  Hence, there is no straightforward way of using our earlier algebraic computations in this setup.

To alleviated this problem, we consider these operators on the domain
$\D=\bigcup_{r>0}D(\rme^{rN})$.  Note that this set is indeed a vector subspace of $\F_+$.
Now, if $\psi_n\in \hh_n$, $x\in \lattice$, $\sigma\in\{\pm\}$,
we have that
\[
\norm{a(x,\sigma) \psi_n}^2 \le (n+1) \norm{\psi_n}^2\,,
\]
and thus also for any $I\in \I(\lattice\times\{\pm\})$, with $M=|I|$,
and any $n\ge 0$, $s>0$,
\begin{align}\label{eq:bospsinbound}
 &
\norm{a^I \psi_n}^2 \le (n+1)^{M} \norm{\psi_n}^2\le M! s^{-M} \rme^{sn+s}\norm{\psi_n}^2\,.
\end{align}
Hence, assuming that $\psi\in D(\rme^{r N})$, $r>0$, we find that $a^I\psi\in D(\rme^{r' N})$ for any $r'<r$. Therefore, $a^I\psi\in \D$
for all $\psi\in \D$.
Since $\rme^{-\ci t H} \psi_n=\rme^{-\ci t H_n} \psi_n$, we
have $\norm{\rme^{r N} \rme^{-\ci t H} \psi}^2
=\sum_{n=0}^\infty \rme^{2rn} \norm{\rme^{-\ci t H} \psi_n}^2=\norm{\rme^{r N} \psi}^2 $, for any $r>0$, and thus $\rme^{-\ci t H}\psi\in \D$ for all $\psi\in \D$, $t\in \R$. Therefore, also $a_t^I\psi\in \D$,
for all $\psi\in \D$ and any $I\in \I(\lattice\times\{\pm\})$, $t\in\R$.
Then, $\fta_t^I\psi\in \D$,
for all $\psi\in \D$ and any $I\in \I(\duallattice\times\{\pm\})$.
We also now obtain the following bound from \eqref{eq:bospsinbound},
valid for any $\psi\in D(\rme^{r N})$ and $0<s<r$,
\begin{align}\label{eq:bospsinboundfull}
 & \norm{a_t^I\psi} = \norm{a_0^I\rme^{-\ci t H}\psi}
 \le \sum_{n=0}^\infty \norm{a_0^I\rme^{-\ci t H_n}\psi_n}
 \le C_{|I|,s} \sum_{n=0}^\infty \rme^{s n}\norm{\psi_n}
 \le \frac{C_{|I|,s}}{\sqrt{1-\rme^{-2(r-s)}}} \norm{\rme^{r N} \psi}
\end{align}
where $C_{|I|,s}$ depends only on $s$ and the length of $I$.

To satisfy Assumption \ref{ass:b_time}, we take
$U=\R$, and define $\A_t$ as the 
complex unital subalgebra of the algebra of linear operators on $\D$,
generated by the collection $\{a_t(x,\sigma)\colon x\in\lattice,
\sigma\in\{\pm\}\}$.
Each $\fta_t(k,\sigma)$ belong to $\A_t$ as an
linear combination of the operators $a_t(x,\sigma)$. Note that by \eqref{eq:a_unbounded} the algebra $\A_t$ contains unbounded
operators. 
Since $\rme^{\ci t H}$ is the inverse operator of
$\rme^{-\ci t H}$, for each $t\in\R$ the map $A_0\mapsto \rme^{\ci t H}A_0\rme^{-\ci t H}$
is an algebra isomorphism from $\A_0$ to $\A_t$ with the inverse map $A_t\mapsto
\rme^{-\ci t H}A_t\rme^{\ci t H}$.

Also for bosons we use a tracial state as in \eqref{eq:state}
and aim to set $\E[A]=\tr(\rho A)$,
but now we need to add some requirements to make sure that
$\E[a_t^I]$ are all well-defined.  In this subsection, we assume that the state has \defem{exponential moments} for the number operator: we assume there is $r>0$ such that $\tr(\rho \rme^{2 r N})<\infty$.  By positivity, then
for any $i$ we have
\[
\norm{\rme^{r N}\varphi^{(i)}}^2 = \innerprod{\rme^{2rN}\varphi^{(i)}}{\varphi^{(i)}} <\infty\,,
\]
and, therefore, we conclude that each eigenvector of the density operator satisfies $\varphi^{(i)}\in D(\rme^{rN})\subset \D$.
Moreover, choosing $s=\frac{r}{2}$ in \eqref{eq:bospsinboundfull}, we obtain a constant $C$ which depends only on $|I|$ and $r$, such that
\begin{align*}
 \norm{a_t^I\varphi^{(i)}}
 \le C \norm{\rme^{r N}\varphi^{(i)}}\,.
\end{align*}
Then,
\[
 \sum_{n=0}^\infty p_i | \innerprod{a_t^I\varphi^{(i)}}{\varphi^{(i)}}|
 \le C \sum_{n=0}^\infty p_i \norm{\varphi^{(i)}}\norm{\rme^{r N}\varphi^{(i)}}
 \le C \sqrt{\tr \rho} \sqrt{\tr (\rme^{2rN}\rho)}<\infty\,.
\]
Hence, we can define,
\[
   \E[A] := \sum_{i=1}^\infty p_i \innerprod{A\varphi^{(i)}}{\varphi^{(i)}}\,,
   \qquad A\in \bigcup_{t\in\R}\A_t\,, \]
as an absolutely convergent sum
since any such $A$ is a finite linear combination of terms of the form
$a_t^{I}$, $I\in \I(\lattice\times\{\pm\})$.  We also record the time-uniform bound
 $|\E[a_t^I] | \le C(|I|,r)$ for $I\in \I(\lattice\times\{\pm\})$, $t\in \R$.

Then, the map $A\mapsto \E[A]$ is linear on the vector space $\A_t$ and satisfies $\E[\id]=1$.  Hence, each $\E|_{\A_t}$ is an expectation map on $\A_t$. In the case of the field $a$ we will use $b_t=a_t\colon \lattice\times\{\pm\}\to\A_t$ and for 
the field $\fta$ we will take $b_t=\fta_t\colon \duallattice\times\{\pm\}\to\A_t$.
We have now checked that these definitions yields a collection which satisfies Assumption \ref{ass:b_time}.

Now we will consider the corresponding \defem{symmetric} cumulants and Wick polynomials, for same two different cases as for fermions: we either consider the field $a$ or the Fourier transformed field $\fta$.
For bosons, symmetric cumulants are appropriate, since then the annihilation and creation operators  satisfy
CCR and we can apply  Proposition \ref{prop:cum_per_invariance},
and these cumulants are
permutation invariant, apart from the second order.
In addition, if the quantum evolution leads to states whose number distribution is close to the Bose--Einstein distribution, as would be expected for bosons near equilbrium,
then this condition can be characterized by assuming $\bcum[\es{a_t}{I}]\approx 0$ if $I\in\I(\lattice\times\{\pm\})$, $|I|>2$, analogously to Gaussian distributions for normal random variables.

Finally, we need to show that also the time-evolution of these cumulants satisfies the requirements of Assumption \ref{ass:derivative}.
Given $t\in \R$ and $A\in \A_t$,  let $A_0 =\rme^{-\ci t H}A\rme^{\ci t H} \in \A_{0}$,
we define
\[
 \partial A := \ci \rme^{\ci t H}[H, A_0]\rme^{-\ci t H}
 = \dds \rme^{\ci s H}A_0\rme^{-\ci s H}\bigg|_{s=t}
 \,,
\]
Note that, since $H|_\D \in \A_0$, also $[H,A_0]\in \A_0$ and thus the above
definition yields a map $\partial :\A_t \to \A_t$,
which also satisfies
$\partial(A B) = \partial(A) B+A \partial(B)$ for all $A,B\in \A_t$.
Repeating the argument used in the fermionic case, we also conclude that $\ddt a_t(x,\sigma) = \partial a_t(x,\sigma)$ and it is given by \eqref{eq:atxevol}.
Therefore, the map $t\mapsto a_t(x,\sigma)\psi$ is smooth for any fixed $\psi\in \D$ and
for any non-empty $I\in\I(\lattice\times\{\pm\})$ we have 
\[
 \innerprod{\psi}{a_t^I\psi}
 = \innerprod{\psi}{a_0^I\psi}
 + \int_0^t\! \rmd s\,\innerprod{\psi}{\partial(a_s^I)\psi}\,.
\]
If $\psi = \varphi^{(i)}$, we can find $C$ which only depends on $|I|$, $r$, $\lattice$, $\alpha$ and $\lambda V$, such that
\[
|\innerprod{\psi}{\partial(a_s^I)\psi}|\le C
\norm{\varphi^{(i)}} \norm{\rme^{r N}\varphi^{(i)}}\,.
\]
By a similar Cauchy--Scharz estimate as above, we can then justify using Fubini's theorem to conclude
\[
 \E[a_t^I] =\E[a_0^I] + \int_0^t\! \rmd s\,  \E[\partial(a_s^I)]\,,
\]
where the integrand is absolutely continuous.

Thus the map $t\mapsto \E[a_t^I]$ is continuously differentiable and satisfies $\ddt \E[a_t^I] = \E[\partial(a_t^I)]$ for all non-empty $I\in \I(\lattice\times\{\pm\})$.
We conclude that this construction of $\partial$ and $\E$ satisfies Assumption \ref{ass:derivative} for any $t\in \R$ and non-empty $I \in \I(\lattice\times\{\pm\})$.
In addition, with $j =(x,\sigma)$, we have
\[
  \partial a_t(j) =
  \sum_{I'\in\I_{j}}R^{I'}_{j}(t)a^{I'}_t\,
\]
using the same choices of $\I_{(x,\sigma)}$ and $R^{I'}_{j}(t)$ as for fermions.
By Theorem \ref{thm:cumhieR},
the corresponding symmetric cumulants $(\bcum[\es{a_t}{I}])_{I\in\I(\lattice\times\{\pm\})}$ satisfy \eqref{eq:cumhieRbos}.
The case of symmetric cumulants $(\bcum[\es{\fta_t}{I}])_{I\in\I(\duallattice\times\{\pm\})}$ of the Fourier transformed field $\fta$
 is treated similarly, and also they satisfy an evolution hierarchy, as determined in Theorem \ref{thm:cumhieR}.

These results are summarized in the following Theorem.

\begin{theorem}
    \label{thm:boson_hie}
    Consider a tracial state which has exponential moments:
    there is $r>0$ such that $\tr (\rme^{2 r N}\rho)<\infty$.
   Then the series in \eqref{eq:state} converges absolutely for all $A\in \cup_{t\in\R} \A_t$, and we use it to define the map $\E$.
   The corresponding symmetric cumulants
    $(\bcum[\es{a_t}{I}])_{I\in\I(\lattice\times\{\pm\})}$ and
    $(\bcum[\es{\fta_t}{I}])_{I\in\I(\duallattice\times\{\pm\})}$ satisfy the
    symmetric cumulant hierarchy \eqref{eq:cumhieRbos}, namely,
    for $t\in\R$ there holds 
       \begin{align}
        \label{eq:boshie}
        &\partial_t
        \bcum[\es{a_t}{I}]
        =
        \sum_{(\ell,(x,\sigma))\in I}^{}
        \ci \sigma \sum_{y\in \Lambda}
        \Bigg(
        \alpha(x-y)
        \bcum[\es{a_t}{I^{\ell}_{x\to y}}]
        +
        \lambda
         V(x-y)
        \nonumber \\&
        \qquad \quad \times
    \sum_{\pi\in\p(I'+I\setminus \ell)}
        \prod_{A\in \pi}
        \ind(A\cap   I'\ne \emptyset)
        \bcum[\es{a_t}{P_\ell^{I',I}(A)}]
       \Biggr|_{I'= I^3(y;x,\sigma)}  \Bigg)
        \,,
    \end{align}
    where  $P_\ell^{I',I}$ is as in Remark \ref{rem:insert_permutation},
    $I^{\ell}_{x\to y} := I\setminus \{(\ell,(x,\sigma))\}\cup \{(\ell,(y,\sigma))\}$, and
    $I^3(y;x,\sigma)$ is defined in \eqref{eq:I3plus}--\eqref{eq:I3minus}.

    In addition,  for $t\in\R$ and any  $I\in\I(\duallattice\times\{\pm\})$ there holds
    \begin{align}
        \label{eq:bosfthie}
        &\partial_t\bcum[\es{\fta_t}{I}]=
        \ci\Omega(I)
        \bcum[\es{\fta_t}{I}]
+  \ci \lambda
        \sum_{(\ell,(k,\sigma))\in I}^{}
        \int_{(\duallattice)^3}
        \rmd k_{1}
        \rmd k_{2}
        \rmd k_{3}
        \delta_{L}(k-k_1-k_2-k_3)
        \sigma \FT{V}(k_2+k_{2+\sigma })
        \nonumber \\&
        \qquad \quad \times
       \sum_{\pi\in\p(I'+I\setminus \ell)}
        \prod_{A\in \pi}
        \ind(A\cap   I'\ne \emptyset)
        \bcum[\es{a_t}{P_\ell^{I',I}(A)}]
       \Biggr|_{I'= \hat I^3(k;\sigma)}
    \end{align}
   where $\Omega(I):=\sum_{(\ell,(k,\sigma))\in I}^{}
        \sigma \omega(k)$, and
   $\hat I^3(k;\sigma):=\{(1,(k_1,+)),(2,(k_2,\sigma)),(3,(k_3,-))\}$.
\end{theorem}

\begin{remark}
As in the fermionic case, one can simplify the hiearchy by replacing the polynomials by Wick polynomials in the definition of $\partial$, and then using Theorem \ref{thm:cumhie}.  The calculations and conclusions are very similar to those at the end of Section \ \ref {sec:fermions} and will not be reproduced here.
\end{remark}

\subsection{Truncation of hierarchies leading to Boltzmann--Nordheim equations}
\label{sec:derivation}

As one further application of the above results, we show here a connection between the above cumulant hierarchies and kinetic theory of these two lattice systems.  We aim at spatially homogeneous kinetic theory and recall that the two evolutions are expected to match only after a scaling limit in which the lattice size is sent to infinity and the coupling to zero.  More precisely, we focus on the case with $L\gg 1$ and $0<\lambda\ll 1$.

Derivation of kinetic equations for wave turbulence, a problem whose correlation structure is close to the one we find here, has already been accomplished for certain continuum nonlinear Schr\"odinger equations \cite{deng_full_2023,DengHani_fullrange_2023}.  The precise limit studied there amounts to choosing $L=\lambda^{-p}$ for certain $p>0$, constructing suitable families of initial data, and then inspecting the evolution of finite moments of the Fourier transformed field up to times $t=O(\lambda^{-2})$ and taking $\lambda\to \infty$.   For the usual nonlinearity, one then obtains a Boltzmann--Peierls type wave kinetic equation which has been used to study wave turbulence.  The conjecture for the case studied here is that a similar limit should result to a fermionic Boltzmann--Nordheim equation for fermions and to a bosonic
Boltzmann--Nordheim equation for bosons \cite{lukkarinen_not_2009}.

Let us therefore now mimic the assumptions used in the study of kinetic scaling limits, and rewrite the argument from Appendix A of 
\cite{lukkarinen2026finitelatticekineticequations} to derive a version of the above Boltzmann--Nordheim equations using fixed, finite $L,\lambda$.
We will study the Fourier transformed hierarchies for bosons and fermions simultaneously in the case
of translation and gauge invariant initial data.
Our goal is to derive the error term for the finite lattice kinetic equation studied in \cite{lukkarinen2026finitelatticekineticequations},
by showing that its solution is a good approximation to the covariance matrix of the system as long as the effect of certain 4th and 6th order cumulants can be ignored.  The target finite lattice
kinetic equation assumes that a fixed kinetic time $T_*>0$ is given,
corresponding to $t=\lambda^{-2} T_*$ in the original, microscopic time units,
and we  then aim to find a function $W\colon [0,T_*]\times\duallattice\to\R$ for which
\begin{align}
    \label{eq:flke}
      W_{T}(k_0)
      &=
      W_0(k_0)
      +
      \int_0^T
      \rmd t 
       \ccol(W_{t},T-t)(k_0)
       \,.
\end{align}
Here, the finite lattice collision operator is given by
\begin{align}
    \nonumber
    &
    \ccol(W_{t},T-t)(k_0)
   = 
   \pi \int_{(\duallattice)^3}
    \rmd k_1\rmd k_2\rmd k_3
    \delta_L(k_0+k_1-k_2-k_3)
    \delta_{\lambda,T-t}(\omega_0+\omega_1-\omega_2-\omega_3)
    \\
    \label{eq:ccol}
                         &\times
    \left( \FT{V}(k_1-k_2)+\theta\FT{V}(k_1-k_3) \right)^2
    \Big( \theta(h_1h_2h_3+h_0h_2h_3-h_0h_1h_3-h_0h_1h_2)+h_2h_3-h_0h_1 \Big)
    \,,
\end{align}
where $\theta=+$ for bosons and $\theta=-$ for fermions, and we have employed the shorthand notations
$\omega_i=\omega(k_i)$ and $h_i=W_{t}(k_i)$ and used the following approximation for a $\delta$-function,
\begin{align}
     \delta_{\lambda,\tau_0}(\Omega)
     =
     \frac{1}{2\pi}
     \int^{\lambda^{-2}\tau_0}_{-\lambda^{-2}\tau_0} \rmd s
    e^{-\ci s\Omega}
    \,, \qquad \lambda,\tau_0 >0\,, \ \Omega \in \R\,.
    \label{eq:energy_delta}
\end{align}

It has been shown in \cite{lukkarinen2026finitelatticekineticequations} that this problem does have well-behaved solutions, in the sense that given sufficiently small $T_*>0$,
there is a large class of initial data and potential functions so that for all large enough $L$ the solutions exist and remain bounded in the sup-norm.  If a limit of such solutions would also exist
when $L\to\infty$ followed by $\lambda\to 0$,  then this limit function could be expected to solve the corresponding limit equation.  Formally, we then have
$\ccol(W_{T},\tau)(k_0)\to \mathcal{C}(W_T)(k_0)$ for any $\tau>0$,
where $ \mathcal{C}(W)$ is the collision operator of the
fermionic Boltzmann--Nordheim equation, if $\theta=-$, and of the bosonic
Boltzmann--Nordheim equation, if $\theta=+$.  Hence, we could expect the solutions to the finite lattice equation to be well approximated by solutions to the appropriate Boltzmann--Nordheim equation.

To convert the assumptions into the current setup,
we assume that a tracial state has been given and that it is both
translation and gauge invariant, and assume that it has exponential moments for the bosonic case.  Let $\E$ denote the corresponding
map, as defined by \eqref{eq:state} for $A$ in the approriate algebra.

In the present context, translation invariance of the state means that
$\mean{ \tau_{x_0}^{*} A\tau_{x_0}}=\mean{A}$ for all bounded operators $A$
and any $x_0\in \lattice$, where $\tau_{x_0}$ is a unitary map
defined by setting
\[
 (\tau_{x_0} \psi)_n(x_1,\ldots,x_n) :=
 \psi_n(x_1-x_0,\ldots,x_n-x_0)\,, \qquad n\in \N\,,\ (x_1,\ldots,x_n)\in \Lambda^n \,.
\]
for $\psi \in \F_\pm$.  Note that then $\tau_{x_0}$ is a indeed a bounded, unitary operator on $\F_\pm$, with $\tau_{x_0}^{-1}=\tau_{x_0}^{*}=\tau_{-x_0}$.
Hence, the requirement makes sense.  Using \eqref{eq:Hndef}, it is also now obvious that $H$ commutes with each $\tau_{x_0}$ and, thus, also $[\rme^{-\ci t H},\tau_{x_0}]=0$.

Now for any $\Psi\in \F_\pm$, its components can be represent via an action of the appropriate creation operators on the vacuum vector $\Omega_0$, namely,
\[
 \Psi_n(x_1,\ldots,x_n) = \frac{1}{\sqrt{n!}}\innerprod{a^*(x_1)\cdots a^*(x_n)\Omega_0}{\Psi}\,, \qquad n\in \N\,,\ (x_1,\ldots,x_n)\in \Lambda^n \,.
\]
Using this, one can check that now
$\tau_{x_0}a(x-x_0,\sigma)\psi=a(x,\sigma)\tau_{x_0}\psi$ for $\psi\in \D$, $x\in \lattice$, $\sigma\in \{\pm\}$, and thus for any $I\in\I(\lattice\times\{\pm\})$.
\[
 \innerprod{\tau_{x_0}^* a^I\tau_{x_0}\psi}{\psi} = \innerprod{a^I\tau_{x_0}\psi}{\tau_{x_0}\psi}=\innerprod{\tau_{x_0}a^{\tau_{x_0}I}\psi}{\tau_{x_0}\psi}=
 \innerprod{a^{\tau_{x_0}I}\psi}{\psi}
\]
where $\tau_{x_0}I=\{(\ell,(x_\ell-x_0,\sigma_\ell)\ : \ (\ell,x_\ell,\sigma_\ell)\in I\}$.
Hence, by translation invariance and invariance of the time-evolution, we can now conclude that for all $t\in \R$ and $I\in\I(\lattice\times\{\pm\})$,
\[
 \E[a_t^I] = \E[a_t^{\tau_{x_0}I}]\,.
\]
Averaging over $x_0$ and taking Fourier transforms, we thus obtaining the following relations for translation invariant states
\begin{align}
    \label{eq:transinvmom}
\E[a_t^I] = \frac{1}{|\lattice|}\sum_{x_0\in \lattice}\E[a_t^{\tau_{x_0}I}]\,,\qquad
    \E[{\fta}_t^{\hat{I}}] =
 \delta_L\biggl(\sum_{\ell\in {\hat{I}}} k_\ell\biggr)
 \frac{1}{|\Lambda|}\E[{\fta}_t^{\hat{I}}]\,,
\end{align}
valid
for any $t\in \R$ and nonempty  $I\in\I(\lattice\times\{\pm\})$, $\hat I\in\I(\duallattice\times\{\pm\})$.
The shorthand notations used above will be clarified in the following Remark \ref{rem:elements_shorthand}.
\begin{remark}
    \label{rem:elements_shorthand}
    For given $\hat I\in\I(\duallattice\times\{\pm\})$, the shorthand notation
    $\ell\in\hat{I}$ means $\ell\in\pr_1(\hat{I})$ and given such $\ell$, we
    denote by $k_\ell$ and $\sigma_\ell$  the elements for which
    $(\ell,(k_\ell,\sigma_\ell))\in \hat{I}$.
    Similarly for given $I\in\I(\lattice\times\{\pm\})$,  the shorthand notation
    $\ell\in{I}$ means $\ell\in\pr_1({I})$ and  we denote by
    $x_\ell$ and $\sigma_\ell$  the elements for which
    $(\ell,(x_\ell,\sigma_\ell))\in {I}$.
\end{remark}

Gauge invariance here relates to $U(1)$-invariance over rotations of the one-particle space.  In the Fock space, the corresponding unitary operators can be  parametrized by  $\varphi\in \R$, and then setting
\[
 (\tau_{\varphi} \psi)_n(x_1,\ldots,x_n) :=
 \rme^{-\ci n \varphi}\psi_n(x_1,\ldots,x_n)\,, \qquad n\in \N\,,\ (x_1,\ldots,x_n)\in \Lambda^n \,.
\]
In other words, $\tau_{\varphi} = \rme^{-\ci \varphi N}$, and this is a unitary operator on $\F_\pm$, with $\tau_{\varphi}^* = \tau_{-\varphi}$.
It is straightforward to check that now, for any $\varphi\in \R$,
$[\rme^{-\ci t H},\tau_{\varphi}]=0$ and $a(x,\sigma)\tau_{\varphi}\psi =\rme^{\ci \sigma \varphi}\tau_{\varphi}a(x,\sigma)\psi$ for $\psi\in \D$, $x\in \lattice$, $\sigma\in \{\pm\}$.  Hence, for a gauge invariant state we have
\[
 \E[a_t^I] =
  \rme^{\ci \varphi\sum_{\ell\in I}\sigma_\ell}
  \E[a_t^{I}]
  \,,
\]
for all $\varphi, t\in \R$ and $I\in\I(\lattice\times\{\pm\})$.
Averaging over $\varphi$, we thus obtain
\begin{align}
    \label{eq:gaugeinvmom}
 \E[a_t^I] = \ind\biggl(\sum_{\ell\in I}\sigma_\ell=0\biggr)
 \E[a_t^I]\,,\qquad
 \E[{\fta}_t^{\hat I}] = \ind\biggl(\sum_{\ell\in {\hat I}}\sigma_\ell=0\biggr)
 \E[{\fta}_t^{\hat I}]\,,
\end{align}
valid
for any $t\in \R$ and nonempty $I\in\I(\lattice\times\{\pm\})$, $\hat I\in\I(\duallattice\times\{\pm\})$.

After these preliminaries, we now consider the antisymmetric cumulant hierarchy of $\fcum_t[I]:=\fcum[\es{\fta_t}{I}]$,
$I\in\I(\duallattice\times\{\pm\})$, for fermionic operators, and the symmetric cumulant hiearchy $\bcum_t[I]:=\bcum[\es{\fta_t}{I}]$, for bosonic operators.  Some of the computations will now be valid for either of the two cases, and then we drop the supscript and simply denote $\kappa_t[I]$ if the formula is true for both $\kappa_t=\fcum_t$ and for $\kappa_t=\bcum_t$.
For instance, now the gauge and translation invariance imply,
via \ref{eq:transinvmom} and \ref{eq:gaugeinvmom},
that
both cumulants satisfy
\begin{align}
    \label{eq:bothcuminv}
 \kappa_t[I] = \delta_L\!\left(\sum_{\ell\in I} k_\ell\right)
 \frac{1}{|\Lambda|}\kappa_t[I]\,,\qquad
 \kappa_t[I] =  \ind\biggl(\sum_{\ell\in I}\sigma_\ell=0\biggr)
 \kappa_t[I]\,,
\end{align}
for all $I=\in\I(\duallattice\times\{\pm\})$
with the shorthand notations clarified in Remark \ref{rem:elements_shorthand}.
This result can be concluded by induction using the recursion relations \eqref{eq:expressionK} and \eqref{eq:expressionK_fer}.
In
particular, all odd moments and cumulants vanish, including the expectation
values $\E[\fta_t(k,\sigma)]$ which correspond to the cases where $n=1$.
Therefore, $\E|_{\A_t}$ is now an even expectation map
also in the bosonic case.

Therefore, we only need to consider cumulants
 $I\in\I(\duallattice\times\{\pm\})$ with $|I|\ge 2$ even
 and for which $\sum_{\ell\in I}\sigma_\ell=0$ since the only other non-zero cumulant is the trivial one, with $I=\emptyset$.
Hence, Theorems \ref{thm:fermion_hie}  and \ref{thm:boson_hie}
now imply, for any $t\in\R$,
\begin{align}
        \partial_t\fcum_t[{I}] & =
        \ci \Omega(I)\fcum_t[I]
        +\ci
        \lambda
        \LL^{\mathrm{as}}_-(\fcum_t)[I]
        +\ci
        \lambda
        \B^{\mathrm{as}}_0(\fcum_t)[I]
        +\ci
        \lambda
        \B^{\mathrm{as}}_+(\fcum_t)[I]\, ,
        \nonumber \\
        \partial_t\bcum_t[{I}] & =
        \ci \Omega(I)\bcum_t[I]
        +\ci
        \lambda
        \LL^{\mathrm{s}}_-(\bcum_t)[I]
        +\ci
        \lambda
        \B^{\mathrm{s}}_0(\bcum_t)[I]
        +\ci
        \lambda
        \B^{\mathrm{s}}_+(\bcum_t)[I]\, ,
        \label{eq:evo_kappa}
\end{align}
where
both hierarchies contain $\Omega(I):=\sum_{\ell\in I}^{}
        \sigma_\ell \omega(k_\ell)$, and the only terms depending on higher order cumulants are given by
\begin{align}
    \label{eq:L-}
    \LL^{\mathrm{s}}_-(\kappa)[I]
    &=
     \sum_{\ell \in I}^{}
        \int_{(\duallattice)^3}\rmd k'_{123}\delta_{123}(k_\ell)
                \FT{V}_{123}(\sigma_\ell)
\kappa[I'_\ell +I\setminus \ell] \,, \\
        \LL^{\mathrm{as}}_-(\kappa)[I]
    &=
     \sum_{\ell \in I}^{} \sgn(\ell)
        \int_{(\duallattice)^3}\rmd k'_{123}\delta_{123}(k_\ell)
                \FT{V}_{123}(\sigma_\ell)
\kappa[I'_\ell +I\setminus \ell]
        \,,
\end{align}
where
   $I'_\ell := \{(1,(k'_1,+)),(2,(k'_2,\sigma_\ell)),(3,(k'_3,-))\}$,
and we have introduced the shorthand notations
$\rmd k'_{123}:=\rmd k'_{1}\rmd k'_{2}\rmd k'_{3}$,
$\delta_{123}(k):=\delta_{L}(k-k'_1-k'_2-k'_3)$ and $\FT{V}_{123}(\sigma)=
\sigma \FT{V}(k'_2+k'_{2+\sigma})$.  Here, we have used
Propositions \ref {prop:cum_per_invariance} and \ref{prop:cum_per_invariance_F} to reorder the cumulants into a common order where the new, integrated variables come first; note that $|I'_\ell + I\setminus \ell|\ge 4$.  In the antisymmetric case,
this yields the factor
$\sgn(\ell) :=(-1)^{\ell-1}$ assuming that $\pr_1(I) =\{1,2,\ldots,|I|\}$
is given in the standard form.

The remaining terms with $\B_0$ and $\B_+$ are nonlinear and
usually different for the two hierarchies.
The terms in $\B_+$ will only contain lower order cumulants, i.e., have only partitions $\pi$ such that all clusters $A\in \pi$ satisfy $|A|<|I|$.  The terms in $\B_0$ will contain all remaining partitions.  Since every cluster has  even $|A|\ge 2$, there are only three  allowed clusters $A\subset I'_\ell$ and all other clusters must contain elements from both $I'_\ell$ and $I\setminus \ell$.
Thus,
\begin{align}
    & \label{eq:B+}
            \B^{\mathrm{as}}_+(\kappa)[I]
    =
    \ind(|I|\ge 4)
      \sum_{\ell \in I}^{}
        \int_{(\duallattice)^3}\rmd k'_{123}\delta_{123}(k_\ell)
                \FT{V}_{123}(\sigma_\ell)
\nonumber \\ & \qquad \times
\sum_{\substack{\pi\in\p_2(I'_\ell+I\setminus \ell)\\|\pi|\ge2}}
        \sgn(\pi; P^{I',I}_\ell)
        \prod_{A\in \pi}
        \ind(|A|<|I|,\, A\cap   I'\ne \emptyset)\kappa[P^{I',I}_\ell(A)]\biggr|_{I'= I'_\ell}\,.
\end{align}
The definition of the operator $\B^{\mathrm{s}}_+(\kappa)[I]$ is obtaining by
replacing $ \sgn(\pi; P^{I'_\ell,I}_\ell)$ with $1$ in the above formula; note
that also in that case only  partitions with all  clusters even can yield a
nonzero contribution.
Finally, collecting the remaining partitions into $\B_0$, we obtain
\begin{align}    \label{eq:B0as}
    &     \B^{\mathrm{as}}_0(\kappa)[I]
    =
      \sum_{\ell \in I}^{}
        \int_{(\duallattice)^3}\rmd k'_{123}\delta_{123}(k_\ell)
                \FT{V}_{123}(\sigma_\ell)
  \nonumber \\ & \qquad \times
        \Bigg(
     \sum_{m=1}^3 (-1)^{m-1}
        \kappa[I'_\ell\setminus E]
        \kappa[{P_\ell^{E,I}(E+I\setminus\ell)}]\biggr|_{E\subset I'_\ell,\, \pr_1 E =\{m\}}
  \nonumber \\ & \qquad \quad
    +\ind(|I|\ge4)
    \sum_{\substack{\pi\in\p_2(I'_\ell+I\setminus \ell)\\|\pi| = 2}}
        \sgn(\pi; P^{I',I}_\ell)
  \nonumber \\ & \qquad \qquad \times
        \prod_{A\in \pi}
        \ind(|A|\in\{2,|I|\},\, A\not\subset I',\, A\not\subset (I\setminus \ell+3))\kappa[P^{I',I}_\ell(A)]\biggr|_{I'= I'_\ell}\Bigg)\,,\\
         \label{eq:B0s}
    &     \B^{\mathrm{s}}_0(\kappa)[I]
    =
      \sum_{\ell \in I}^{}
        \int_{(\duallattice)^3}\rmd k'_{123}\delta_{123}(k_\ell)
                \FT{V}_{123}(\sigma_\ell)
        \Bigg(
     \sum_{m=1}^3
        \kappa[I'_\ell\setminus E]
        \kappa[{P_\ell^{E,I}(E+I\setminus\ell)}]\biggr|_{E\subset I'_\ell,\, \pr_1 E =\{m\}}
  \nonumber \\ & \qquad 
  +\ind(|I|\ge4)
    \sum_{\substack{\pi\in\p_2(I'_\ell+I\setminus \ell)\\|\pi| = 2}}
        \prod_{A\in \pi}
        \ind(|A|\in\{2,|I|\},\, A\not\subset I',\, A\not\subset (I\setminus
    \ell+3))\kappa[P^{I',I}_\ell(A)]\biggr|_{I'= I'_\ell}\Bigg)\,.
\end{align}

In this case, we are
interested in controlling the covariance, the case with $|I|=2$.
We first define, for $t\in\R$, $k\in \duallattice$, the numbers
\[
 W_t(k) := \frac{1}{|\lattice|}\kappa_t[(-k,+),(k,-)]
 =  \frac{1}{|\lattice|} \E[\FT{a}_{t}(-k,+)\FT{a}_{t}(k,-)]
 = \sum_{y\in \lattice} \rme^{-\ci 2\pi k\cdot y} \E[a^*_{t}(0)a_{t}(y)]
 \, ,
\]
where in the last step we have used translation invariance to simplify the
result.  Taking a complex conjugate and using the properties of creation and
annihilation operators, we then conclude that $W_t(k)\in \R$.  By
\eqref{eq:bothcuminv}, it also holds
then
\begin{align}
    \label{eq:W}
    \kappa_t[(k',+),(k,-)] = \delta_L(k'+k)W_t(k)\,,
\end{align}
for all $k'\in \Lambda^*$.
Analogously, we define $\tilde{W}\colon \R\times\duallattice\to \R$ by
\begin{align}
    \label{eq:Wtilde}
    \tilde{W}_t(k): = \frac{1}{|\lattice|}\kappa_t[(k,-),(-k,+)]=
     \sum_{y\in \lattice} \rme^{-\ci 2\pi k\cdot y} \E[a_{t}(y)a^*_{t}(0)]\,.
\end{align}
Now, by the commutation relations, we always have
\begin{align}
    \tilde{W}_t(k) = \theta W_t(k)+1 \,.
\end{align}
where $\theta=\pm 1$, as defined earlier.

Take $k_0\in\duallattice$ and $t\in\R$. Let $I_{k_0}=\{(1,(-k_0,+)),\, (2,(k_0,-))\}$. Then
$W_t(k_0)=\kappa_t[I_{k_0}]/|\lattice|$. Since $\omega(-k)=\omega(k)$, we have that
$\Omega(I_{k_0}){\kappa_t}[I_{k_0}]=0$.  Also, $\B_+(\kappa_t)[I_{k_0}]=0$ as
$\ind(|I_{k_0}|\ge4)=0$.  Substitution, using translation invariance and symmetry of the potential yields
$\B_0(\kappa_t)[I_{k_0}] =0$. 
Therefore,
\begin{align}
    \label{eq:Wreal}
     W_ 
    {T\lambda^{-2}}
    (k_0) 
    &= W_0(k_0)
    +
    \ci \frac{\lambda}{|\Lambda|}\int_0^{T\lambda^{-2}}\rmd s\LL_-(\kappa_s)[I_{k_0}]\,.
\end{align}
This means that time evolution of $W$ is determined by 4th order cumulants. 

Let $I_{4}=\{(1,(k_a,+)),\,(2,(k_b,+)),\,(3,(k_c,-)),\,(4,(k_d,-))\}$, 
with $k_a,\dots,k_d\in\duallattice$ and $k_a+k_b+k_c+k_d=0 \bmod \duallattice$.
By permutation 
invariance (Propositions \ref{prop:cum_per_invariance} and \ref{prop:cum_per_invariance_F}), 
gauge and translation invariance, the nonzero 4th order cumulants
are determined by
sequences of the form $I_4$. 
 Applying Duhamel's formula to \eqref{eq:evo_kappa}
implies 
\begin{align}
    \label{eq:kappa4}
    \kappa_s[I_4] =  \ci
    \lambda
    \int_0^s \!\rmd r\,
    \rme^{\ci(s-r)
        \Omega_{abcd}
    }
    \B_+(\kappa_r)[I_4]
    +
    \ci 
    \lambda
    \mathcal{E}_4(s,\kappa)[I_4]\,,
\end{align}
where $\Omega_{abcd}=\omega(k_a)+\omega(k_b)-\omega(k_c)-\omega(k_d)$  and
\begin{align*}
    \mathcal{E}_4(s,\kappa)[I_4]
    &=
    \int_0^s\! \rmd r\,
    \rme^{\ci(s-r)
        \Omega_{abcd}
    }
    (\LL_-(\kappa_r)[I_4]+\B_0(\kappa_r)[I_4])
    + {\frac{1}{\ci \lambda}}
    \rme^{\ci s
        \Omega_{abcd}
    }
    \kappa_0[I_4]\,.
\end{align*}
Straightforward computation, together with gauge and translation invariance,
implies
\begin{align}
    \nonumber
    \frac{1}{|\Lambda|}
    \B_+(\kappa_r)[I_4] 
    & = 
    {W}_r(k_c)W_r(k_d)
    \Big[
        \tilde{W}_r(-k_b)
        \Big(
            \FT{V}(k_b+k_c)
            +
            \theta
            \FT{V}(k_b+k_d)
        \Big)
        \\
    \nonumber
    &
    \qquad
    \qquad
    \qquad
    \qquad
    \qquad
        +
        {W}_r(-k_a)
        \Big(
            \FT{V}(k_a+k_c)
            +
            \theta
            \FT{V}(k_a+k_d)
        \Big)
    \Big]
    \\
    \nonumber
    &
    \quad
    -
    {W}_r(-k_a)W_r(-k_b)
    \Big[
        \tilde{W}_r(k_c)
        \Big(
            \FT{V}(k_b+k_c)
            +
            \theta
            \FT{V}(k_a+k_c)
        \Big)
        \\
    \label{eq:B+computed}
    &
    \qquad
    \qquad
    \qquad
    \qquad
    \qquad
        +
        {W}_r(k_d)
        \Big(
            \FT{V}(k_b+k_d)
            +
            \theta
            \FT{V}(k_a+k_d)
        \Big)
    \Big]
    \,,
\end{align}
where the $\theta$ factors are coming from the factor $\sgn(\pi)$ in \eqref{eq:B+}

Take $k_0\in\duallattice$ and let $I_{0123+}=\{(1,(k_1,+)),\,(2,(k_2,+)),\,(3,(k_3,-)),\,(4,(k_0,-))\}$ and 
$I_{0123-}=\{(1,(-k_0,+)),\,(2,(k_1,+)),\,(3,(k_2,-)),\,(4,(k_3,-))\}$ for $k_1,k_2,k_3\in\duallattice$.
Now using \eqref{eq:kappa4} to \eqref{eq:Wreal} implies
\begin{align}
    \nonumber
    W_ 
    {T\lambda^{-2}}
    (k_0) 
    &
    = W_0(k_0)
    -
    \Bigg(
    \int_0^{T} \rmd t
    \int_0^
    {\lambda^{-2}(T-t)}
    \rmd s
        \int_{(\duallattice)^3}
        \rmd k_{123}
    \sum_{\sigma\in\{\pm\}}
     \rme^{-\ci s(\sigma
\omega(k_0)+\omega(k_3) -\omega(k_1)-\sigma\omega(k_2))
     }
        \\
        \label{eq:Wreal2}
    &\quad 
    \times
        \delta_{123}(-\sigma k_0)
        \FT{V}_{123}(\sigma) 
    \frac{1}{|\Lambda|}
     \B_+(\kappa_{T\lambda^{-2}})[I_{0123\sigma}]
    \Bigg)
    -
    \frac{\lambda^2}{|\Lambda|}
    \mathcal{E}_2(T\lambda^{-2},\kappa)[I_{k_0}]\,,
\end{align}
where we have used Fubini's theorem to change the order of $s$ and $t$ integrals, then made a change of variables 
$s\mapsto s-t$, followed by a change of variables $t= r\lambda^2$. Here,
\begin{align*}
    \mathcal{E}_2(T\lambda^{-2},\kappa)[I_{k_0}]
    &= 
     \int_0^
    {T\lambda^{-2}}
    \rmd s
    \sum_{\sigma\in\{\pm\}}
        \int_{(\duallattice)^3}
        \rmd k_{123}
        \delta_{123}(-\sigma k_0)
        \FT{V}_{123}(\sigma) 
    \mathcal{E}_4(s,\kappa)[I_{0123\sigma}]\,.
\end{align*}
We note that \eqref{eq:B+computed}
implies
\begin{align}
    \nonumber
        -\delta_{L}(k_0+k_1+k_2+k_3)
        &
     \B_+(\kappa_{t})[(-k_0,+),\,(-k_3,+),\,(-k_2,-),\,(-k_1,-)]
     \\
    \label{eq:symmetry}
        &
     =
        \delta_{L}(k_0+k_1+k_2+k_3)
     \B_+(\kappa_{t})[I_{0123+}]
\end{align}
Therefore, the second term on the right hand side of \eqref{eq:Wreal2} is equal to 
\begin{align}
    \nonumber
    &\int_0^T
    \rmd t
    \int_{-\lambda^{-2}(T-t)}^
    {\lambda^{-2}(T-t)}
    \rmd s
        \int_{(\duallattice)^3}
        \rmd k_{123}
     \rme^{-\ci s(
\omega(k_0)+\omega(k_3) -\omega(k_1)-\omega(k_2))
     }
        \\
        \label{eq:precol}
    &\quad 
    \times
        \delta_{L}(k_0+k_1+k_2+k_3)
        \FT{V}(k_2+k_3) 
    \frac{1}{|\Lambda|}
     \B_+(\kappa_{T\lambda^{-2}})[I_{0123+}]
    \,,
\end{align}
where we have merged ``$\sigma=+$'' and ``$\sigma=-$'' -terms by modifying only the
``$\sigma=-$'' term, as follows:  We make a change of variables
$(k_1,k_2,k_3)\mapsto (-k_3,-k_2,-k_1)$, then use \eqref{eq:symmetry}, and 
finally make a change of variables $s\mapsto -s$.  Symmetries of $\omega$ and
$\FT{V}$ have also been used to simplify the formula. Hence, after making a change of variables $(k_1,k_2,k_3)\mapsto (-k_3,-k_2,k_1)$ in \eqref{eq:precol},
we obtain from \eqref{eq:B+computed} and \eqref{eq:Wtilde}
 that the covariance $W$
defined in \eqref{eq:W} satisfies 
\begin{align}
    \label{eq:Wreal3}
    &
    W_ 
    {T\lambda^{-2}}
    (k_0) 
    = W_0(k_0)
    +
        \int_0^T
      \rmd t 
       \ccol(W_{T\lambda^{-2}},T-t)(k_0)
-
    \frac{\lambda^2}{|\Lambda|}
    \mathcal{E}_2(T\lambda^{-2},\kappa)[I_{k_0}]
    \,.
\end{align}
To arrive at the given form for $\ccol$, we use that $k_0-k_2=k_1-k_3$ or the integrand is zero due to the $\delta_L$-term,
and observe that it is possible to replace $\FT{V}(k_1-k_2)( \FT{V}(k_1-k_2)+\theta\FT{V}(k_1-k_3))$ by 
$\frac{1}{2}( \FT{V}(k_1-k_2)+\theta\FT{V}(k_1-k_3))^2$ 
in the integrand, due to its symmetry when swapping $k_2\leftrightarrow k_3$.

Hence, whenever the 4th and 6th order cumulants are such that
\begin{align}
    \label{eq:error}
\lambda^2|\Lambda|^{-1} |\mathcal{E}_2(T\lambda^{-2},\kappa)[I_{k_0}]|
\end{align}
is small, 
we obtain from \eqref{eq:Wreal3} that a good approximation for $W$ is given by the finite lattice kinetic equation
\eqref{eq:flke}. Note that \eqref{eq:error} is zero if the state is quasi-free, but in general quasi-freeness is not preserved in time evolution.
The remaining major open question is to actually prove that \eqref{eq:error} goes to zero 
when the lattice size $L$ goes to $\infty$ and the coupling $\lambda$ goes to $0$. 
Another problem we hope to solve in the near future is to check when the
solutions to the finite lattice kinetic equation \eqref{eq:flke} converge at
this limit to a solution of the Boltzmann--Nordheim equation.

%% file: bibliography.bib
@article{Crdenas2025,
  title = {Quantum {B}oltzmann Dynamics and Bosonized Particle-Hole Interactions in Fermion Gases},
  volume = {193},
  ISSN = {1572-9613},
  url = {http://dx.doi.org/10.1007/s10955-025-03557-z},
  DOI = {10.1007/s10955-025-03557-z},
  number = {1},
  journal = {Journal of Statistical Physics},
  publisher = {Springer Science and Business Media LLC},
  author = {Cárdenas,  Esteban and Chen,  Thomas},
  year = {2025},
  month = Dec 
}

@misc{lukkarinen2026finitelatticekineticequations,
      title={Finite lattice kinetic equations for bosons, fermions, and discrete {NLS}}, 
      author={Jani Lukkarinen and Sakari Pirnes and Aleksis Vuoksenmaa},
      year={2026},
      eprint={2601.10486},
      archivePrefix={arXiv},
      primaryClass={math-ph},
      url={https://arxiv.org/abs/2601.10486}, 
      note = {Preprint arXiv 2601.10486}
}

@article{lukkarinen_weakly_2011,
	title = {Weakly nonlinear {Schrödinger} equation with random initial data},
	volume = {183},
	language = {en},
	number = {1},
	journal = {Inventiones mathematicae},
	author = {Lukkarinen, Jani and Spohn, Herbert},
	year = {2011},
	pages = {79--188},
}

@article{lukkarinen_not_2009,
	title = {Not to {Normal} {Order}—{Notes} on the {Kinetic} {Limit} for {Weakly} {Interacting} {Quantum} {Fluids}},
	volume = {134},
	language = {en},
	number = {5-6},
	journal = {Journal of Statistical Physics},
	author = {Lukkarinen, Jani and Spohn, Herbert},
	year = {2009},
	pages = {1133--1172},
}

@misc{lukkarinen_generation_2025,
	title = {Generation of chaos in the cumulant hierarchy of the stochastic {Kac} model},
	publisher = {arXiv},
	author = {Lukkarinen, Jani and Vuoksenmaa, Aleksis},
	year = {2025},
	note = {Preprint arXiv:2407.17068 [math-ph]},
}

@article{deng_full_2023,
	title = {Full derivation of the wave kinetic equation},
	volume = {233},
	language = {en},
	number = {2},
	journal = {Inventiones mathematicae},
	author = {Deng, Yu and Hani, Zaher},
	year = {2023},
	pages = {543--724},
}

@article{benedetto_weak-coupling_2005,
	title = {{On} {the} {Weak}-{Coupling} {Limit} {for} {Bosons} {and} {Fermions}},
	volume = {15},
	issn = {0218-2025, 1793-6314},
	url = {https://www.worldscientific.com/doi/10.1142/S0218202505000984},
	doi = {10.1142/S0218202505000984},
	language = {en},
	number = {12},
	urldate = {2025-05-21},
	journal = {Mathematical Models and Methods in Applied Sciences},
	author = {Benedetto, D. and Pulvirenti, M. and Castella, F. and Esposito, R.},
	month = dec,
	year = {2005},
	pages = {1811--1843},
}

@article{chen_emergence_2023,
	title = {On the {Emergence} of {Quantum} {Boltzmann} {Fluctuation} {Dynamics} near a {Bose}–{Einstein} {Condensate}},
	volume = {190},
	language = {en},
	number = {4},
	journal = {Journal of Statistical Physics},
	author = {Chen, Thomas and Hott, Michael},
	year = {2023},
	pages = {85},
}

@article{kac_foundations_1956,
	title = {Foundations of {Kinetic} {Theory}},
	volume = {3},
	journal = {Proceedings of The third Berkeley symposium on mathematical statistics and probability},
	author = {Kac, Mark},
	year = {1956},
	pages = {171--197},
}

@article{paul_size_2019,
	title = {On the size of chaos in the mean field dynamics},
	volume = {231},
	issn = {0003-9527, 1432-0673},
	url = {http://arxiv.org/abs/1708.07701},
	doi = {10.1007/s00205-018-1280-y},
	number = {1},
	urldate = {2023-06-09},
	journal = {Archive for Rational Mechanics and Analysis},
	author = {Paul, Thierry and Pulvirenti, Mario and Simonella, Sergio},
	month = jan,
	year = {2019},
	note = {Preprint arXiv:1708.07701 [math-ph, physics:quant-ph]},
	pages = {285--317},
}

@article{Pulvirenti2016,
  title = {The {Boltzmann--Grad} limit of a hard sphere system: analysis of the correlation error},
  volume = {207},
  ISSN = {1432-1297},
  url = {http://dx.doi.org/10.1007/s00222-016-0682-4},
  DOI = {10.1007/s00222-016-0682-4},
  number = {3},
  journal = {Inventiones mathematicae},
  publisher = {Springer Science and Business Media LLC},
  author = {Pulvirenti,  M. and Simonella,  S.},
  year = {2016},
  month = Aug,
  pages = {1135–1237}
}

@article{Frst2013,
  title = {Derivation of a matrix-valued {Boltzmann} equation for the {Hubbard} model},
  volume = {46},
  ISSN = {1751-8121},
  url = {http://dx.doi.org/10.1088/1751-8113/46/48/485002},
  DOI = {10.1088/1751-8113/46/48/485002},
  number = {48},
  journal = {Journal of Physics A: Mathematical and Theoretical},
  publisher = {IOP Publishing},
  author = {F\"{u}rst,  Martin L R and Lukkarinen,  Jani and Mei,  Peng and Spohn,  Herbert},
  year = {2013},
  month = Nov,
  pages = {485002}
}

@misc{madsen2026asymptoticbehaviorkinetictime,
      title={On the asymptotic behavior at the kinetic time of a weakly interacting Fermi gas}, 
      author={Peter S. Madsen and Phan Thành Nam and Herbert Spohn and Minh-Binh Tran},
      year={2026},
      eprint={2605.13499},
      archivePrefix={arXiv},
      primaryClass={math-ph},
      url={https://arxiv.org/abs/2605.13499},
      note ={Preprint arXiv 2605.13499}
}

@article{Salmhofer2009,
  title = {Clustering of Fermionic Truncated Expectation Values Via Functional Integration},
  volume = {134},
  ISSN = {1572-9613},
  url = {http://dx.doi.org/10.1007/s10955-009-9698-0},
  DOI = {10.1007/s10955-009-9698-0},
  number = {5-6},
  journal = {Journal of Statistical Physics},
  publisher = {Springer Science and Business Media LLC},
  author = {Salmhofer,  Manfred},
  year = {2009},
  month = Mar,
  pages = {941–952}
}

@article{erdos_quantum_2004,
	title = {On the {Quantum} {Boltzmann} {Equation}},
	volume = {116},
	language = {en},
	number = {1-4},
	journal = {Journal of Statistical Physics},
	author = {Erdős, László and Salmhofer, Manfred and Yau, Horng-Tzer},
	year = {2004},
	pages = {367--380},
}

@inbook{Bodineau2023,
  title = {Dynamics of dilute gases: a statistical approach},
  ISBN = {9783985475605},
  url = {http://dx.doi.org/10.4171/ICM2022/130},
  DOI = {10.4171/icm2022/130},
  booktitle = {International Congress of Mathematicians},
  publisher = {EMS Press},
  author = {Bodineau,  Thierry and Gallagher,  Isabelle and Saint-Raymond,  Laure and Simonella,  Sergio},
  year = {2023},
  month = Dec,
  pages = {750–795}
}

@misc{deng_long_2025,
	title = {Long time derivation of the {Boltzmann} equation from hard sphere dynamics},
	publisher = {arXiv},
	author = {Deng, Yu and Hani, Zaher and Ma, Xiao},
	year = {2025},
	note = {Preprint arXiv:2408.07818 [math].  To appear in Annals of Mathematics, URL \url{https://annals.math.princeton.edu/articles/22284}},
}

@misc{DengHani_fullrange_2023,
  doi = {10.48550/ARXIV.2301.07063},
  url = {https://arxiv.org/abs/2301.07063},
  author = {Deng,  Yu and Hani,  Zaher},
  title = {Derivation of the wave kinetic equation: full range of scaling laws},
  note = {Preprint arXiv:2301.07063},
  publisher = {arXiv},
  year = {2023},
  copyright = {Creative Commons Attribution 4.0 International}
}

@incollection{lanford73,
	author = "O. E. Lanford",
	title = "Entropy and equilibrium states in classical statistical mechanics", 
        booktitle = "Statistical Mechanics and Mathematical Problems",
	editor = "A. Lenard",
        Publisher="Springer",
        Address="Berlin", 
        year=1973, pages="1--113"}

@article{deng_propagation_2024,
  title = {Propagation of chaos and higher order statistics in wave kinetic theory},
  volume = {28},
  ISSN = {1435-9863},
  url = {http://dx.doi.org/10.4171/JEMS/1488},
  DOI = {10.4171/jems/1488},
  number = {2},
  journal = {Journal of the European Mathematical Society},
  publisher = {European Mathematical Society - EMS - Publishing House GmbH},
  author = {Deng,  Yu and Hani,  Zaher},
  year = {2024},
  month = June,
  pages = {673–733}
}

@article{buckmaster_onset_2021,
	title = {Onset of the wave turbulence description of the longtime behavior of the nonlinear {Schrödinger} equation},
	volume = {225},
	issn = {0020-9910, 1432-1297},
	url = {https://link.springer.com/10.1007/s00222-021-01039-z},
	doi = {10.1007/s00222-021-01039-z},
	language = {en},
	number = {3},
	urldate = {2025-11-07},
	journal = {Inventiones mathematicae},
	author = {Buckmaster, T. and Germain, P. and Hani, Z. and Shatah, J.},
	month = sep,
	year = {2021},
	pages = {787--855},
}

@misc{deng_long_2024,
	title = {Long time justification of wave turbulence theory},
	url = {http://arxiv.org/abs/2311.10082},
	doi = {10.48550/arXiv.2311.10082},
	urldate = {2025-11-07},
	publisher = {arXiv},
	author = {Deng, Yu and Hani, Zaher},
	month = apr,
	year = {2024},
	note = {Preprint arXiv:2311.10082 [math]},
}

@article{lukkarinen_wick_2016,
	title = {Wick polynomials and time-evolution of cumulants},
	volume = {57},
	issn = {0022-2488, 1089-7658},
	url = {https://pubs.aip.org/jmp/article/57/8/083301/958884/Wick-polynomials-and-time-evolution-of-cumulants},
	doi = {10.1063/1.4960556},
	language = {en},
	number = {8},
	urldate = {2025-11-07},
	journal = {Journal of Mathematical Physics},
	author = {Lukkarinen, Jani and Marcozzi, Matteo},
	month = aug,
	year = {2016},
	pages = {083301},
}

@article{Dymov2021,
  title = {Formal Expansions in Stochastic Model for Wave Turbulence 1: Kinetic Limit},
  volume = {382},
  ISSN = {1432-0916},
  url = {http://dx.doi.org/10.1007/s00220-021-03955-w},
  DOI = {10.1007/s00220-021-03955-w},
  number = {2},
  journal = {Communications in Mathematical Physics},
  publisher = {Springer Science and Business Media LLC},
  author = {Dymov,  Andrey and Kuksin,  Sergei},
  year = {2021},
  month = Feb,
  pages = {951–1014}
}

@article{Dymov2022,
  title = {Formal Expansions in Stochastic Model for Wave Turbulence 2: Method of Diagram Decomposition},
  volume = {190},
  ISSN = {1572-9613},
  url = {http://dx.doi.org/10.1007/s10955-022-02998-0},
  DOI = {10.1007/s10955-022-02998-0},
  number = {1},
  journal = {Journal of Statistical Physics},
  publisher = {Springer Science and Business Media LLC},
  author = {Dymov,  Andrey and Kuksin,  Sergei},
  year = {2022},
  month = Oct 
}

@misc{dymov_kinetic_2025,
	title = {Kinetic approximation for equations of discrete turbulence in the subcritical case},
	url = {http://arxiv.org/abs/2505.16488},
	doi = {10.48550/arXiv.2505.16488},
	urldate = {2025-11-07},
	publisher = {arXiv},
	author = {Dymov, Andrey},
	month = jun,
	year = {2025},
	note = {Preprint arXiv:2505.16488 [math-ph]},
}

@incollection{cadamuro_kinetic_2018,
	address = {Cham},
	title = {Kinetic {Theory} and {Thermalization} of {Weakly} {Interacting} {Fermions}},
	volume = {270},
	isbn = {978-3-030-01601-2 978-3-030-01602-9},
	url = {http://link.springer.com/10.1007/978-3-030-01602-9_1},
	urldate = {2025-11-07},
	booktitle = {Macroscopic {Limits} of {Quantum} {Systems}},
	publisher = {Springer International Publishing},
	author = {Lukkarinen, Jani},
	editor = {Cadamuro, Daniela and Duell, Maximilian and Dybalski, Wojciech and Simonella, Sergio},
	year = {2018},
	doi = {10.1007/978-3-030-01602-9_1},
	note = {Series Title: Springer Proceedings in Mathematics \& Statistics},
	pages = {1--28},
}

@misc{Hannani_stoch_KdV_2022,
  doi = {10.48550/ARXIV.2210.17445},
  url = {https://arxiv.org/abs/2210.17445},
  author = {Hannani,  Amirali and Rosenzweig,  Matthew and Staffilani,  Gigliola and Tran,  Minh-Binh},
  title = {On the wave turbulence theory for a stochastic {KdV} type equation -- {Generalization} for the inhomogeneous kinetic limit},
  publisher = {arXiv},
  year = {2022},
  copyright = {arXiv.org perpetual,  non-exclusive license},
  note = {Preprint arXiv:2210.17445}
}

@book{brattelirobinson2,
    author = {Bratteli, Ola. and Robinson, Derek W.},
    address = {New York},
    isbn = {0-387-10381-3},
    language = {eng},
    publisher = {Springer-Verlag},
    title = {Operator algebras and quantum statistical mechanics. 2, Equilibrium states models in quantum statistical mechanics},
    year = {1981},
}

@article{Speicher1994,
    author = {Speicher, Roland},
    issn = {0025-5831},
    journal = {Mathematische annalen},
    language = {eng},
    number = {1},
    pages = {611-628},
    title = {Multiplicative functions on the lattice of non-crossing partitions and free convolution},
    volume = {298},
    year = {1994},
}

@book{NicaSpeicher2006,
    author = {Nica, Alexandru and Speicher, Roland},
    address = {Cambridge},
    booktitle = {Lectures on the combinatorics of free probability},
    isbn = {0521858526},
    language = {eng},
    publisher = {Cambridge University Press},
    series = {London Mathematical Society lecture note series ; 335},
    title = {Lectures on the combinatorics of free probability},
    year = {2006},
}
